\documentclass[fleqn,opre,nonblindrev]{arxiv}
\OneAndAHalfSpacedXI

\usepackage{endnotes}
\let\footnote=\endnote
\let\enotesize=\normalsize
\def\notesname{Endnotes}%
\def\makeenmark{$^{\theenmark}$}
\def\enoteformat{\rightskip0pt\leftskip0pt\parindent=1.75em
	\leavevmode\llap{\theenmark.\enskip}}

\usepackage{amsmath,amssymb,amsfonts}

\usepackage{natbib}
 \bibpunct[, ]{(}{)}{,}{a}{}{,}%
 \def\bibfont{\small}%
 \def\BIBand{and}%

\usepackage{rotating}
\usepackage{fancyvrb}

\TheoremsNumberedThrough     % Preferred (Theorem 1, Lemma 1, Theorem 2)
\ECRepeatTheorems
\JOURNAL{Operations Research}

\MANUSCRIPTNO{OPRE-2023-05-0252}

\usepackage{amssymb}
\usepackage{amsfonts}
\usepackage{dsfont}
\usepackage{bbm}
\usepackage{mathtools}
\usepackage{extarrows}

\usepackage{acronym}
\usepackage{booktabs}       % professional-quality tables
\usepackage{latexsym}
\usepackage{wasysym}
\usepackage{xspace}
\usepackage{comment}
\usepackage{longtable}
\usepackage{enumerate}

\usepackage{tikz}
\usetikzlibrary{calc,patterns}
\usepackage{setspace}

\usepackage{hyperref}
\hypersetup{
colorlinks=false,
linktocpage=true,
pdfstartview=FitH,
breaklinks=true,
pdfpagemode=UseNone,
pageanchor=true,
pdfpagemode=UseOutlines,
plainpages=false,
bookmarksnumbered,
bookmarksopen=false,
bookmarksopenlevel=1,
hypertexnames=true,
pdfhighlight=/O,
urlcolor=purple,linkcolor=MyBlue,citecolor=MyGreen, % <--- for screen
pdftitle={},
pdfauthor={},
pdfsubject={},
pdfkeywords={},
pdfcreator={pdfLaTeX},
pdfproducer={LaTeX with hyperref}
}
\usepackage{url}

\newcommand{\textpar}[1]{\textup(#1\textup)}

\newcommand{\argdot}{\,\cdot\,}
\newcommand{\eqdef}{\xlongequal{\!\!\text{def}\!\!}}

\DeclareMathOperator{\Bernoulli}{\mathsf{Bernoulli}}
\DeclareMathOperator{\Binomial}{\mathsf{Binomial}}
\DeclareMathOperator{\CoaPPoA}{\mathsf{PPoA}^{\mathsf{coa}}}
\DeclareMathOperator{\CorPPoA}{\mathsf{PPoA}^{\mathsf{cor}}}
\DeclareMathOperator{\Expect}{\mathsf{E}}
\DeclareMathOperator{\expo}{e}
\DeclareMathOperator{\MPPoA}{\mathsf{PPoA}^{\mathsf{mix}}}
\DeclareMathOperator{\OPoA}{\mathsf{OPoA}}
\DeclareMathOperator{\PoA}{\mathsf{PoA}}
\DeclareMathOperator{\Poisson}{\mathsf{Poisson}}
\DeclareMathOperator{\PPoA}{\mathsf{PPoA}}
\DeclareMathOperator{\Prob}{\mathsf{P}}

\DeclareRobustCommand{\stirling}{\genfrac\{\}{0pt}{}}

\DeclarePairedDelimiter{\braces}{\{}{\}}
\DeclarePairedDelimiter{\bracks}{[}{]}
\DeclarePairedDelimiter{\parens}{(}{)}

\DeclarePairedDelimiter{\abs}{\lvert}{\rvert}

\DeclarePairedDelimiter{\ceil}{\lceil}{\rceil}
\DeclarePairedDelimiter{\floor}{\lfloor}{\rfloor}

\DeclarePairedDelimiterX{\braket}[2]{\langle}{\rangle}{#1,#2}

\DeclarePairedDelimiterX{\inner}[2]{\langle}{\rangle}{#1,#2}
\DeclarePairedDelimiterX{\setdef}[2]{\{}{\}}{#1:#2}

\DeclarePairedDelimiterXPP{\probof}[1]{\Prob}{(}{)}{}{%

#1}

\DeclarePairedDelimiterXPP{\exof}[1]{\Expect}{[}{]}{}{%

#1}

\theoremstyle{EX}
\newtheorem{case}{Case}

\newcommand{\acdef}[1]{\emph{\acl{#1}} \textpar{\acs{#1}}\acused{#1}}

\newacro{ACG}{atomic congestion game}
\newacro{BCG}{Bernoulli congestion game}

\newacro{PoA}{price of anarchy}
\newacro{PoS}{price of stability}
\newacro{OPoA}{ordinary price of anarchy}
\newacro{PPoA}{prophet price of anarchy}
\newacro{SC}{social cost}
\newacro{ESC}{expected social cost}
\newacro{SO}{social optimum}
\newacro{SOC}{socially optimal cost}
\newacro{OSO}{ordinary social optimum}
\newacro{OSOC}{ordinary socially optimal cost}
\newacro{PSO}{prophet social optimum}
\newacro{PSOC}{prophet socially optimal cost}

\newacro{NE}{Nash equilibrium}
\newacroplural{NE}[NE]{Nash equilibria}
\newacro{BNE}{Bayes-Nash equilibrium}
\newacroplural{BNE}[BNE]{Bayes-Nash equilibria}
\newacro{PNE}{pure Nash equilibrium}
\newacroplural{PNE}[PNE]{pure Nash equilibria}

\newacro{WE}{Wardrop equilibrium}
\newacroplural{WE}[WE]{Wardrop equilibria}

\newacro{KKT}{Karush\textendash Kuhn\textendash Tucker}
\newacro{OD}[O/D]{origin-destination}
\newacro{BPR}{Bureau of Public Roads}

\begin{document}
%\setcounter{page}{1}
%\DOI{10.1287/xxxx.0000.0000}%

\RUNAUTHOR{Cominetti et al.}

\RUNTITLE{Ordinary and Prophet Planning under Uncertainty in Bernoulli Congestion Games}

\TITLE{Ordinary and Prophet Planning under Uncertainty in Bernoulli Congestion Games}

\ARTICLEAUTHORS{%
\AUTHOR{Roberto Cominetti}
%,\textsuperscript{a} Second Author,\textsuperscript{b} Third Author,\textsuperscript{c} Fourth Author,\textsuperscript{c}

\AFF{Facultad de Ingenier\'ia y Ciencias, Universidad Adolfo Ib\'a\~nez, Santiago, Chile, \EMAIL{roberto.cominetti@uai.cl}}
%\textsuperscript{b}School of Industrial Engineering, Good College, Collegeville, Maine 01234 \EMAIL{secauth@goodcoll.edu}; 
%\textsuperscript{c}Their Common Affiliation \EMAIL{thauth@anywhere.edu, fourauth@anywhere.edu}

%mirko.janc@informs.org
\AUTHOR{Marco Scarsini}

\AFF{Dipartimento di Economia e Finanza, Luiss University, Roma, Italy, 
\EMAIL{marco.scarsini@luiss.it}}

\AUTHOR{Marc Schr\"oder}
\AFF{School of Business and Economics, Maastricht University, Maastricht, The Netherlands, 
\EMAIL{m.schroder@maastrichtuniversity.nl}}

\AUTHOR{Nicolas E. Stier-Moses}
\AFF{Central Applied Science, Meta, Menlo Park, USA, 
\EMAIL{nicostier@yahoo.com}}
}

\ABSTRACT{
We consider an atomic congestion game in which each player $i$  participates in the game with an exogenous and known probability $p_{i}\in(0,1]$, independently of everybody else, or stays out and incurs no cost.
We compute the parameterized \acl{PoA} to characterize the impact of demand uncertainty on the efficiency of selfish behavior, considering two different notions of a social planner. 
A prophet planner knows the realization of the random participation in the game; 
the ordinary planner does not.
As a consequence, a prophet planner can compute an adaptive social optimum that selects different solutions depending on the players that turn out to be active, whereas an ordinary planner faces the same uncertainty as the players and can only minimize the expected social cost according to the player participation distribution. 
For both type of planners we obtain tight bounds for the \acl{PoA},   by solving 
suitable optimization problems parameterized by the maximum participation probability $q=\max_{i} p_{i}$. 
In the case of affine costs, we find an analytic expression for the corresponding bounds.
}


\MSCCLASS{Primary: 91A14; secondary: 91A06, 91A10, 91A43, 91B70}

\ORMSCLASS{Games/group decisions Noncooperative}

%\SUBJECTCLASS{????}

%\AREAOFREVIEW{Markets, Platforms, and Revenue Management}

%\HISTORY{to be done by journal}

\FUNDING{
Roberto Cominetti's research was supported by FONDECYT [ Grants 1171501 and ANID/PIA/ACT192094].
Marco Scarsini gratefully acknowledges the support and hospitality of FONDECYT 1130564 and N\'ucleo Milenio ``Informaci\'on y Coordinaci\'on en Redes.''
His research was supported by the Istituto Nazionale di Alta Matematica [Grant GNAMPA CUP\_E53C22001930001], the Ministero dell'Universit\`a e della Ricerca [Grants PRIN 2017  ALGADIMAR and PRIN 2022 2022EKNE5K], and the European Union–Next Generation EU, component M4C2, investment 1.1 [Grant PRIN PNRR P2022XT8C8]. 
The views and opinions expressed, however, are solely those of the authors and do not necessarily reflect those of the European Union or the European Commission. 
Neither the European Union nor the European Commission can be held responsible for them.}

\maketitle

\section{Introduction}
\label{se:intro}

Atomic congestion games, introduced by \citet{Ros:IJGT1973}, have been extensively studied as a prominent class of potential games and have been the starting point of numerous modeling efforts that capture interactions mediated through marketplaces and networks. 
They are motivated by real-life situations in which individuals make decisions with the goal of optimizing their cost, latency, power, or other relevant metrics, and outcomes arise from players' choice of the various resources. 
The epitome of these congestion models has been road traffic routing. 
In this example, the players of the game represent commuters who choose a route that minimizes their traveling time. 
Because one commuter's realized time depends on choices made by other commuters, their behavior has been typically modeled using equilibrium concepts in the corresponding congestion game. We refer to \citet{Rou:AGT2007} and references therein.

Equilibria may be suboptimal, as already observed by \citet{Pig:Macmillan1920}, and planners have had to manage that. Suboptimality means that, compared to the equilibrium, there may be a set of different routing choices that reduces congestion and generates a lower total travel time for all commuters collectively. 
Starting with the seminal work of \citet{KouPap:STACS1999}, there has been an abundance of work that contrasts the equilibrium view to what a central planner could achieve under similar traffic conditions. 
For this thought exercise, planners are supposed to have the power to route traffic to improve conditions, even if the resulting traffic pattern is not at equilibrium. 
Although this routing solution may not be implementable in practice, it provides a benchmark to judge the efficiency of equilibria.

One important aspect in the study of congestion games is uncertainty. 
See, e.g.,  \citet{AshMonTen:AAMAS2006}  and the literature review in Section~\ref{suse:comparison-work}.
In practice, commuters need to make routing decisions under incomplete information about the traffic conditions in roads, the amount of traffic, the presence of accidents, etc. Uncertainty not only challenges commuters but also traffic planners, although the tools and information available to either of them could be different. 
Traditionally, planners have used manually collected data (e.g., traffic counts) and surveys (e.g., travel census) to gather information that is subsequently used to fit models and project current and future traffic conditions. 
More recently, technology has enabled access to real-time information which introduces further asymmetries between the planner and commuters. High-tech firms and telecommunication companies can more easily pinpoint where commuters are at any moment through GPS signals in phones and cars. 
The abundance of location information can be used by government planners, or by traffic routing platforms such as Apple Maps, Google Maps, and Waze, to estimate the number of commuters on the road and the routes they chose. 
This leads to a central understanding of current traffic conditions at any point in time. Going back to how planners can benchmark traffic conditions, a platform that possesses detailed traffic information, and particularly how many commuters are on the road on a particular day -- instead of just on an average day -- could use that information to compute an optimal traffic pattern customized to that particular day.

Focusing on whether planners may possess real-time information or not, in this paper we lay out a framework that captures the difference between less and more powerful planners that have foresight in the traffic conditions. 
We associate players, which represent commuters, to a random type that represents that they are either present or not, and we consider two different social planners whose goal is to route commuters optimally. 
The first, referred to as \emph{ordinary planner}, only has access to the distribution of commuters. 
Using the terminology of incomplete information games, the ordinary planner assigns an action -- i.e.,  a route -- to every player, before observing the realized type vector, i.e.,   which  commuters are actually present in the network at the moment. 
The second social planner, referred to as \emph{prophet planner}, knows who is present in the network at the time of routing and plans sets of routes for each contingency. 
In other words, a prophet planner can adapt the routing pattern to the set of present commuters in the system. 
Because the planner's decision is the solution of an optimization problem, the additional information available to a prophet planner positively contributes to optimizing the system and therefore it is a tighter benchmark for an equilibrium than that of an ordinary planner.
The term \emph{prophet} is used here to indicate a planner who can anticipate the future and react optimally. 
Similar concepts and terminology are used in on-line algorithms and prophet inequalities.

\label{page-prophet}

In this paper, we specialize this framework to congestion games with uncertain demand. 
Although our analysis applies to the whole class of congestion games, we will often use the language of routing games to provide a more vivid representation of the model. 
Starting from a congestion game with atomic players, we assume that each player may be present in the game with a given probability, or not participate otherwise.   
Similar ideas were used in \citet{meir2012congestion} and \citet{AngFotLia:AGT2013}.
Connecting this idea to traffic congestion in cities, commuters who travel often know by experience how many other commuters are typically on the road. 
However, the actual number of commuters to be encountered is uncertain and is likely to vary around its typical value. 
Such variability implies that for choosing the optimal route, players must anticipate the consequences of the uncertainty.

We model this situation as a game of incomplete information. The key assumption of our work is that each player $i\in\{1,\ldots,n\}$ participates in the game with probability $p_{i}$, independently of other players, and otherwise stays out of the game. 
The population of players (prior to the random entry) and their probabilities $p_{i}$ are common knowledge, but the actual realization of the uncertainty is unknown to players at the time when they make their strategic choices in the game.
If all players happen to have the same probabilities $p_{i}= q$, the number of active players is a binomial random variable with parameters $n$ and $q$. 
However, we allow for heterogeneous probabilities across players, which, given our motivation, is a natural setting to consider. 
We call games of this form \emph{Bernoulli congestion games}.    

An instance of the game can be thought as the representation of the network conditions at a particular (small) time interval of the day, so that the participation of each player at that time interval is stochastic and thus only a random subset of the population of players is actually present. 
The participation probabilities at any given time interval of the day are affected by factors such as weather, road conditions, failures of traffic equipment such as traffic lights, or events that influence the traffic patterns such as shows or sports.
Whereas such external random factors provide coordination signals we assume that, conditional on the realizations of those factors, the individual participation events are  stochastically independent.

\subsection{Our Contribution}
\label{suse:our-contribution}

We study the general framework of Bernoulli congestion games and their equilibria.  
Our goal is to shed light on how the uncertainty in player's participation impacts the efficiency of the system
 as measured by the \ac{PoA}, defined as the worst-case ratio between the expected social cost of an equilibrium and the minimum expected cost achievable by a central planner
\citep[see, e.g., ][]{KouPap:STACS1999,KouPap:CSR2009}.
We distinguish two types of central planners: ordinary and prophet. 
The former faces the same uncertainty as the players  and assign strategies before knowing which players will be present:
the players who are present use the assigned strategies whereas the strategies of the absent players are discarded. 
In contrast, the prophet planner can adapt the strategies to be used to the specific set of players participating in the game. 
These scenarios yield two corresponding measures of inefficiency: the \acfi{OPoA} and the \acfi{PPoA}.

The main conceptual contribution of this paper is the distinction between ordinary and prophet \ac{PoA}. 
Although both cases have been considered earlier, by \citet{meir2012congestion} and \citet{CorHoeSch:TRB2019} for the \ac{OPoA}
and by \citet{Rou:ACMTEC2015} for the \ac{PPoA}, a direct comparison was not explicit. 
These previous studies considered more general stochastic player participation including correlated cases, for
which the worst case bounds occur in the deterministic case $p_{i}\equiv 1$. 
Here, we restrict to the independent case -- and hence uncorrelated -- with heterogeneous Bernoulli players, and perform 
a finer analysis that investigates how the inefficiency varies as a function of the maximal participation probability $q=\max_{i}p_{i}$.

Concretely, for any given class of cost functions $\mathcal{C}$ we define $\OPoA(\mathcal{C},q)$ and $\PPoA(\mathcal{C},q)$ respectively, as the maximum values for the ordinary and prophet \ac{PoA} across all instances with resource costs in $\mathcal{C}$ and $p_{i}\le q$. We parameterize our computations  in terms of the participation probabilities $p_{i}$ of the various players, and we show that, for both measures, the largest values across all instances with $p_{i}\leq q$ occur when $p_{i}\equiv q$.
Thus, the worst-case analysis for Bernoulli congestion games can be reduced to the case of homogeneous probabilities. 
Moreover, because both measures consider the worst-case ratio across all instances with $p_{i}\le q$, the \ac{OPoA} and \ac{PPoA} turn out to be nondecreasing in $q$. 
In particular they are maximal at  $q=1$, where they coincide with the  deterministic bounds in the  previous studies mentioned earlier.  
However,  Example~\ref{ex:monotonicity} and  Remark~\ref{re:monotonicity} show that \emph{for a fixed game}  the \ac{OPoA} and \ac{PPoA} may decrease with respect to some specific $p_{i}$ and even with respect to the maximal probability $q=\max_{i}p_{i}$.

Corollary~\ref{co:util2} shows that the homogeneous case can be further reduced  to a deterministic game with adjusted expected costs. This allows us to exploit the tools of $(\lambda,\mu)$-smoothness (see Section~\ref{sususe:comparison-work-1} for the background and details about this concept): Theorem~\ref{th:different-p-i} provides tight bounds for the \ac{OPoA} in general Bernoulli congestion games with nondecreasing costs and heterogeneous players. 

A similar analysis is undertaken for the \acl{PPoA}. To this end we refine the smoothness concept into $(\lambda,\mu,q)$-smoothness, which yields upper bounds for \ac{PPoA}. 
As in the ordinary setting, we prove in Theorem~\ref{th:lambda-mu-pr} that the optimal bounds derived from this refined smoothness framework are tight.
We note however that, in contrast with the standard smoothness framework, the games we use to show tightness are not routing games but general congestion games. 

It is important to highlight that our results in Sections~\ref{se:mod} and \ref{se:homogeneous-prob}, and in particular our tight bounds 
$\OPoA(\mathcal{C},q)$  and $\PPoA(\mathcal{C},q)$, are valid for general Bernoulli congestion games with nondecreasing costs and heterogeneous player's probabilities.

To illustrate the general results with a concrete class of cost functions, we perform a detailed study of a common framework considered in the literature: the class $\mathcal{C}_{\mathsf{aff}}$ of nondecreasing and nonnegative \emph{affine costs}. 
Theorem~\ref{th:upper} provides an analytic expression for the tight worst-case bounds $\OPoA(\mathcal{C}_{\mathsf{aff}},q)$ as a function of $q\in(0,1]$,  as illustrated by the green curve in Figure~\ref{fi:bounds}.
This bound is continuous and increasing, and exhibits three distinct regions with kinks at $\bar q_{0}=1/4$ and $\bar{q}_{1}\sim 0.3774$.
For $q=1$, we recover the $5/2$-bound 
of \citet{ChrKou:STOC2005} for deterministic atomic congestion games with affine costs, 
whereas for all $q<1$ we get a smaller bound.
A surprising feature here is that for $q\le 1/4$ we have a constant tight bound of $4/3$ -- which coincides with the \ac{PoA} for nonatomic games with affine costs, as shown by 
\citet{RouTar:JACM2002} -- independently of the structure of the congestion game and for any number of players. 
Similarly, Proposition~\ref{pr:lm-linear} and Theorem~\ref{th:Tight-PPoA-Affine}
provide the tight worst-case bound $\PPoA(\mathcal{C}_{\mathsf{aff}},q)$ for the prophet \acl{PoA},
expressed as the lower envelope of a countable family of functions as shown by the blue curve in Figure~\ref{fi:bounds}. This bound is continuous and increasing, and converges to $2$ for $q\to 0$ and to $5/2$ when $q\to 1$. 

\begin{figure}[ht]
\centering
\includegraphics[scale=0.75]{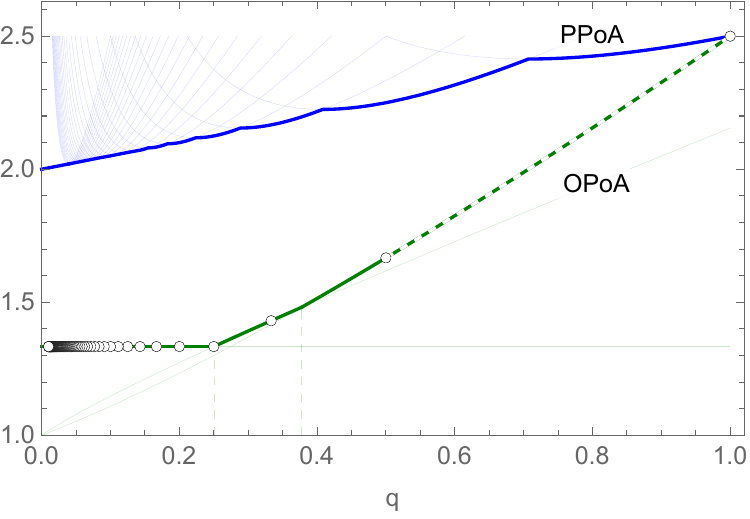}
\caption{\label{fi:bounds} Tight bounds for the \acs{OPoA} and \acs{PPoA} (thick lines in the figure) for Bernoulli congestion games with affine costs as a function of $q=\max_{i}p_{i}$. 
The dots corresponding to $q=1/n$ for $n\in\mathbb{N}_{+}$ in the \acs{OPoA} curve, as well as the dashed segment corresponding to $q\ge 1/2$, mathematically coincide with previously known bounds for different models (see Section~\ref{suse:comparison-work}). 
Our results fill in the gaps for that curve and make its regimes explicit. The vertical dashed lines depict the breakpoints where the \acs{OPoA} changes regimes. The dot at $1/3$ does not coincide with the kink of the curve at $\bar{q}_{1}\sim 0.3774$ but they coincide at $\bar{q}_{0}=1/4$.
}
\end{figure}

\subsection{Related Work}
\label{suse:comparison-work}
In this section, we frame our contributions in relation to the closest work in the literature. This provides the necessary context and understanding of our assumptions and results. 

\subsubsection{\Acl{OPoA} for nondecreasing costs and heterogeneous players.}
\label{sususe:comparison-work-1}
Our analysis of the \acl{OPoA} is most closely related to the work of \cite{ChrKou:STOC2005} who computed the \acl{PoA} for atomic unsplittable congestion games with affine costs by finding two coefficients for which an inequality for equilibria and optima holds. 
Related ideas appeared around the same time in \cite{HarVeg:P4CCAAN2007} and \cite{AlaDumGaiMonSch:SIAMJC2011} for a variety of settings. 
Collectively this set of techniques became known as the $(\lambda,\mu)$-smoothness framework, as coined and systematized by \citet{Rou:JACM2015} in a work that surveyed past uses and extended the framework to congestion games with general nondecreasing costs and other classes of games. The bounds obtained by this technique were shown to hold not only for pure equilibria, but also for mixed, correlated, and coarse-correlated equilibria.

As discussed previously, our \ac{OPoA} bounds result from an application of $(\lambda,\mu)$-smoothness. This is supported by Corollary~\ref{co:util2} and Theorem~\ref{th:different-p-i} which combined show that the worst-case for the \ac{OPoA} occurs with homogeneous probabilities and boils down to study a \emph{deterministic} game with expected costs. This allows to leverage the tools  of $(\lambda,\mu)$-smoothness to derive tight bounds
$\OPoA(\mathcal{C},q)$
for general Bernoulli congestion games with
nondecreasing costs and heterogeneous players.

Among previous work that studied specifically the \ac{OPoA} for congestion games with random players, \cite{meir2012congestion} considered a model where the stochastic player participation can be correlated, whereas resource costs can be  increasing or decreasing (for a similar model with both player and resource failures see \citealt{li2017congestion}). 
Most of their results concerned the case of Bernoulli games with homogeneous 
probabilities $p_{i}\equiv q$. 
They showed  how uncertainty eliminates
bad equilibria when  $q\approx 1$, and established 
 the lower semi-continuity of the \ac{OPoA} at $q=1$, and full continuity 
for routing games on parallel networks. 
For  fixed $q<1$ they also showed that \ac{PoA} can grow with the number of players.
Our results complement this by showing that the largest \ac{OPoA} across all congestion games with heterogeneous probabilities  $p_{i}\le q$
occurs in the homogeneous case $p_{i}\equiv q$, and that tight bounds can be established using $(\lambda,\mu)$-smoothness. Moreover, for affine costs we compute explicitly  the tight bounds 
$\OPoA(\mathcal{C}_{\mathsf{aff}},q)$ as a function of $q\in (0,1]$. 

\citet{CorHoeSch:TRB2019} also considered atomic games with stochastic player participation and arbitrary correlations, and proved that $(\lambda,\mu)$-smoothness bounds extend to Bayes-Nash equilibria of the incomplete information games. However, when applied to Bernoulli games 
for a fixed maximal participation probability $q\in(0,1)$, 
their bounds are not tight. 

\subsubsection{Prophet \ac{PoA} for nondecreasing costs and heterogeneous players.}
\label{sususe:comparison-work-2}
The closest previous result for the \acl{PPoA} is contained in \citet{Rou:ACMTEC2015} who showed that
for incomplete information games where player types are independent, the $(\lambda,\mu)$-smoothness bounds 
remain valid for the \ac{PPoA} considering all \aclp{BNE}. 
This yields bounds that are robust and insensitive with respect to the underlying distribution of types.
In particular, for congestion games with affine costs this yields the uniform bound $\PPoA\leq 5/2$ for the \acl{PPoA}. 

Here we introduce a weaker notion of $(\lambda,\mu,q)$-smoothness, specifically tailored to deal with Bernoulli games with 
$p_{i}\le q$, which provides finer parameterized bounds $\PPoA(\mathcal{C},q)$
that are sensitive to the maximum participation probability $q=\max_{i}p_{i}$. We show  that these bounds are tight (see Theorem~\ref{th:lambda-mu-pr}), and we compute them explicitly for the class of affine costs (see Theorem~\ref{th:Tight-PPoA-Affine}).

\subsubsection{Tight Bounds for Affine Cost Functions}
\label{sususe:comparison-work-3}
Our tight bounds for the \ac{OPoA} and \ac{PPoA} can be computed explicitly but laboriously for the simplest class of nonnegative and nondecreasing affine costs, which is one of the most common settings considered in the literature. 
For this class, Theorem~\ref{th:upper} shows that $\OPoA(\mathcal{C}_{\mathsf{aff}},q)$ as a function of $q\in(0,1]$ exhibits three distinct regions with kinks at $\bar q_{0}=1/4$ and at $\bar{q}_{1}\sim 0.3774$, the real root of $8 q^{3}+4q^{2}=1$ (see Figure~\ref{fi:bounds}).
In the lower region $0< q\leq\bar q_{0}$, $\OPoA(\mathcal{C}_{\mathsf{aff}},q)$ is constant and equal to $4/3$, which coincides with the \acl{PoA} for nonatomic congestion games with affine costs. 
In the middle region $\bar q_{0}\leq q\leq\bar q_{1}$, we have $\OPoA(\mathcal{C}_{\mathsf{aff}},q)=(1+q+\sqrt{q(2+p)})/(1-q+\sqrt{q(2+q)})$,
whereas in the upper region $\bar q_{1}\le q\leq 1$, the $\OPoA(\mathcal{C}_{\mathsf{aff}},q)=1+q+q^{2}/(1+q)$.
For $q=1$ we recover the known bound  of $5/2$ for deterministic atomic congestion games with affine costs, whereas for all $q<1$ we get a smaller bound. 
The computation of these tight bounds is a non-trivial application of  $(\lambda,\mu)$-smoothness, especially in the intermediate range $\bar q_{0}\leq q\leq\bar q_{1}$ which is the most challenging from a technical viewpoint.

Proposition~\ref{pr:lm-linear} and Theorem~\ref{th:Tight-PPoA-Affine} exploit the 
alternative $(\lambda,\mu,q)$-smoothness to compute the worst-case bounds $\PPoA(\mathcal{C}_{\mathsf{aff}},q)$ for the \acl{PPoA}.
The resulting bound is again tight and is given by the lower envelope of the functions $\{(\ell(\ell+1)pq^2+2\ell q+1)/(2\ell q)\mid \ell\in\mathbb{N}\setminus\{0\}\}$ (see Figure~\ref{fi:bounds}). 
The bound converges to $2$ when $q\to 0$ and to $5/2$ when $q\to 1$. 

\label{page-parametrization}
Some parts of our $\OPoA(\mathcal{C}_{\mathsf{aff}},q)$ bounds for affine costs coincide with previous bounds found by \citet{PilNikSha:TEAC2016}, \citet{BilMosVin:MOR2023} and  \citet{KleSch:TCS2019}, although for different models that -- somehow surprisingly -- turn out to have the same mathematical structure as ours. 
Specifically, \citet{PilNikSha:TEAC2016} considered an atomic congestion game where players using a given resource are randomly ordered and their costs depend on their position in this order.
For risk-neutral players, the model exhibits the same structure as ours with $p_{i}\equiv 1/2$.
\citet{BilMosVin:MOR2023}  considered a model with link failures where players select robust strategies that comprise a fixed number $\rho$ of 
edge-disjoint routes, and established tight bounds that coincide with ours when 
$p_{i}\equiv q=1/\rho$ with integer $\rho\in\{1,2,\ldots\}$.
In this model it only makes sense to consider discrete values of $q$'s of the form $1/\rho$, which provides little insight for what happens with continuous $q$,  specially in the range $1/4< q< 1/2$, where the function $\OPoA(\mathcal{C}_{\mathsf{aff}},q)$ has two kinks at $\bar q_{0}$ and $\bar{q}_{1}$. 
In a different direction, \citet{KleSch:TCS2019} studied routing games with affine costs and with two additional parameters $\rho$ and $\sigma$ that affect  the costs perceived by the players and central planner. 
Our bounds for affine costs and homogeneous players 
($p_{i}\equiv q$) are formally equivalent to the bounds in \citet{KleSch:TCS2019} with $\rho=\sigma= q$,
although the models and parameters have completely different meanings. 
Moreover, the results in \citet{KleSch:TCS2019}  only cover the interval $q\in[1/2,1]$, whereas we provide tight bounds in the full interval $q\in(0,1]$.
Incidentally, we note that the analytic expression 
of the bound for $q\in[1/2,1]$ remains valid and tight on the larger interval $q\in[\bar q_{1},1]$.     

Despite the fact that some parts of the \ac{OPoA} curve coincide with the expressions in 
\citet{PilNikSha:TEAC2016}, \citet{BilMosVin:MOR2023}, and \citet{KleSch:TCS2019},
we emphasize that the results are conceptually  different: our curve
represents the maximum \acl{PoA} across all Bernoulli games with $p_{i}\le q$, whereas
these previous papers consider neither Bernoulli games nor random participation of players, but rather games with homogeneous players whose costs and/or strategies depend on some uniform parameters which
formally end up playing a similar role as the maximum participation probability $q$.
However, \emph{a priori} it is far from obvious that those previous results
bear any connection with Bernoulli games, and it is the analysis that reveals these formal and partial coincidences.

\subsubsection{Other prior related work}
\label{sususe:comparison-work-4}

The inefficiency of equilibria in congestion games has been studied since the introduction of these games, and more
extensively since the late 1990's   after the work of \citet{KouPap:STACS1999,KouPap:CSR2009}, by means of worst-case bounds for the \acdef{PoA}.  These bounds differ substantially for atomic and nonatomic congestion games. 

For nonatomic games, where the demand is perfectly divisible, the equilibrium concept is due to \citet{War:PICE1952} and has been thoroughly studied starting with  \citet{BecMcGWin:Yale1956}.
Tight bounds for the \ac{PoA} in these games were obtained  for
specific classes of cost functions by \citet{RouTar:JACM2002,RouTar:GEB2004,Rou:JCSS2003,Rou:ACMSIAM2005} and \citet{CorSchSti:MOR2004,CorSchSti:GEB2008}.  

We refer to \citet{Rou:AGT2007,RouTar:AGT2007,CorSti:WEORMS2011} for surveys of these early results.
For atomic congestion games, both in its 
weighted and unweighted versions, the \ac{PoA} was examined in \citet{ChrKou:STOC2005}, \citet{DumGai:WINE2006}, \citet{HarVeg:P4CCAAN2007}, \citet{SurTotZho:A2007}, and \citet{AweAzaEps:SIAMJC2013}. \citet{AlaDumGaiMonSch:SIAMJC2011} provided exact bounds for the \ac{PoA} when costs are polynomial functions.
Inspired by these results and techniques, \citet{Rou:JACM2015} introduced the unifying terminology of $(\lambda,\mu)$-smoothness, and showed that the bounds derived in this manner are not only valid for pure equilibria, but also for mixed, correlated, and coarse-correlated equilibria. 
 
Various studies made different calls about what aspect to highlight, fixing some input parameters and taking the worst-case \ac{PoA} among chosen families of instances. Some results were given parametrically as a function of some scalar quantity, to shed light on how the efficiency of equilibria depends on this scalar. For instance,
\cite{CorSchSti:GEB2008} computed the parameterized \ac{PoA} as a function of the level of congestion in a game, in order to explain why lightly congested networks have low \ac{PoA}.
In a different direction, the impact of altruistic behavior of players in atomic congestion games with affine costs was investigated by \citet{CarKakKanKyrPap:TGC2010} for homogeneous players, and by \citet{ChedeKKemSch:TEAC2014} when each player has a different altruism coefficient. Although the latter deals with heterogeneous players, which bears some similarity with our model with heterogeneous probabilities, a major difference is that the 
social cost in these studies does not depend on the altruism parameters, whereas in our case both the ordinary and prophet optimal costs are affected by the stochastic player participation, the same as for the equilibrium.
As a consequence, the models in these papers  and the corresponding \ac{PoA} bounds differ substantially from our bound in Theorem~\ref{th:upper}.
Other recent papers studied the behavior of the \ac{PoA} in nonatomic routing games as a function of the total traffic demand. Among these, 
\citet{ColComMerSca:WINE2017}, \citet{ColComSca:TOCS2019}, and \citet{ColComMerSca:OR2020}
studied the asymptotic behavior of the \ac{PoA} in light and heavy traffic regimes, and showed that, under mild conditions, full efficiency is achieved in both limit cases. Similar results for congestion games in heavy traffic were obtained by \citet{WuMohChenXu:OR2021}.
A non-asymptotic analysis of the behavior of the \ac{PoA} as a function of the demand can be found in \citet{ComDosSca:MPB2024} and \citet{WuMoh:MOR2022}.
These papers studied the behavior of the \ac{PoA} for a given game as a function of the traffic demands.
By contrast, in the present paper we provide tight worst-case bounds for the \ac{PoA} for \acp{BCG}
as a function of the maximal participation probability $q\in (0,1]$.

Centrally to the motivation of this paper and as mentioned in the introduction,
attention has recently turned to incomplete information settings.
\citet{GaiMomTie:TCS2008} studied the \ac{PoA} for congestion games on a network with capacitated parallel edges, where players are of different types -- the type of each agent being the traffic that the agent moves -- and types are private information. \citet{AshMonTen:AAMAS2006} and \citet{UkkWal:TL2010} considered network games in which agents have incomplete information about the demand. \citet{OrdSti:TS2010}, \citet{NikSti:OR2014}, \citet{ComTor:MOR2016}, and \citet{LiaNikSti:MOR2019} studied the consequences of risk aversion on models with stochastic cost functions.
\citet{AngFotLia:AGT2013} studied a routing game over parallel links with Bernoulli players who are risk-averse and minimize the value-at-risk of the travel times, 
showing that for affine costs the \ac{PoA} is never larger than the number $n$ of players.
\citet{PenPolTen:GEB2009} and \citet{PenPolTen:DAM2011} dealt with congestion games with failures. 
\citet{WanDoaChe:TRB2014} considered nonatomic routing games with random demand and examined the behavior of the \ac{PoA} as a function of the demand distribution. 
\cite{Wrede2019} considered the same model as ours,  restricting the attention to games with a small number of players and giving a precise characterization of the ordinary price of anarchy for two players. 

\citet{AceMakMalOzd:OR2018,WuAmiOzd:OR2021} studied the impact of information on nonatomic congestion games.
\citet{MacScaTom:GEB2022,MacScaTom:GEB2023} considered learning in repeated nonatomic routing games where the costs functions are unknown and the demands are stochastic. \citet{GriHoeKliKog:ec2022} considered congestion games where a benevolent  planner (e.g., a mobility service such as Google Maps, Apple maps or Waze) has perfect information on the realization of an unknown state of nature, and can use this informational advantage to improve the efficiency of the equilibrium behavior by sending a public signal.
\citet{ZhuSav:IEEETCNS2022} studied a nonatomic congestion model where a planner can affect the agents' behavior via either public or private recommendations. 
These various papers on incomplete information games made different calls on the power of the social planner. 
For example, \citet{GaiMomTie:TCS2008}, \citet{WanDoaChe:TRB2014}, \citet{AngFotLia:AGT2013} and \citet{CorHoeSch:TRB2019} considered ordinary planners, whereas \citet{syrgkanis2015price} studied prophet planners.

\subsection{Organization of the Paper}
\label{se:organization}

The general setting of Bernoulli congestion games is formally described in Section~\ref{se:mod}.
In Section~\ref{se:homogeneous-prob} we present our tight worst-case bounds for the ordinary and \acl{PPoA}, for general nondecreasing cost functions and heterogeneous players.
Next, Section~\ref{se:poa_and_pos} computes explicitly these bounds for the class of affine costs.
Section~\ref{se:concl} contains conclusions and open problems. 
All missing proofs can be found in Section~EC.1 of the supplementary material. For convenience, Section~EC.5 in the supplementary material contains a list of symbols.

\section{Bernoulli Congestion Games and the Price of Anarchy}
\label{se:mod}

\subsection{Atomic Congestion Games}
Consider a finite set of resources $\mathcal{E}$
and a finite set of players $\mathcal{N}=\braces{1,\dots,n}$ where 
each player $i\in\mathcal{N}$ has a set of feasible strategies $\mathcal{S}_{i}\subseteq 2^{\mathcal{E}}$.
Given a strategy profile $\boldsymbol{s}\in\mathcal{S}\eqdef \times_{i\in\mathcal{N}} \mathcal{S}_{i}$, 
the cost for player $i$ is given by
\begin{equation}
\label{eq:Ci}
C_{i}(\boldsymbol{s})\eqdef\sum_{e\in s_{i}}c_{e}(n_{e}(\boldsymbol{s}))\, ,
\end{equation}
where $n_{e}(\boldsymbol{s})$ is the \emph{load} of resource $e\in\mathcal{E}$ defined as the number of players using that resource
\begin{equation}\label{eq:N-e}
n_{e}(\boldsymbol{s})\eqdef\sum_{j\in\mathcal{N}}\mathds{1}_{\{e\in s_{j}\}},
\end{equation}
and $c_{e}:\mathbb{N}\to[0,\infty)$ is a nondecreasing cost function of the resource $e$, with $c_{e}(x_{e})$ the cost experienced by each player using resource $e$ when the load is $x_{e}$. 
The tuple $\Gamma=\parens*{\mathcal{N},\mathcal{E},\mathcal{S},(c_{e})_{e\in\mathcal{E}}}$ defines an \acdef{ACG}. 

A \acdef{PNE} is a  strategy profile $\hat{\boldsymbol{s}}\in\mathcal{S}$ such that no player $i\in \mathcal{N}$ can benefit by unilaterally deviating from $\hat{s}_{i}$, that is, for each player $i\in\mathcal{N}$ and every $s_{i}\in\mathcal{S}_{i}$, we have
\begin{equation}
\label{eq:Nash}
C_{i}(\hat{\boldsymbol{s}})\leq C_{i}(s_{i},\hat{\boldsymbol{s}}_{-i}),
\end{equation}
where $\hat{\boldsymbol{s}}_{-i}$ is the strategy profile of all players except $i$.
The set of all \acp{PNE} is denoted by $\mathsf{NE}(\Gamma)$. 
\cite{Ros:IJGT1973} showed that every \acl{ACG} $\Gamma$ is an exact potential game and, as a consequence, its set $\mathsf{NE}(\Gamma)$ of pure Nash equilibria is nonempty. 

The \acdef{SC} of a strategy profile $\boldsymbol{s}\in\mathcal{S}$ is defined as the sum of all players' costs
\begin{equation}
\label{eq:SC}
C(\boldsymbol{s})\eqdef\sum_{i\in \mathcal{N}} C_{i}(\boldsymbol{s})=\sum_{e\in \mathcal{E}}n_{e}(\boldsymbol{s})\, c_{e}(n_{e}(\boldsymbol{s})),
\end{equation}
and  a \acdef{SO} is any profile $\boldsymbol{s}^{\ast}$ that minimizes this social cost 
\begin{equation}
\label{eq:optimum-social-cost}    
C(\boldsymbol{s}^{\ast})=\min_{\boldsymbol{s}\in\mathcal{S}}C(\boldsymbol{s}).
\end{equation}
The \acdef{PoA} is defined as
\begin{equation}
\label{eq:PoA-game}
\PoA(\Gamma)\eqdef\max_{\boldsymbol{s}\in\mathsf{NE}(\Gamma)}\frac{C(\boldsymbol{s})}{C(\boldsymbol{s}^{\ast})}.
\end{equation}
If $C(\boldsymbol{s}^{\ast})=0$, then $C(\boldsymbol{s})=0$ for all $\boldsymbol{s}\in\mathsf{NE}(\Gamma)$, and so in that case we artificially set $\PoA(\Gamma)=1$.

For a family $\mathcal{C}$ of  nonnegative and nondecreasing cost functions, we call $\mathcal{G}(\mathcal{C})$ the class of all \aclp{ACG}~$\Gamma$ with
costs $c_{e}\in\mathcal{C}$, and we look for 
 bounds on $\PoA(\Gamma)$ that hold uniformly 
for all such games, that is, we seek upper bounds for
\begin{equation}
\label{eq:PoA-class}
\PoA(\mathcal{C})\eqdef \sup\big\{\PoA(\Gamma):\Gamma\in\mathcal{G}(\mathcal{C})\big\}.
\end{equation}
A flexible tool to estimate $\PoA(\mathcal{C})$ is  the concept of smoothness: a family $\mathcal{C}$  is called \emph{$(\lambda,\mu)$-smooth} with $\lambda\geq 0$ and $\mu\in[0,1)$, if 
\begin{align}
&k \,c(1+m)  \leq  \lambda\, k\, c(k)  +  \mu\, m \, c(m) \nonumber\\
\label{eq:def-lambda-mu}
&\quad    \forall  k,m\in \mathbb{N}, \forall c(\argdot)\in\mathcal{C}.
\end{align}
It is well known that this condition implies $\PoA(\mathcal{C})\leq \lambda/(1-\mu)$, so that
\begin{align}
&\PoA(\mathcal{C})\leq\gamma(\mathcal{C})\eqdef\inf\left[\lambda/(1-\mu) \colon  \right.\nonumber\\
\label{eq:PoA_Bnd}
&\quad\left.(\lambda,\mu)\text{ satisfies \eqref{eq:def-lambda-mu}}\right].
\end{align}
For specific classes of costs, these estimates evolved in a series of papers by \citet{ChrKou:STOC2005,SurTotZho:A2007,AlaDumGaiMonSch:SIAMJC2011,AweAzaEps:SIAMJC2013,Rou:JACM2015}. 
\citet[theorem 5.8]{Rou:JACM2015} proved that these bounds are tight for \aclp{ACG}, that is:

\begin{theorem}
[\citealt{Rou:JACM2015}]
\label{th:Rou-lambda-mu}
For each family\, $\mathcal{C}$ of  nonnegative and nondecreasing cost functions we have $\PoA(\mathcal{C})=\gamma(\mathcal{C})$.
Moreover, if\, $\mathcal{C}$ contains the zero cost function $c_{0}(\argdot)$, then
the supremum in $\PoA(\mathcal{C})$
is achieved by network congestion games.
\end{theorem}

\citet{Rou:JACM2015} showed that the same bounds hold when the maximum in $\PoA(\Gamma)$ is taken over the class of mixed equilibria of the game $\Gamma$, and even over the larger classes of correlated equilibria and coarse correlated equilibria.
We will adapt the $(\lambda,\mu)$-smoothness framework to derive finer bounds for the \acl{PoA} in Bernoulli congestion games defined next. 
It turns out that Theorem~\ref{th:Rou-lambda-mu} is a corollary of our main result Theorem~\ref{th:lambda-mu-pr}.
For simplicity we start focusing on pure strategies, and postpone the discussion of mixed and correlated equilibria to Section~\ref{se:PoA_MixedEquilibria}.

\subsection{Bernoulli Congestion Games}

A \acdef{BCG} is an \acl{ACG} $\Gamma$ in which every player $i\in \mathcal{N}$ participates with some probability $p_{i}$ (independently of the other players), and otherwise remains inactive and incurs no cost.
A \ac{BCG} with probability vector $\boldsymbol{p}=\parens*{p_{i}}_{i\in \mathcal{N}}$ will be denoted by
 $\Gamma^{\boldsymbol{p}}$.
Clearly, in the deterministic case with $p_{i}\equiv 1$ for all $i\in\mathcal{N}$, a \ac{BCG} coincides with the \acl{ACG} $\Gamma$.

Without loss of generality, we assume that players (randomly) staying out of the game incur a cost equal to $0$.
Because staying out is determined exogenously and thus non-strategic, we could assign any cost to staying out of the game. 
The choice of cost for the outside option does not affect the set of Nash equilibria. 
Furthermore, a positive cost of non-participation would increase the cost of every strategy profile at equilibrium and under an optimal solution equally. 
This would make the \ac{PoA} artificially smaller, so in a worst-case analysis it is appropriate to make the constant equal to zero.

The random variables $W_{i}\sim \Bernoulli(p_{i})$, which indicate whether player $i$ is active, are assumed to be independent across players. 
We also assume that players choose their strategies before observing the actual realization of these random variables, so that no player knows for sure who will be present in the game.\footnote{In the context of routing games, \citet{NguPal:EJOR1988,Mil:N2001,MarNguSch:OR2004} considered a richer set of strategies called \emph{hyperpaths}
in which players are allowed to update their priors along their journey and switch to alternative routes.}
Thus, a \ac{BCG} can be framed as a game with incomplete information where each player has two possible types: active or inactive. 
The standard solution concept in this setting is the \acl{BNE}, where each player's strategy is contingent on the player's own type.
However, in \acp{BCG} an inactive player has no impact over the other players and has zero cost, regardless of the chosen strategy profile. 
Hence, it suffices to prescribe the strategies to be used when active. 
In what follows we describe explicitly  a  \acl{BNE} for \acp{BCG}.

For a strategy profile $\boldsymbol{s}=(s_{i})_{i\in\mathcal{N}}\in\mathcal{S}$ and all $e\in\mathcal{E}$ we define the random resource loads 
\begin{subequations}
\label{eq:N-e-stoch}
\begin{align}
N_{e}(\boldsymbol{s})&=\sum_{j\in\mathcal{N}} W_{j}\,\mathds{1}_{\{e \in s_{j}\}},
\intertext{and} 
N^{-i}_{e}(\boldsymbol{s})&=\sum_{j\neq i}W_{j}\,\mathds{1}_{\{e \in s_{j}\}}\quad\text{for all}~i\in\mathcal{N},
\end{align}
\end{subequations}
considering, respectively, either all the players who use the resource $e$, or all these players except player $i$.
With a slight abuse of notation, we use the same symbols as in \eqref{eq:Ci} and \eqref{eq:SC} for the expected cost of a player and the social cost, respectively, which are redefined  as

\begin{subequations}
\label{eq:eq:c-stoch}
\begin{align}
\label{eq:c-i-stoch}
C_{i}(\boldsymbol{s})
&=\Expect\left[W_{i}\sum_{e\in s_{i}}c_{e}(N_{e}(\boldsymbol{s}))\right]\nonumber\\
&=\sum_{e\in s_{i}}p_{i}\Expect\bracks*{c_{e}\parens*{1+N^{-i}_{e}(\boldsymbol{s})}},\\
\label{eq:c-soc-stoch}
C(\boldsymbol{s})
&=\sum_{i\in\mathcal{N}} C_{i}(\boldsymbol{s})
=\Expect\left[\sum_{e\in\mathcal{E}}N_e(\boldsymbol{s})\,c_{e}(N_{e}(\boldsymbol{s}))\right].
\end{align} 
\end{subequations}Alternatively, these costs might also be expressed in terms of the random set of active players  $\mathcal{I}=\braces*{j\in\mathcal{N} \colon W_{j}=1}$ and 
\begin{equation}\label{eq:load_alt}
n^{\mathcal{I}}_{e}(\boldsymbol{s})=\sum_{j\in\mathcal{I}}\mathds{1}_{\{e \in s_{j}\}},
\end{equation}
with which we have $N_{e}(\boldsymbol{s})=
n^{\mathcal{I}}_{e}(\boldsymbol{s})$ and $p(\mathcal{I})
\eqdef\prod_{j\in \mathcal{I}}p_{j}\,\prod_{j\not\in\mathcal{I}}(1-p_{j})$, so that
\begin{align}
C_{i}(\boldsymbol{s})
&=\Expect\left[W_{i}\sum_{e\in s_{i}}c_{e}\left(n^{\mathcal{I}}_{e}(\boldsymbol{s})\right)\right]\nonumber\\
\label{eq:cost-ie}
&=\sum_{e\in s_{i}}\sum_{\mathcal{I}\ni i}p(\mathcal{I})\, c_{e}\left(n^{\mathcal{I}}_{e}(\boldsymbol{s})\right),\\
\label{eq:SC-ie}
C(\boldsymbol{s})
&= \Expect\left[\sum_{e\in\mathcal{E}}n^{\mathcal{I}}_{e}(\boldsymbol{s})\cdot c_{e}\left(n^{\mathcal{I}}_{e}(\boldsymbol{s})\right)\right].
\end{align}

\begin{definition}
\label{de:BNE}
A strategy profile $\hat{\boldsymbol{s}}$ is a \acdef{BNE} for $\Gamma^{\boldsymbol{p}}$ if, for each $i\in\mathcal{N}$ and  all $s_{i}\in \mathcal{S}_{i}$, we have $C_{i}(\hat{\boldsymbol{s}})\leq  C_{i}(s_{i},\hat{\boldsymbol{s}}_{-i})$.
The set of all such \acp{BNE} is denoted by $\mathsf{NE}(\Gamma^{\boldsymbol{p}})$.
\end{definition}

Using \citet{Ros:IJGT1973}'s theorem, we prove that the set of \aclp{BNE} is nonempty by 
noting that every \ac{BCG} admits an exact potential function $\Phi \colon \mathcal{S}\to\mathbb{R}$
such that, for each strategy profile $\boldsymbol{s}\in\mathcal{S}$ and any unilateral deviation $\boldsymbol{s}'=\parens{s'_{i}, \boldsymbol{s}_{-i}}$ 
by a player $i$, we have 
\begin{equation}
\label{eq:potential_def}
C_{i}(\boldsymbol{s}')-C_{i}(\boldsymbol{s})=\Phi(\boldsymbol{s}')-\Phi(\boldsymbol{s}).
\end{equation}
 
A similar fact was observed in \citet[theorem~1]{meir2012congestion} for more general congestion games where the stochastic player participation can exhibit correlations, stating that such games admit a \emph{weighted potential} function and hence there exist pure \aclp{BNE}. 
As proved next, in the special case of Bernoulli games we have in fact an 
\emph{exact potential}.

\begin{proposition}
\label{pr:Meir-et-al}
Every \ac{BCG} $\Gamma^{\boldsymbol{p}}$ is an exact potential game.
In particular $\mathsf{NE}(\Gamma^{\boldsymbol{p}})$ is nonempty.
\end{proposition}

\proof{Proof.}
By considering the expectation of Rosenthal's potential
\begin{equation}\label{eq:potential}
\Phi(\boldsymbol{s}) \eqdef \Expect\bracks*{\sum_{e\in\mathcal{E}}\sum_{k=1}^{N_{e}(\boldsymbol{s})}c_{e}(k)},
\end{equation}
it suffices to note that \eqref{eq:potential_def} follows directly from the fact that the difference
\begin{equation}\label{eq:potential2}
\Phi(\boldsymbol{s})-C_{i}(\boldsymbol{s})=\Expect\bracks*{\sum_{e\in\mathcal{E}}\sum_{k=1}^{N^{-i}_{e}(\boldsymbol{s})}c_{e}(k)}
\end{equation}
does not depend on $s_{i}$.
\Halmos
\endproof

A particularly relevant case is that of homogeneous players with $p_{i}\equiv q$ for all $i\in\mathcal{N}$. 
In this case, the loads $N_{e}(\boldsymbol{s})$ are binomial random variables and $\Gamma^{\boldsymbol{p}}$ is equivalent to a deterministic atomic game with suitably modified costs.
Indeed, define the Binomial expectation of a cost function $c:\mathbb{N}\to\mathbb{R}_+$ as
\begin{align}
c^{q}(k)
\eqdef q\,\Expect\bracks*{c(1+X)} \quad \text{for }k\in\mathbb{N} \nonumber\\
\label{eq:cq}
\text{ and } X\sim\Binomial(k-1,q).
\end{align}
The following technical lemma will prove useful for dealing with 
the case $p_{i}\equiv q$ and, more generally, 
for heterogeneous probabilities such that $p_{i}\le q$.
\begin{lemma}
\label{le:CBH} 
Let $Y=\sum_{i=1}^mY_{i}$ 
and $X=\sum_{i=1}^m Y_{i}Z_{i}$
with $Y_{i}\sim\Bernoulli(r_{i})$
and $Z_{i}\sim\Bernoulli(q)$
two families of independent random variables.
Then
\begin{align*}
q\Expect\left[c(1+X)\right]
&=\Expect\left[c^{q}(1+Y)\right]
\intertext{and}
\Expect\left[X\,c(X)\right]
&=\Expect\left[Y\,c^{q}(Y)\right].   
\end{align*}

In particular, if  $X\sim\Binomial(m,q)$ then $q\Expect\bracks*{c(1+X)} = c^{q}(1+m)$ and 
$\Expect\bracks*{X\,c(X)} = m\,c^{q}(m)$.
\end{lemma}

\proof{Proof.}
Conditionally on the event $Y=k$,
the variable $X$ is distributed as a $\Binomial(k,q)$,
so that, from the very definition of $c^{q}(\argdot)$, we have
$q\Expect\big[c(1+X) \mid Y=k\big]=c^{q}(1+k)$. 
Then, the first claim follows from the tower law of iterated expectations:    
\begin{align*}
q\Expect\big[c(1+X)\big]
&=\Expect\big[q\Expect[c(1+X) \mid Y]\big]\\
&=\Expect\big[c^{q}(1+Y)\big].
\end{align*}
To establish the second identity we write
\begin{align*}
\Expect\big[X\,c(X)\big]
&=\Expect\Big[\sum_{i=1}^m Y_{i}Z_{i}\,c(X)\Big]\\
&=\sum_{i=1}^mr_{i}\,q\,
\Expect\Big[c(1+\sum_{j\neq i}Y_{j}Z_{j})\Big],    
\end{align*}
and using the first identity we conclude
\begin{align*}
\Expect\big[X\,c(X)\big]
&=\sum_{i=1}^mr_{i}\,c^{q} \Big(1+\sum_{j\neq i}Y_{j}\Big)\\
&=\sum_{i=1}^m
\Expect\Big[Y_{i}\, c^{q} \Big(\sum_{j=1}^m Y_{j}\Big)\Big]\\
&=\Expect\big[Y\,c^{q}(Y)\big].
\Halmos
\end{align*}
\endproof

Applying  Lemma~\ref{le:CBH} to \eqref{eq:c-i-stoch} and \eqref{eq:c-soc-stoch}, we get the following direct consequence. 

\begin{corollary}
\label{co:util2} 
Let $\Gamma^{\boldsymbol{p}}$ be a \acl{BCG}  with
$p_{i}\le q$. 
Set 
$r_{i}=p_{i}/q\in [0,1]$ and  
$N_{e}^{Y}(\boldsymbol{s})=\sum_{j\in\mathcal{N}}Y_{j}\mathds{1}_{\{e\in s_{j}\}}$
with 
$Y_{j}\sim \Bernoulli(r_{j})$
independent random variables.
Then
\begin{align}
\label{eq:disp12}
C_{i}(\boldsymbol{s})
&=\Expect\bracks*{\sum_{e\in s_{i}}c_{e}^{q}(N_{e}^{Y}(\boldsymbol{s}))},\\
\label{eq:disp22}
C(\boldsymbol{s})
&=\Expect\bracks*{\sum_{e\in\mathcal{E}}N_{e}^{Y}(\boldsymbol{s})\cdot c_{e}^{q}(N^{Y}_{e}(\boldsymbol{s}))},
\end{align}
so that $\Gamma^{\boldsymbol{p}}$
is equivalent -- in terms of player costs
and social cost -- to a \acl{BCG} $\Gamma_{ q}^{\boldsymbol{r}}$
with costs $c^{q}_{e}(\argdot)$
and  probabilities $r_{i}$.
In particular, when $p_{i}\equiv q$ we 
have $r_{i}\equiv 1$ and $N^{Y}_{e}(\boldsymbol{s})=n_{e}(\boldsymbol{s})$ so that
$\Gamma^{\boldsymbol{p}}$ is equivalent 
to a deterministic
\acl{ACG} $\Gamma_{ q}$ with costs $c_{e}^{q}(\argdot)$.
\end{corollary}

\section{Social Optimum and Price of Anarchy with General Costs}
\label{se:homogeneous-prob}

For games with incomplete information  we can define several notions of social optimum, depending on the relevant social cost function and the information available to the central planner. 
In the present context, we consider  the total \acdef{ESC} 
\begin{equation}\label{eq:ESC}
C(\boldsymbol{s})\eqdef \sum_{i\in\mathcal{N}} C_{i}(\boldsymbol{s})=\Expect\bracks*{\sum_{e\in\mathcal{E}}n^{\mathcal{I}}_{e}(\boldsymbol{s})\cdot c_{e}(n^{\mathcal{I}}_{e}(\boldsymbol{s}))}.
\end{equation}
An \acdef{OSO} is a profile $\boldsymbol{s}^{\ast}\in\mathcal{S}$ that minimizes the following expected cost:
\begin{equation}
\label{eq:ordinary-social-optimum}  
C_{\mathsf{ord}} \eqdef C(\boldsymbol{s}^{\ast})=\min_{\boldsymbol{s}\in\mathcal{S}}C(\boldsymbol{s}),
\end{equation}
which induces the \acdef{OPoA}, defined as 
\begin{equation}
\label{eq:PoAnp}
\OPoA(\Gamma^{\boldsymbol{p}})\eqdef\max_{\boldsymbol{s}\in\mathsf{NE}(\Gamma^{\boldsymbol{p}})}\frac{C(\boldsymbol{s})}{C_{\mathsf{ord}}}.
\end{equation}

The quantity $\OPoA(\Gamma^{\boldsymbol{p}})$ measures the inefficiency of the worst equilibrium $\hat{\boldsymbol{s}}\in\mathsf{NE}(\Gamma^{\boldsymbol{p}})$ by comparing its expected social cost $C(\hat{\boldsymbol{s}})$ to the optimum $C_{\mathsf{ord}}$ of a central planner who faces the same uncertainty about which players are present.
We now introduce the \acdef{PSO}, i.e.,  a harder benchmark that is achievable by a hypothetical planner who has full information and selects an optimal strategy adapted to each specific realization of the set of active players.
More precisely, to achieve a \ac{PSO}, the planner observes the realized set of active players $\mathcal{I}\subseteq\mathcal{N}$, and selects an optimal strategy profile $\boldsymbol{s}^\mathcal{I}=(s^\mathcal{I}_{i})_{i\in\mathcal{N}}\in\mathcal{S}$ that solves
\begin{equation}
\label{eq:optimal-prophet-action}
C_{\mathsf{pr}}(\mathcal{I})\eqdef C(\boldsymbol{s}^{\mathcal{I}})=\min_{\boldsymbol{s}\in\mathcal{S}}\sum_{e\in\mathcal{E}}n^{\mathcal{I}}_{e}(\boldsymbol{s})\cdot c_{e}\parens*{n^{\mathcal{I}}_{e}(\boldsymbol{s})}.
\end{equation}
This minimum is obviously smaller than the one obtained by the fixed ordinary optimal profile $\boldsymbol{s}^{\ast}$, so that taking expectation with respect to the random set $\mathcal{I}$ yields a smaller expected cost 
\begin{align}
C_{\mathsf{pr}}
&\eqdef\sum_{\mathcal{I}\subseteq\mathcal{N}}p(\mathcal{I}) C_{\mathsf{pr}}(\mathcal{I})\nonumber\\
\label{eq:PESC}
&=
\sum_{\mathcal{I}\subseteq\mathcal{N}}p(\mathcal{I})\sum_{e\in\mathcal{E}}n^{\mathcal{I}}_{e}(\boldsymbol{s}^{\mathcal{I}})\cdot c_{e}\parens*{n^{\mathcal{I}}_{e}(\boldsymbol{s}^{\mathcal{I}})}.
\end{align}
The \acdef{PPoA} is then defined as
\begin{equation}
\label{eq:PoAp}
\PPoA(\Gamma^{\boldsymbol{p}})\eqdef\max_{\boldsymbol{s}\in\mathsf{NE}(\Gamma^{\boldsymbol{p}})}\frac{C(\boldsymbol{s})}{C_{\mathsf{pr}}}.
\end{equation}

Observe that, by linearity of expectation, neither the ordinary nor the prophet planner can profit by optimizing over mixed strategies and in both cases it suffices to optimize over the pure strategy profiles $\boldsymbol{s}\in\mathcal{S}$. 
However, as far as equilibria are concerned, the maximum in \ac{PoA} could in principle be different if we considered either \emph{pure} or \emph{mixed} equilibria. 
For simplicity of exposition we focus on pure strategies, though in Section~\ref{se:PoA_MixedEquilibria} we show that our upper bounds on the \ac{PoA} are also valid for mixed equilibria.

\begin{example}
\label{ex:low2}
Consider a routing game with $2$ players on the Pigou network of Figure~\ref{fi:pigou} in which both players have the same origin $\mathsf{O}$ and destination $\mathsf{D}$, and the same $p_{i}= q$.
The cost function on the top link is linear, whereas the cost function in the bottom link is constant.
\begin{figure}[ht]
\centering
\begin{tikzpicture}[thick,scale=1.0, every node/.style={transform shape}]
 \node[draw, circle,minimum size = 0.7cm](Ob)  at (0,0) {\small $\mathsf{O}$};
 \node[draw, circle,minimum size = 0.7cm](Db)  at (6,0) {\small $\mathsf{D}$};
 \draw[->, >=latex,bend right=20] (Ob) to node[midway,fill=white] {$c_{2}(x_{2})=1+q$} (Db) ;
 \draw[->, >=latex,bend left=20] (Ob) to node[midway,fill=white] {$c_{1}(x_{1})=x_{1}$} (Db) ;
\end{tikzpicture}

\caption{\label{fi:pigou} A Pigou network. 
The edges are annotated with their cost functions $c_{e}(x_{e})$.}
\end{figure}
The profile where both players use the upper path is a \acl{BNE},
because in this case both paths have expected cost $q\,(1+q)$.
An \acl{OSO} $\boldsymbol{s}^{\ast}$ is achieved by pre-assigning one path to each player, independently of the fact that they are active or not. 
In this solution, when the player assigned to the upper path does not show up, the strategy misses the possibility of re-routing the other player on this  unused and cheaper path. By contrast, the \acl{PSO} exploits this flexibility: when both players are active, they are routed on different paths; if only one player shows up, this player is routed on the upper path.
This implies
\begin{align*}
\OPoA(\Gamma^{\boldsymbol{p}})
&=\frac{2+2q}{2+q}
\intertext{and}
\PPoA(\Gamma^{\boldsymbol{p}})
&=
\frac{2q(1+q)}{2q(1 - q) + q^2(2+q)}
=\frac{2+2q}{2+q^2}.
\end{align*}
\end{example}

\subsection{Tight Bounds for the \acl{OPoA}}
\label{suse:tight-bounds}

In what follows we derive bounds for the \ac{OPoA} that hold uniformly for all \acp{BCG} with probabilities $p_{i}$ below a fixed threshold $q$.

\begin{definition}
\label{de:Cp}
Given a family $\mathcal{C}$ of nonnegative and nondecreasing costs and $q\in(0,1]$, we call 
$\OPoA(\mathcal{C},q)$ the supremum
of the \acl{OPoA} $\OPoA(\Gamma^{\boldsymbol{p}})$ across all games 
$\Gamma^{\boldsymbol{p}}$ in the class $\mathcal{G}(\mathcal{C},q)$ of \aclp{BCG} with $p_{i}\le q$ for all $i\in\mathcal{N}$ and  $c_{e}\in\mathcal{C}$ for all $e\in\mathcal{E}$. 
We also let $\mathcal{C}^{q}$ denote the class of all functions $c^{q}$ defined by \eqref{eq:cq} with $c\in\mathcal{C}$.
\end{definition}

Our first main result, Theorem~\ref{th:different-p-i} below, shows that the worst-case instances for $\OPoA(\mathcal{C},q)$ occur in the case of homogeneous players with $p_{i}\equiv q$, whereas the inclusion $\mathcal{G}(\mathcal{C},p)\subseteq\mathcal{G}(\mathcal{C},q)$ for $p\le q$ implies that $\OPoA(\mathcal{C},q)$ is nondecreasing in $q$.  
It follows that the only relevant parameter to
characterize the worst-case \ac{OPoA} is the maximal probability $q=\max_{i}p_{i}$, from which we deduce that
$\OPoA(\mathcal{C},q)$ coincides with the bound $\gamma(\mathcal{C}^{q})$ for deterministic games with costs  in $\mathcal{C}^{q}$, obtained by minimizing the quotient $\lambda/(1-\mu)$ over all pairs $(\lambda,\mu)$ that satisfy
\begin{align}
  k \,c^{q}(1+m)  
  &\leq  \lambda\, k\, c^{q}(k)  +  \mu\, m \, c^{q}(m) \nonumber \\
\label{eq:r1}  
&\quad \forall  k,m\in \mathbb{N},~\forall c(\argdot)\in\mathcal{C}\,.
\end{align}
Note that, although we only consider the single parameter $q$, the bound $\OPoA(\Gamma^{\boldsymbol{p}})\leq\gamma(\mathcal{C}^{q})$ is valid for all congestion games $\Gamma^{\boldsymbol{p}}$ with heterogeneous probabilities $p_{i}\le q$ and costs in $\mathcal{C}$. 
The best previously known bound  was $\OPoA(\Gamma^{\boldsymbol{p}})\leq\gamma(\mathcal{C})$, which
can be derived from theorem~3.7 in \citet{Rou:ACMTEC2015} and is independent of $q$, whereas our sharp bound provides the tight worst-case estimate $\OPoA(\mathcal{C},q)=\gamma(\mathcal{C}^{q})$.

Observe that \eqref{eq:r1} coincides with \eqref{eq:def-lambda-mu} when $q=1$. 
In fact, \eqref{eq:def-lambda-mu} is stronger as it  implies  \eqref{eq:r1} for all $q$. 
Indeed, replacing   $k$ and $m$ in \eqref{eq:def-lambda-mu} with independent variables
$X\sim\Binomial(k,q)$ and $\widetilde{X}\sim\Binomial(m,q)$, then taking expectation, and using  Lemma~\ref{le:CBH}, we obtain \eqref{eq:r1}.
This confirms that our bound is tighter, namely
$\gamma(\mathcal{C}^{q})\leq\gamma(\mathcal{C})$.
As illustrated by the green curve in Figure~\ref{fi:bounds}, for affine costs, the bound $\gamma(\mathcal{C}_{\mathsf{aff}}^{q})$ is strictly smaller than $\gamma(\mathcal{C}_{\mathsf{aff}})=5/2$, except when $q=1$. 

We proceed to establish these previous claims, which are derived by suitably combining Theorem~\ref{th:Rou-lambda-mu}, Corollary~\ref{co:util2}, and theorem~5.3 of \citet{CorHoeSch:TRB2019}.

\begin{theorem}
\label{th:different-p-i}
For each family $\mathcal{C}$ of nonnegative and nondecreasing cost functions and each $q\in(0,1]$, we have $\OPoA(\mathcal{C},q) = \gamma(\mathcal{C}^{q})$. 
Moreover, if  zero costs are allowed, i.e.,  $0\in\mathcal{C}$, then the supremum of  $\OPoA(\Gamma^{\boldsymbol{p}})$ over $\Gamma^{\boldsymbol{p}}\in \mathcal{G}(\mathcal{C},q)$ is  achieved by restricting to network congestion games with homogeneous probabilities $p_{i}\equiv q$.
\end{theorem}

\proof{Proof.}
From Corollary~\ref{co:util2},
each $\Gamma^{\boldsymbol{p}}\in \mathcal{G}(\mathcal{C},q)$
is equivalent to a \acl{BCG} $\Gamma_{ q}^{\boldsymbol{r}}$
with costs $c^{q}_{e}(\argdot)$
and  probabilities $r_{i}=p_{i}/q$,
so that 
$\OPoA(\Gamma^{\boldsymbol{p}})=\OPoA(\Gamma_{ q}^{\boldsymbol{r}})$.
From \citet[theorem~5.3]{CorHoeSch:TRB2019}, any $(\lambda,\mu)$-smoothness bound for the deterministic game 
$\Gamma_{ q}$ with costs $c_{e}^{q}(\argdot)$ remains valid for the Bernoulli game $\Gamma_{ q}^{\boldsymbol{r}}$, hence $\OPoA(\Gamma^{\boldsymbol{p}})=\OPoA(\Gamma_{ q}^{\boldsymbol{r}})\leq\gamma(\mathcal{C}^{q})$.
Taking supremum over $\Gamma^{\boldsymbol{p}}$ we get $\OPoA(\mathcal{C},q)\leq\gamma(\mathcal{C}^{q})$.

Conversely, each deterministic congestion game $\Gamma_{ q}\in\mathcal{G}(\mathcal{C}^{q})$ is equivalent to a \acl{BCG} $\Gamma^{\boldsymbol{p}}$ with $p_{i}\equiv q$ so that $\PoA(\Gamma_{ q})=\OPoA(\Gamma^{\boldsymbol{p}})\leq \OPoA(\mathcal{C},q)$. 
Taking the supremum over  $\Gamma_{ q}\in\mathcal{G}(\mathcal{C}^{q})$ and using Theorem~\ref{th:Rou-lambda-mu} we conclude $\gamma(\mathcal{C}^{q})\leq \OPoA(\mathcal{C},q)$. 
This also shows that the worst case for $\Gamma^{\boldsymbol{p}}\in\mathcal{G}(\mathcal{C},q)$ occurs with $p_{i}\equiv q$, and, from \citet[Section 5.5]{Rou:JACM2015}, we have that such worst case can be realized with network congestion games. 
\halmos
\endproof

\begin{corollary}
\label{co:PPoA_monotonicity}
For each family $\mathcal{C}$ of nonnegative and nondecreasing cost functions, the map $q \mapsto \gamma(\mathcal{C}^{q})$ is nondecreasing.
\end{corollary}

The usefulness of the previous result depends on our ability to estimate
$\gamma(\mathcal{C}^{q})$, which is not easy in general. 
For the class $\mathcal{C}_{\mathsf{aff}}$ of affine costs we have $\mathcal{C}_{\mathsf{aff}}^{q}\subseteq\mathcal{C}_{\mathsf{aff}}$
so that, combining \citet{ChrKou:STOC2005} and \citet{Rou:JACM2015}, we get
$\gamma(\mathcal{C}_{\mathsf{aff}}^{q})\leq \gamma(\mathcal{C}_{\mathsf{aff}})=5/2$.
In Section~\ref{suse:tight-upper-bounds} we  explicitly compute $\gamma(\mathcal{C}_{\mathsf{aff}}^{q})$ and show that it is  equal to $4/3$ for all $q\in[0,1/4]$ and then becomes strictly increasing and reaches the upper bound $\gamma(\mathcal{C}_{\mathsf{aff}})=5/2$ at $q=1$.

Theorem~\ref{th:different-p-i} shows that the worst case for \ac{OPoA} occurs when all the $p_{i}$'s are equal,
so that the only relevant parameter is the maximal probability $q=\max_{i}p_{i}$. 
The following example illustrates why other parameters such as the average  or the minimum participation probabilities play no role in characterizing the worst-case \ac{OPoA}. 

\begin{example}
\label{ex:average}
Consider a game $\Gamma$ with $p_{i}=1$ and $\OPoA(\Gamma)=5/2$. See, for example, \citet[theorem~2]{ChrKou:STOC2005} or \citet[theorem~10]{AweAzaEps:SIAMJC2013}. 
Consider now a game $\Gamma'$ having additional dummy players who have low participation probabilities and  can use only one resource with zero cost. 
Then  $\OPoA(\Gamma')$ remains equal to $5/2$ although the minimum and average probabilities have changed. 
\end{example}

Given that the worst case bound
 $\OPoA(\mathcal{C},q)$ is nondecreasing in $q$,
 a natural question is whether this also holds for
 a fixed game $\Gamma$ so that $\OPoA(\Gamma)$ would increase with the
 maximal probability $q$, or even with respect to each $p_{i}$ separately.
The example below shows that both properties may fail. 
This does not exclude the possibility that some form of component-wise monotonicity might hold when we consider the worst-case \ac{OPoA} with a fixed number of players $n$ and variable probabilities $(p_{i})_{i=1}^{n}$.  

\begin{example}
\label{ex:monotonicity}
Consider any game $\Gamma$ with $p_{i}\le q<1$ and $\OPoA(\Gamma)=C(\hat{\boldsymbol{s}})/C(\boldsymbol{s}^{\ast})>1$ (see, for instance,  Example~\ref{ex:low2}). 
Consider now a new instance $\Gamma'$ with one additional player $n+1$, whose action set consists of a single resource $e$ that has cost $c_{e}(x)=x$ and is not part of any other player's action set. 
Then $\OPoA(\Gamma')= (C(\boldsymbol{s})+p_{n+1})/(C(\boldsymbol{s}^{\ast})+p_{n+1})$, which is decreasing in $p_{n+1}$. 
So, in particular, if $p_{n+1}\geq q$, then the $\OPoA(\Gamma')$ decreases with the maximal probability.
\end{example}

\subsection{Tight Bounds for the \acl{PPoA}}

We proceed to derive tight bounds for the \ac{PPoA} that hold uniformly for all \acp{BCG} with probabilities $p_{i}$ below a fixed threshold $q$.
Naturally, we expect the bounds to be larger than the tight bound $\gamma(\mathcal{C}^{q})$ for the  \ac{OPoA}. 
To this end, we introduce a slight modification of the smoothness concept, which we call
\emph{$(\lambda,\mu,q)$-smoothness}. For
$\lambda\geq 0$, $\mu\in[0,1)$ and $q\in(0,1]$, we consider the inequality
\begin{align}
&\frac{1}{q}\,k\,c^{q}(1+m)\leq \lambda\,k\,c(k)+\mu\,m\,c^{q}(m)\nonumber\\
\label{eq:ccp}
&\quad\forall  k,m\in \mathbb{N},~\forall c(\argdot)\in\mathcal{C},
\end{align}
and we set
\begin{equation}
\label{eq:BndPr}
\gamma_{\mathsf{pr}}(\mathcal{C},q)\eqdef\inf\braces*{\lambda/(1-\mu)
\colon (\lambda,\mu) \text{ satisfies }\eqref{eq:ccp}}.
\end{equation}
Because \eqref{eq:ccp} holds trivially for the zero function $c_{0}(\argdot)\equiv 0$ and when $k=0$, we can restrict this condition to the set
\begin{equation}
\label{eq:triples}
\mathcal{T} = \braces{(c,k,m) \colon  c\in\mathcal{C}\setminus\{c_{0}\}, k \in \mathbb{N}_{+},m\in\mathbb{N} }.
\end{equation}

We will show that $\gamma_{\mathsf{pr}}(\mathcal{C},q)$ yields a tight bound for the worst-case \acl{PPoA}.
Note that, whereas the expected cost $c^{q}(\argdot)$ appears on both sides of \eqref{eq:ccp}, the term $\lambda\,k\,c(k)$ on the right involves the original cost, so that \eqref{eq:ccp} is half way between \eqref{eq:def-lambda-mu} and \eqref{eq:r1}. 
Indeed, using  Lemma~\ref{le:CBH},  by  replacing $m$ in \eqref{eq:def-lambda-mu} with $X\sim\Binomial(m,q)$ and taking expectation, it follows that \eqref{eq:def-lambda-mu} implies \eqref{eq:ccp}; moreover, replacing $k$ in \eqref{eq:ccp} with $\widetilde{X}\sim\Binomial(k,q)$ and taking expectation, we get that \eqref{eq:ccp} implies \eqref{eq:r1}.
These implications translate into the following order for these bounds.
\begin{lemma} 
\label{le:implications-betas}
For each family $\mathcal{C}$ of nonnegative and nondecreasing cost functions and each $q\in(0,1]$ we have $\gamma(\mathcal{C}^{q})\leq \gamma_{\mathsf{pr}}(\mathcal{C},q)\leq \gamma(\mathcal{C})$.
\end{lemma}

\proof{Proof.}
The result follows directly from the implications \eqref{eq:def-lambda-mu} $\implies$ \eqref{eq:ccp} $\implies$ \eqref{eq:r1}.
\Halmos
\endproof

The inequalities in Lemma~\ref{le:implications-betas} can be strict.
This  will be illustrated in Section~\ref{se:poa_and_pos}, where we prove that
for the class $\mathcal{C}_{\mathsf{aff}}$ of affine cost functions, $\gamma_{\mathsf{pr}}(\mathcal{C}_{\mathsf{aff}},q)$ is a tight bound for the \ac{PPoA} with
$\gamma(\mathcal{C}_{\mathsf{aff}}^{q})<\gamma_{\mathsf{pr}}(\mathcal{C}_{\mathsf{aff}},q)<\gamma(\mathcal{C}_{\mathsf{aff}}) = 5/2$ for all $q\in (0,1)$.

Our estimates for the \ac{PPoA} exploit a special type of mixed strategies $\sigma_{i}^{\ast}$ where each player $i$ mimics the strategy of the prophet by sampling the other potentially active players. 

\begin{definition}
\label{de:prophet-like}
A \emph{prophet-like} strategy for player $i$ is a mixed strategy $\sigma_{i}^{\ast}$ that chooses a prophet optimal strategy $s^{\mathcal{I}}_{i}$ for a randomly chosen subset $\mathcal{I}$ of players that includes $i$ with certainty, together with a sample of the other players where each  $j\neq i$ is included with probability $p_{j}$.
\end{definition}

Note that in these prophet-like strategies each player samples a personal random set $\mathcal{I}$, independently of the other players. These samples need not coincide with the actual realization of the Bernoulli random variables $W_{i}$ that determine who actually takes part in the game. 

Using these special mixed strategies $\sigma_{i}^{\ast}$, we can prove an upper bound on the \acl{PPoA} that leverages $(\lambda,\mu,q)$-smoothness. 

\begin{proposition}
\label{pr:different-p-i-pro}
For each \acl{BCG} $\Gamma^{\boldsymbol{p}}\in\mathcal{G}(\mathcal{C},q)$, 
we have $\PPoA(\Gamma^{\boldsymbol{p}})\leq\gamma_{\mathsf{pr}}(\mathcal{C},q)$.
\end{proposition}

Similarly to the case of the ordinary planner, the next result shows that these bounds are tight for every family $\mathcal{C}$ of nonnegative and nondecreasing cost functions and each $q\in(0,1]$.

\begin{theorem}
\label{th:lambda-mu-pr}
For each family $\mathcal{C}$ of nonnegative and nondecreasing cost functions and each $q\in(0,1]$, we have $\PPoA(\mathcal{C},q)=\gamma_{\mathsf{pr}}(\mathcal{C},q)$.
\end{theorem}

By taking $q=1$, Theorem~\ref{th:lambda-mu-pr} yields $\PoA(\mathcal{C})=\gamma(\mathcal{C})$, providing an alternative proof of \citet[theorem~5.8]{Rou:JACM2015}. 
Comparing both proofs, ours uses a compactness argument that directly reduces the analysis to a finite subfamily of costs.
Another difference between the lower bound construction of \cite{Rou:JACM2015} and Theorem~\ref{th:lambda-mu-pr} is that, in order to handle the case $q<1$, we need to give the prophet sufficient flexibility so as to distribute players as equally as possible across the resources. 
To achieve this,  our tight examples allow a multitude of alternative strategies for the players and, as a consequence, it is unclear whether one can find tight examples encoded in routing games, instead of general congestion games. 
This is in contrast with the bound $\PoA(\mathcal{C})=\gamma(\mathcal{C})$ in \cite{Rou:JACM2015} which was shown to be attainable
by routing games, provided that the zero cost function $c_{0}(\argdot)$ belongs to the class $\mathcal{C}$.

\begin{corollary}
\label{co:gamma-nondecreasing}
For each family $\mathcal{C}$ of nonnegative and nondecreasing cost functions, we have that $\gamma_{\mathsf{pr}}(\mathcal{C},q)$ is nondecreasing in $q$.
\end{corollary}

\begin{remark}
\label{re:monotonicity}
We highlight the fact that, for a fixed game $\Gamma$, $\PPoA(\Gamma)$ can decrease with the maximal probability $q$, and even with respect to a $p_{i}$ separately. 
This can be shown using the same construction as in  Example~\ref{ex:monotonicity}.
\end{remark}

A consequence of Theorem~\ref{th:lambda-mu-pr} is that, for a fixed number of players $n$ and under a mild growth condition on the family of costs $\mathcal{C}$, the \acl{OPoA} and \acl{PPoA} converge to $1$ as the probabilities $p_{i}$ tend to $0$. 
In Section~\ref{se:poa_and_pos} we will see that this is no longer the case when the number of players is not bounded.

\begin{proposition} 
\label{pr:growth-cond}
Let $\OPoA(\mathcal{C},n,q)$ and 
$\PPoA(\mathcal{C},n,q)$  denote the supremum of
$\OPoA(\Gamma^{\boldsymbol{p}})$ and $\PPoA(\Gamma^{\boldsymbol{p}})$ 
respectively, over all \aclp{BCG}
$\Gamma^{\boldsymbol{p}}\in\mathcal{G}(\mathcal{C},q)$ with a fixed number $n$ of players.
Suppose that there exists a constant $H$ such that 
$c(n) \leq  H c(1)$ for all $c\in\mathcal{C}$.
Then $1\leq \OPoA(\mathcal{C},n,q)\leq \PPoA(\mathcal{C},n,q)\to 1$ as $q\to 0$.
\end{proposition}

The proof, presented in Section~EC.1 of the supplementary material, proceeds by showing that the $(\lambda,\mu,q)$-smoothness condition \eqref{eq:ccp} holds with 
$\mu=0$ and a suitable $\lambda=\lambda(q)$ such that
$\lambda(q)\to 1 $ when $q\to 0$.
The assumption $c(n)\leq H c(1)$ holds trivially when the family $\mathcal{C}$ is finite. This is the case when we consider a fixed graph $G$ with given costs $c_{e}(\argdot)$ and a fixed number of players, and we study the behavior of the \ac{PoA} when $\max_{i\in\mathcal{N}}p_{i}\to 0$.
Another interesting case is when $\mathcal{C}$ is the class of all polynomials with nonnegative coefficients $a_{i}\geq 0$ and maximum degree $d$.
Indeed, for such polynomials we have
$c(n)\leq n^{d} \sum_{i=0}^{d} a_{i}=n^{d}c(1)$
and we can set $H=n^{d}$.

\subsection{Extension to Mixed and Correlated Equilibria}
\label{se:PoA_MixedEquilibria}

Although unweighted congestion games admit pure equilibria -- also in the stochastic version studied in this paper -- there are good reasons for considering weaker solution concepts, such as mixed, correlated, and coarse correlated equilibria. 
In particular, when agents use no-regret algorithms, the empirical distribution of play
is an approximate correlated equilibrium in the case of internal regret, and an approximate correlated equilibrium in the case of external regret \citep[see, e.g.,][]{Han:CTG1957,FosVoh:GEB1997,CesLug:CUP2006}.

These solution concepts have also been studied before in congestion games. \citet{Rou:JACM2015} showed that every \ac{PoA} bound based on the $(\lambda,\mu)$-smoothness condition \eqref{eq:def-lambda-mu} remains valid for mixed equilibria, correlated equilibria and coarse correlated equilibria. 
It follows directly from this that the same holds for the estimates of the \acl{OPoA} based on \eqref{eq:r1}. 
Below we establish analogous bounds for the \acl{PPoA} based on $(\lambda,\mu,q)$-smoothness. 

We first recall the notions of mixed and correlated equilibria.
For any probability distribution $\tau\in\Delta(\mathcal{S})$, we let $C_{i}(\tau)=\Expect_{\boldsymbol{s} \sim \tau}[C_{i}(\boldsymbol{s})]$ denote the expected cost of player $i$, and $C(\tau)=\Expect_{\boldsymbol{s} \sim \tau}[C(\boldsymbol{s})]=\sum_{i\in\mathcal{N}}C_{i}(\tau)$ denote the expected social cost.
The expected cost of a player is taken over the appropriate distribution depending on the context, as we describe below. 
Examples include Bernoulli players, and Bayes-Nash mixed or correlated equilibria.

A mixed strategy profile is
a tuple $\boldsymbol{\sigma} = (\sigma_{j})_{j\in\mathcal{N}}$, where $\sigma_{j}\in\Delta(\mathcal{S}_{j})$ is a mixed strategy for player $j$.
Each player draws a strategy $s_{j} \sim \sigma_{j}$ independently, so that the strategy profile $\boldsymbol{s} \in \mathcal{S}$ is distributed according to the product probability measure $ \sigma=\otimes \sigma_{j}$. 
Note that there is a one-to-one correspondence between the tuples $\boldsymbol{\sigma} = (\sigma_{j})_{j\in\mathcal{N}}$ and the product probabilities $\sigma=\otimes \sigma_{j}$.
A \emph{Bayes-Nash mixed equilibrium} is then a probability $\hat{\sigma} = \otimes \hat{\sigma}_{j}$ such that, for each player $i$ and every alternative strategy $\sigma_{i}\in\Delta(\mathcal{S}_{i})$, we have
\begin{equation}
\label{eq:MixedBayesNash}
C_{i}(\hat{\sigma}_{i}\otimes\hat{\sigma}_{-i})\leq C_{i}(\sigma_{i}\otimes\hat{\sigma}_{-i}),
\end{equation}
where $\sigma_{-i} = \otimes_{j\neq i}\sigma_{j}$ stands for the product probability of the family
$(\sigma_{j})_{j\neq i}$.

The weaker \emph{Bayes-Nash correlated equilibrium} is a probability distribution $\hat{\tau}\in\Delta(\mathcal{S})$ (not necessarily of product form) such that, for each deviating strategy $s_{i}^{\ast}\in\mathcal{S}_{i}$ by any player $i$, we have
\begin{equation}
\label{eq:corrEq}
C_{i}(\hat{\tau})=\Expect_{\boldsymbol{s}\sim\hat{\tau}}\bracks*{C_{i}(\boldsymbol{s})}
\leq\Expect_{\boldsymbol{s}_{-i} \sim \hat{\tau} \mid s_{i}^{\ast}}\bracks*{C_{i}(s_{i}^{\ast},\boldsymbol{s}_{-i})},
\end{equation}
where $\hat{\tau} \mid s_{i}^{\ast}$ denotes the conditional distribution
of $\boldsymbol{s}_{-i}$ given $s_{i}^{\ast}$.

The even weaker \emph{Bayes-Nash coarse correlated equilibrium} is a distribution $\hat{\tau}\in\Delta(\mathcal{S})$ such that for each deviating strategy $s_{i}^{\ast}\in\mathcal{S}_{i}$ by any player $i$ we have
\begin{equation}
\label{eq:coarseEq}
C_{i}(\hat{\tau})=\Expect_{\boldsymbol{s}\sim\hat{\tau}}\bracks*{C_{i}(\boldsymbol{s})}
\leq\Expect_{\boldsymbol{s}\sim\hat{\tau}}\bracks*{C_{i}(s_{i}^{\ast},\boldsymbol{s}_{-i})}.
\end{equation}

We let $\mathsf{NE}^{\mathsf{mix}}(\Gamma^{\boldsymbol{p}})$, $\mathsf{NE}^{\mathsf{cor}}(\Gamma^{\boldsymbol{p}})$, and $\mathsf{NE}^{\mathsf{coa}}(\Gamma^{\boldsymbol{p}})$ denote the set of Bayes-Nash mixed equilibria,  correlated equilibria, and coarse correlated equilibria, respectively.
The corresponding definitions of 
\acl{PPoA} are similar to the one given in 
\eqref{eq:PoAp}:
\begin{subequations}
\label{eq:Poa-various}
\begin{align}
\label{eq:PoAp-mixed}
\MPPoA(\Gamma^{\boldsymbol{p}})
&\eqdef\max\left\{{C(\boldsymbol{\sigma})}:\boldsymbol{\sigma}\in\mathsf{NE}^{\mathsf{mix}}(\Gamma^{\boldsymbol{p}})\right\}/{C_{\mathsf{pr}}},\\
\label{eq:PoAp-corr}
\CorPPoA(\Gamma^{\boldsymbol{p}})
&\eqdef\max\left\{C(\tau):\tau\in\mathsf{NE}^{\mathsf{cor}}(\Gamma^{\boldsymbol{p}})\right\}/{C_{\mathsf{pr}}},\\
\label{eq:PoAp-coarse}
\CoaPPoA(\Gamma^{\boldsymbol{p}})
&\eqdef\max\left\{C(\tau):\tau\in\mathsf{NE}^{\mathsf{coa}}(\Gamma^{\boldsymbol{p}})\right\}/{C_{\mathsf{pr}}}.
\end{align}
\end{subequations}
Notice that the maxima in \eqref{eq:Poa-various} are well defined. Indeed, for any fixed game $\Gamma^{\boldsymbol{p}}$ the  maxima are attained because the social cost function $C(\argdot)$ is continuous and the sets $\mathsf{NE}^{\mathsf{mix}}(\Gamma^{\boldsymbol{p}})$, $\mathsf{NE}^{\mathsf{cor}}(\Gamma^{\boldsymbol{p}})$, and $\mathsf{NE}^{\mathsf{coa}}(\Gamma^{\boldsymbol{p}})$ are compact.

\begin{theorem} 
\label{th:PoAMixed}
For each \acl{BCG} $\Gamma^{\boldsymbol{p}}\in\mathcal{G}(\mathcal{C},q)$ 
we have 
\begin{align}
\PPoA(\Gamma^{\boldsymbol{p}})
&\leq\MPPoA(\Gamma^{\boldsymbol{p}})\leq \CorPPoA(\Gamma^{\boldsymbol{p}})\nonumber\\
\label{eq:PoA-pr-ineq}
&\leq \CoaPPoA(\Gamma^{\boldsymbol{p}})\leq\gamma_{\mathsf{pr}}(\mathcal{C},q). 
\end{align}
Moreover, all these bounds are tight.
\end{theorem}

\proof{Proof.}
The order between the different prices of anarchy 
follows directly from the chain of inclusions
$\mathsf{NE}(\Gamma^{\boldsymbol{p}})\subseteq\mathsf{NE}^{\mathsf{mix}}(\Gamma^{\boldsymbol{p}})\subseteq\mathsf{NE}^{\mathsf{cor}}(\Gamma^{\boldsymbol{p}})\subseteq\mathsf{NE}^{\mathsf{coa}}(\Gamma^{\boldsymbol{p}})$,
so that it suffices to establish the rightmost bound $\CoaPPoA(\Gamma^{\boldsymbol{p}})\leq\gamma_{\mathsf{pr}}(\mathcal{C},q)$.
Take any coarse correlated equilibrium 
$\hat{\tau}\in \mathsf{NE}^{\mathsf{coa}}(\Gamma^{\boldsymbol{p}})$ and fix $(\lambda,\mu)$ satisfying \eqref{eq:ccp}. 
Considering 
$s_{i}^{\ast} \sim\sigma_{i}^{\ast}$ for the prophet-like strategies, and taking expectation in \eqref{eq:coarseEq} we get
\begin{align*}
C(\hat{\tau})
&\leq \sum_{i\in\mathcal{N}}\Expect_{s_{i}^{\ast}\sim\sigma_{i}^{\ast}} \Expect_{\boldsymbol{s}\sim\hat{\tau}}\bracks*{C_{i}(s_{i}^{\ast},\boldsymbol{s}_{-i})}\\
&=\Expect_{\boldsymbol{s}\sim\hat{\tau}}\bracks*{\sum_{i\in\mathcal{N}}\Expect_{s_{i}^{\ast} \sim\sigma_{i}^{\ast}}\bracks*{C_{i}(s_{i}^{\ast},\boldsymbol{s}_{-i})}}.
\end{align*}
Then, proceeding as in the proof of Proposition~\ref{pr:different-p-i-pro}, we may use the inequality derived there -- see (EC.4) in the supplementary material -- to obtain
\begin{equation*}
C(\hat{\tau})\leq \Expect_{\boldsymbol{s}\sim\hat{\tau}}\bracks*{\lambda\,C_{\mathsf{pr}}+\mu\,C(\boldsymbol{s})}
=\lambda\,C_{\mathsf{pr}}+\mu\,C(\hat{\tau}),  
\end{equation*}
so that $C(\hat{\tau})/C_{\mathsf{pr}}\leq\lambda/(1-\mu)$, and we conclude by 
taking the infimum over $(\lambda,\mu)$ and 
maximizing over $\hat{\tau}\in\mathsf{NE}^{\mathsf{coa}}(\Gamma^{\boldsymbol{p}})$.
Tightness follows directly from Theorem~\ref{th:lambda-mu-pr}.
\Halmos
\endproof

\section{Price of Anarchy with Affine Costs}
\label{se:poa_and_pos}

This section mostly focuses on atomic \aclp{BCG} with nondecreasing and nonnegative affine costs, that is, we restrict the attention to the class $\mathcal{C}_{\mathsf{aff}}$ of costs of the form 
\begin{equation}
\label{eq:affine-costs}
c(x)=a+bx\quad\text{with}\quad a, b\geq 0. 
\end{equation}
Specifically, Sections~\ref{suse:tight-upper-bounds} and \ref{suse:tight-bounds-PoA-prophet}
respectively provide explicit analytic expressions for the ordinary and prophet \ac{PoA} for affine costs. 
Section~\ref{suse:general-costs} presents some partial extensions for polynomial costs and puts forward two conjectures.

\subsection{Tight Bounds for the \acl{OPoA}}
\label{suse:tight-upper-bounds}
From Theorem~\ref{th:different-p-i} we know that 
$\OPoA(\Gamma^{\boldsymbol{p}})$ is maximal when all the probabilities are equal. 
The following theorem is our main estimate for the \ac{OPoA} with affine costs,
which determines explicitly the tight bounds $\OPoA(\mathcal{C}_{\mathsf{aff}},q)=\gamma(\mathcal{C}_{\mathsf{aff}}^{q})$ as a function of $q\in(0,1]$, with three different regimes. 
See Figure~\ref{fi:envelope} for the details. 

\begin{theorem}
\label{th:upper}
Let $\bar{q}_{0}=1/4$ and let $\bar{q}_{1}\sim 0.3774$ be the real root of $8 q^{3}+4q^{2}=1$.
Then, 
\begin{align}
&\OPoA(\mathcal{C}_{\mathsf{aff}},q)=\gamma(\mathcal{C}_{\mathsf{aff}}^{q})\nonumber\\
\label{eq:PoAbound}
&\quad= 
\begin{cases}
  4/3&\text{ if }\ 0< q\leq \bar{q}_{0},\\
  \dfrac{1+q+\sqrt{q(2+q)}}{1-q+\sqrt{q(2+q)}}&\text{ if }\ \bar{q}_{0}\leq q\leq \bar{q}_{1},\\
1+q+\dfrac{q^{2}}{1+q}&\text{ if }\ \bar{q}_{1}\leq q\leq 1.
\end{cases}
\end{align}
\end{theorem}

\begin{figure}[ht]
\centering
\includegraphics[scale=0.75]{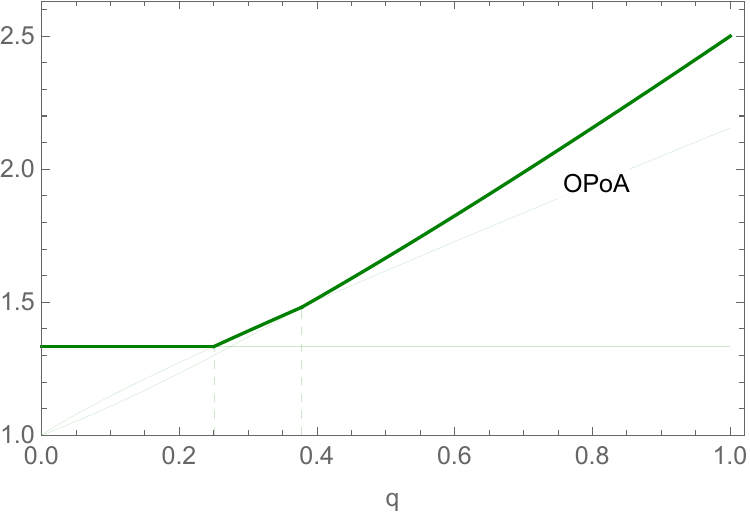}
  \caption{\label{fi:envelope} The upper envelope gives the tight bound on $\OPoA(\mathcal{C}_{\mathsf{aff}},q)$ as a function of $q$.}
\end{figure}

The proof is long and technical, especially in the intermediate range $\bar{q}_{0} \le q\leq\bar{q}_{1}$. It can be found in Section~EC.1 of the supplementary material. 
Here is a short sketch to illustrate the overall ideas.
The proof proceeds through a series of lemmas that characterize the optimal parameters $(\lambda,\mu)$  in \eqref{eq:r1} for each value of $q$.
Even if Theorem~\ref{th:different-p-i} already implies the tightness of the bound \eqref{eq:PoAbound}, in Section~EC.2 of the supplementary material we present three specific examples that attain this bound in the three different ranges of $q$.
These examples are somewhat simpler than those proposed in \citet{Rou:JACM2015}
and \citet{Gai:Thesis2006} and, moreover, they involve only purely linear costs of the form $c(x)=ax$ with $a\geq 0$.

\subsection{Tight Bounds for the \acl{PPoA}}
\label{suse:tight-bounds-PoA-prophet}
We now proceed to find tight bounds $\gamma_{\mathsf{pr}}(\mathcal{C}_{\mathsf{aff}},q)$ for the \acl{PPoA} of \aclp{BCG} with affine costs in $\mathcal{C}_{\mathsf{aff}}$. We can bound 
$\gamma_{\mathsf{pr}}(\mathcal{C}_{\mathsf{aff}},q)$ from above by the lower envelope of a countable family of functions as in Figure~\ref{fi:bounds-prophet}. Later we will show that these bounds are tight.
\begin{proposition}
\label{pr:lm-linear}
Let
\begin{subequations}
\label{eq:Xi-xi}
\begin{align}
\Xi(q)&=\inf_{\ell\geq 1}\xi_{\ell}(q) \\
\intertext{with}
\xi_{\ell}(q)&=\frac{\ell(\ell+1)q^2+2\ell q+1}{2\ell q}
\quad\text{for all }\ell\in \mathbb{N}\setminus\{0\}.
\end{align}
\end{subequations}
Then
\begin{equation}
\label{eq:PPOA-bound-linear}    
\PPoA(\mathcal{C}_{\mathsf{aff}},q)=\gamma_{\mathsf{pr}}(\mathcal{C}_{\mathsf{aff}},q)\leq\Xi(q) .
\end{equation}
\end{proposition}

\begin{figure}[ht]
\centering
\includegraphics[scale=0.75]{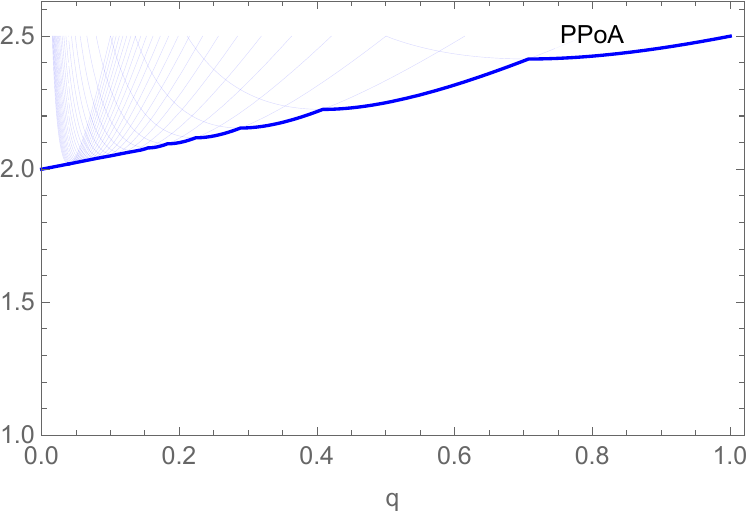}
\caption{\label{fi:bounds-prophet} The lower envelope $\Xi(q)$ gives a tight bound for $\PPoA(\mathcal{C}_{\mathsf{aff}},q)$.}
\end{figure} 
The proof proceeds by identifying,
for any given $\ell\geq 1$, some
specific values of $\lambda$ and $\mu$ that satisfy $(\lambda,\mu,q)$-smoothness with $\lambda/(1-\mu)=\xi_{\ell}(q)$.
Note that for any fixed $\ell$ the bound $\Xi(q)=\xi_{\ell}(q)$ is attained 
for all $q$ in the following interval (see Figure~\ref{fi:bounds-prophet}):
\begin{equation}
\label{eq:range}
\frac{1}{\sqrt{\ell(\ell+1)}}\leq q\leq \frac{1}{\sqrt{\ell(\ell-1)}}.
\end{equation}

\begin{corollary}
\label{co:q-1/l}
If $q=1/\ell$ for some $\ell\in\mathbb{N}\setminus\{0\}$, then $\gamma_{\mathsf{pr}}(\mathcal{C}_{\mathsf{aff}},q)\leq 2+q/2$.
\end{corollary}

\proof{Proof.}
Because $q=1/\ell$ lies in the range \eqref{eq:range}, by direct substitution we have $\Xi(q)=\xi_{\ell}(q)=2+q/2$.
\Halmos
\endproof

In Section~EC.3 of the supplementary material we construct a specific family of \aclp{BCG} $\Gamma^{\boldsymbol{p}}$ with
$p_{i}\equiv q$
for which
$\PPoA(\Gamma^{\boldsymbol{p}})$ approaches
$\Xi(q)$.
This implies that \eqref{eq:PPOA-bound-linear} in Proposition~\ref{pr:lm-linear} holds with equality.
\begin{theorem}
\label{th:Tight-PPoA-Affine}
$\PPoA(\mathcal{C}_{\mathsf{aff}},q)=\gamma_{\mathsf{pr}}(\mathcal{C}_{\mathsf{aff}},q)=\Xi(q)$ for all $q\in (0,1]$.
\end{theorem}

\subsection{Extension to polynomial costs}
\label{suse:general-costs}

Let  $\mathcal{P}_{d}$ be the class of polynomial costs  
$c(x)=\sum_{i=0}^{d} a_{i}x^{i}$ of  degree $d$ with $a_{i}\geq 0$.
Using the raw moments of binomial random variables and Stirling numbers $\stirling{i}{j}$ of the second kind \cite[see theorem 2.2 of][]{Kno:SIAMJAM2008}, the modified cost functions for the monomials $c_{i}(x) \eqdef x^{i}$ become
\begin{equation}
\label{eq:polyq}
c_{i}^{q}(k)=\sum_{j=0}^{i}\stirling{i+1}{j+1}\,q^{j+1}(k-1)\,\cdots(k-j).
\end{equation}
Each $c_{i}^{q}(\argdot)$ remains a polynomial in $k$ of degree $i$, 
but the formula becomes quite complex as $i$ increases.

Thus, deriving analytic expressions for $\gamma(\mathcal{P}_{d}^{q})$ and $\gamma_{\mathsf{pr}}(\mathcal{P}_{d},q)$ as a function of $q$ to compute the price of anarchy -- as done in Theorems~\ref{th:upper} and \ref{th:Tight-PPoA-Affine} for affine costs -- would be even more technical than the analysis in Section~EC.1 of the supplementary material. 

However, from \citet[theorem~7]{ComScaSchSti:MOR2023} we know that the limit of $\gamma(\mathcal{P}_{d}^{q})$ for $q\downarrow 0$ coincides with the nonatomic bound, which, combined with \citet{Rou:JCSS2003}, yields
\begin{equation*}
\lim_{q\downarrow 0}\OPoA(\mathcal{P}_{d},q)
=\frac{1}{1-d/(d+1)^{(d+1)/d}}.    
\end{equation*}
We conjecture that this is not only attained in the limit, but equality holds for all $q\leq
d/\parens*{(2^{d}-1)(d+1)^{(d+1)/d}}$.

On the other hand, we have the following lower bound in the prophet case, which shows that  $\PPoA(\mathcal{P}_{d},{0^+})$ is significantly larger than the $\OPoA(\mathcal{P}_{d},{0^+})$.
The proof is in Section~EC.4 of the supplementary material.
\begin{proposition}
\label{pr:proph}
Let $B_{d} = \sum_{j=0}^{d} \stirling{d}{j}$ be the $d$-th Bell number. 
Then 
$\PPoA(\mathcal{P}_{d},{0^+})\geq  B_{d+1}$.
\end{proposition}

At the other extreme, when $q\uparrow 1$, both $\OPoA(\mathcal{P}_{d},q)$ and $\PPoA(\mathcal{P}_{d},q)$
increase monotonically to the same deterministic atomic bound given in \citet{AlaDumGaiMonSch:SIAMJC2011}, that is 
\begin{align*}
&\OPoA(\mathcal{P}_{d},1)=\PPoA(\mathcal{P}_{d},1)\\
&\quad=\frac{(\varphi+1)^{2 d+1}-\varphi^{d+1}(\varphi+2)^{d}}{(\varphi+1)^{d+1}-(\varphi+2)^{d}+(\varphi+1)^{d}-\varphi^{d+1}}, 
\end{align*}
where $\varphi=\lfloor x^*\rfloor$, with $x^*$ the unique positive solution of the equation $(x+1)^{d}=x^{d+1}$. 

Moreover, as shown by \citet{Rou:JCSS2003} in the nonatomic case and by \citet{AlaDumGaiMonSch:SIAMJC2011} for 
  deterministic atomic games, the bounds in these two limit cases are dominated by the highest degree monomial
$c_{d}(x)=x^{d}$, and the lower terms $c_{i}(x)=x^i$ for $i<d$ can be ignored.
As a matter of fact, for the \ac{OPoA} this holds not only for $q\downarrow 0$ and $q\uparrow 1$ but for all $q\in (0,1]$ (again, the proof is in Section~EC.4 of the supplementary material).

\begin{proposition}
\label{pr:polycosts}
For all $q\in (0,1]$ we have $\OPoA(\mathcal{P}_{d},q)=\gamma(\{c_{d}^{q}\})$.
\end{proposition}  

We conjecture that a similar equality holds for the prophet \ac{PoA}, namely, $\PPoA(\mathcal{P}_{d},q)=
\gamma_{\mathsf{pr}}(\{c_{d}\},q)$.

\section{Conclusions}\label{se:concl}

Our work studies atomic congestion games with stochastic demands. 
In the model we propose,  each player either participates in the game with an idiosyncratic probability or stays out.
We contrast the ensuing equilibria with what can be achieved by central planners with different foresight skills. 
A prophet planner has access to real-time information and can make contingent plans after learning the demand realizations;
an ordinary planner does not have access to up-to-date information and can only plan based on the demand distribution.
Our main results consist of analytic expressions that describe how the \ac{PoA} changes as a function of the user-participation probabilities. 
We have computed these expressions explicitly for games with affine cost functions. 
In a high participation regime, results tend to what is known for deterministic games. 
More interestingly, for low participation, equilibria are closer to social optima because worst-case inefficiencies arise when certain levels of congestion in the system are attained. 
Our results quantify the value of the additional information available to prophet planners, as mediated by the participation probabilities. 
We also note that the resulting curves have various regimes and are not concave nor convex. 
This discussion is related to information availability and how it affects the economics of networks for platforms such as Google maps, Apple maps, and Waze.

Although not a focal part of the paper, we highlight that market operators that act as a prophet planner can better match participating users to optimal routes. 
This could be used to guide the system to a more efficient market outcome, either through route recommendations, routing directly if cars were self-driving, or indirectly through information transmission or pricing. 
Depending on implementation constraints, ordinary planners may only charge fixed fees (e.g., network pricing in London, implemented around two decades ago), versus modern systems with real-time information that perform dynamic pricing (e.g., highway high-occupancy vehicle lanes).

To put the efficiency results in perspective,
\citet{Rou:ACMTEC2015} showed that, when dealing with games of incomplete information, the bounds for the corresponding games of complete information are still valid for prophet planners. 
His framework for incomplete information games is very robust, but requires a smoothness definition that holds across different types \citep[see][definition~3.1 and remark~3.2]{Rou:ACMTEC2015}.
A result in the same spirit appears in \cite{CorHoeSch:TRB2019}, who showed that upper bounds derived from the smoothness framework continue to hold for ordinary planners in \acp{BCG} even if the events of players being active are not independent and identically distributed.
These authors consider a class of games and an information structure that makes these objects games of incomplete information; then they compute bounds for the \ac{PoA} of games in this class over all possible probabilities that characterize the incomplete information.
They show the remarkable result that the performance of the \ac{PoA} does not decay in the presence of incomplete information.

Our results are in a different spirit.
We fix not only the class of games and the information structure, but also the probability measure and examine the behavior of the \acl{PoA} as this probability varies.
In our case, when the probability is characterized by a single parameter $q$, this is tantamount to studying the \ac{OPoA} and \ac{PPoA} as a univariate function of this parameter. 
This means that, for a fixed value of $q$, we consider the worst-case \ac{OPoA} and \ac{PPoA} among all possible instances where participation probabilities of players are bounded above by $q$.
The main results in this respect are:
\begin{enumerate}[(a)]
\item For any family $\mathcal{C}$ of  nonnegative and nondecreasing cost functions and any $q\in(0,1]$ we have $\OPoA(\mathcal{C},q) = \gamma(\mathcal{C}^{q})$ and $\PPoA(\mathcal{C},q)=\gamma_{\mathsf{pr}}(\mathcal{C},q)$. In particular, 
this implies that both 
$\gamma(\mathcal{C}^{q})$ as well as 
$\gamma_{\mathsf{pr}}(\mathcal{C},q)$ are  nondecreasing in $q$.

\item For the class $\mathcal{C}_{\mathsf{aff}}$ of affine costs, we provide analytic expressions for the worst case bounds 
$\OPoA(\mathcal{C}_{\mathsf{aff}},q)$ and
$\PPoA(\mathcal{C}_{\mathsf{aff}},q)$
as functions of $q$. 

\item The presence of two kinks in the  function $\OPoA(\mathcal{C}_{\mathsf{aff}},q)$, which turns out to be constant and equal to $4/3$ for $q\in(0,1/4]$, exactly as in nonatomic congestion games with affine costs, whereas the maximum of $5/2$, which is the \ac{PoA} in the atomic case, is only attained in the limit as $q \to 1$ (see Figure~\ref{fi:bounds}).

\item The presence of countably many kinks in $\PPoA(\mathcal{C}_{\mathsf{aff}},q)$ and its convergence to $2$ as $q \to 0$ (see Figure~\ref{fi:bounds}).
\end{enumerate}

Several natural questions remain open. First, it is unclear how to adapt the current lower bound construction of $\PPoA(\mathcal{C},q)=\gamma_{\mathsf{pr}}(\mathcal{C},q)$ to routing games.

Second, what can be said about the \acdef{PoS}, which captures the inefficiency of the best equilibria, as first defined by \citet{SchSti:P14ACM2003} and coined by \citet{AnsDasKleTarWexRou:SIAMJC2008}, 
\citet[theorem~5]{KleSch:TCS2019} established a tight bound of $1+\sqrt{q/(2+q)}$ for all $q\ge 1/4$ on the ordinary price of stability. 
A tight bound of $4/3$ for all $q\le 1/4$ on the ordinary price of stability completes this characterization. However, how does a characterization looks like for the prophet price of stability?

Another interesting question is to consider a version of Theorem~\ref{th:different-p-i} where the number $n$ of players is kept fixed, and shed light on the efficiency of equilibria.

Finally, our model works also without the independence hypothesis, in the sense that a stochastic congestion game can be defined for any joint distribution of players participation.
The equilibria and optimum of the game will depend on the whole distribution and not just on the marginals, so the game will require a more complex description.
Moreover, if the agents take part in the game in a correlated way, without any constraint on the possible dependence structure, then the best lower and upper bounds for the \ac{OPoA} and the \ac{PPoA} coincide with the bounds for the deterministic game. 
To wit, let the participation of the players be comonotonic with equal marginals $p$, i.e.,  with probability $p$ all players take part in the game and with probability $(1-p)$ they are all absent. 
Take a deterministic game $\Gamma$ with social cost function $C$ and let $\hat{\boldsymbol{s}}$ and $\boldsymbol{s}^{\ast}$ be, respectively, the worst equilibrium and an optimum of this deterministic game.
Then, for any $p\in(0,1]$, in the stochastic congestion we have that the worst equilibrium and the optimum are the same as in the deterministic game, both in the ordinary and prophet cases. 
The corresponding social costs are just the respective social costs of the deterministic game multiplied by $p$.
This implies that for every $p\in(0,1]$ the \ac{OPoA} and the \ac{PPoA} are equal to the \ac{PoA} of the deterministic game.
It is enough to choose a game $\Gamma$ that achieves the worst \ac{PoA} in a class to get our result.

\subsection*{Acknowledgments}
Some results about the ordinary planner that appear in this paper were presented at several seminars and at EC'19. We would like to thank the reviewers and participants for several insightful comments and remarks that made the results better and the presentation more clear. In addition, we would like to thank Vittorio Bil\`o for pointing out the relations to one of his earlier papers.
We thank an anonymous reviewer in the journal submission process for prompting us to consider the case where the hypothesis of independence for the agents' participation is removed.

This collaboration started in Schloss Dagstuhl at the Seminar on Dynamic Traffic Models in Transportation Science in 2018.
Roberto Cominetti gratefully acknowledges support from Luiss University for a visit during which part of this research took place.
Marco Scarsini gratefully acknowledges the support and hospitality of N\'ucleo Milenio ``Informaci\'on y Coordinaci\'on en Redes.''
Part of this work was carried out when Marc Schr\"oder was a visiting professor at Luiss University. 
This project was further carried out when Marc Schr\"oder and Marco Scarsini were taking part in the program on Games, Learning, and Networks at the Institute for Mathematical Sciences, National University of Singapore in 2023.
Marco Scarsini is a member of GNAMPA-INdAM.

\theendnotes

%\iftrue % compile body only (which contains bibliography)
%\iffalse % compile ecompanion (which contains body without bibliography and e-comp bibliography at the end)\

\bibliographystyle{informs2014}
\bibliography{bibstochcong}

%  \end{document}
%\fi

\clearpage

\begin{APPENDICES}
%% Here starts the e-companion (EC)
%%%%%%%%%%%%%%%%%%%%%%%%%%%%%%%%%%%%%%%%%%%%%%%%%%%%%%%%%%
\ECSwitch

%\ECDisclaimer
%%%%%%%%%%%%%%%%%%%%%%%%%%%%%%%%%%%%%%%%%%%%%%%%%%%%%%%%%%

%%% Main head for the e-companion
\ECHead{Supplementary Material to `Ordinary and Prophet Planning under Uncertainty in Bernoulli Congestion Games' by Cominetti, Scarsini, Schr\"oder, and Stier-Moses}

%\newpage
%
%\begin{APPENDIX}{}%{Proofs}

%--------------------Section-------------------------------

%\setcounter{page}{1}

\section{Proofs}
\label{se:proofs}

\subsection*{Proofs of Section~\ref{se:homogeneous-prob}}

We will use the following equivalent expression for $\gamma_{\mathsf{pr}}(\mathcal{C},q)$.
\begin{lemma}
\label{le:boundalt} 
Let  
$\Psi(\omega)
\eqdef 
\sup_{(c,k,m)\in\mathcal{T}}
\beta_{c,k,m}(\omega)$
 denote the supremum of the affine functions
\begin{equation}
\label{eq:affine-functions}
\beta_{c,k,m}(\omega)\eqdef\frac{c^{q}(1+m)}{q\,c(k)}\,\omega+\frac{m \, c^{q}(m)}{k\, c(k)}(1-\omega).
\end{equation}
Then
\begin{equation}\label{eq:bnd-equiv2}
\gamma_{\mathsf{pr}}(\mathcal{C},q)=\inf_{\omega\geq 1}\Psi(\omega).
\end{equation}

\end{lemma}

\proof{Proof.}
If $c(k)=0$ for some $c\in\mathcal{C}\setminus\{c_{0}\}$ and $k\geq 1$, then by taking the largest such $k$ we have $c(k+1)>0$ and then for $m=k$ we get $c^{q}(k)=0<c^{q}(k+1)$ so that the right hand side in \eqref{eq:ccp} is $0$ whereas the expression on the left is strictly positive.
Hence no pair $(\lambda,\mu)$ satisfies \eqref{eq:ccp} and $\gamma_{\mathsf{pr}}(\mathcal{C},q)=\infty$. Similarly,
\begin{equation}
\label{eq:Psi>=} 
\Psi(\omega)\geq\beta_{c,k,k}(\omega)=\frac{c^{q}(1+k)}{q\,c(k)}\omega\equiv\infty,
\end{equation}
so that $\inf_{\omega\geq 1}\Psi(\omega)=\infty=\gamma_{\mathsf{pr}}(\mathcal{C},q)$.

Suppose next that $c(k)>0$ for all $c\in\mathcal{C}\setminus\{c_{0}\}$ and $k\geq 1$.
Then, for each $\mu\in[0,1)$ the smallest
$\lambda$ that satisfies \eqref{eq:BndPr} is
\[
\lambda=\sup_{(c,k,m) \in\mathcal{T}}\frac{\frac{1}{q}\,k \,c^{q}(1+m)-\mu\, m \, c^{q}(m)}{k\, c(k)}.
\]
Dividing by $(1-\mu)$ and using the change of variable $\omega=1/(1-\mu)$,
we obtain  \eqref{eq:bnd-equiv2}.
\Halmos
\endproof

\begin{remark}
The objective function $\Psi(\argdot)$ in \eqref{eq:bnd-equiv2} is a supremum of affine functions; therefore, it
is convex and lower semi-continuous. 
Moreover the infimum is attained at some
$\omega\in [1,\gamma_{\mathsf{pr}}(\mathcal{C},q)]$. 
In fact, if $c(1)>0$ for some $c\in\mathcal{C}$ then $\Psi(\omega)\geq\beta_{c,1,0}(\omega)=\omega$ so that $\Psi(\omega)\to\infty$ as $\omega\to \infty$
and the minimum of $\Psi(\argdot)$ is attained.
Otherwise, by the first argument in the proof of Lemma~\ref{le:boundalt}, we have $\Psi(\omega)\geq \gamma_{\mathsf{pr}}(\mathcal{C},q)=\infty$ and every $\omega\in[1,\infty)$ is a minimizer.
\end{remark}

\proof{Proof of Proposition~\ref{pr:different-p-i-pro}.}
Let $\Gamma^{\boldsymbol{p}}\in\mathcal{G}(\mathcal{C},q)$. If $\Gamma^{\boldsymbol{p}}$ does not satisfy the $(\lambda,\mu,q)$-smoothness condition, then $\gamma_{\mathsf{pr}}(\mathcal{C},q)=\infty$ and the statement follows trivially. So we assume the $(\lambda,\mu,q)$-smoothness condition \eqref{eq:ccp}.  We claim that for the prophet-like strategies $\sigma_{i}^{\ast}$ and any fixed strategy profile $\boldsymbol{s}\in\mathcal{S}$ we have
\begin{equation}
\label{eq:estim}
\sum_{i\in\mathcal{N}}\Expect_{s_{i}^{\ast}\sim\sigma_{i}^{\ast}}
\bracks*{C_{i}(s_{i}^{\ast},\boldsymbol{s}_{-i})}
\leq
\lambda\,C_{\mathsf{pr}}+\mu\,C(\boldsymbol{s}).
\end{equation}
Indeed, let $R$ be the sum on the left of \eqref{eq:estim}.
When $s_{i}^{\ast}\sim\sigma_{i}^{\ast}$ we have $\Prob(s_{i}^{\ast}
=s^{\mathcal{I}}_{i})=p(\mathcal{I})/p_{i}$, 
if $i\in\mathcal{I}$, and $\Prob(s_{i}^{\ast}=s^{\mathcal{I}}_{i})=0$ otherwise, so that
\begin{equation}
\label{eq:initial} 
R = \sum_{i\in \mathcal{N}}\sum_{\mathcal{I}:i\in\mathcal{I}}\frac{p(\mathcal{I})}{p_{i}}\,C_{i}(s^{\mathcal{I}}_{i},\boldsymbol{s}_{-i}).
\end{equation}
We now estimate $C_{i}(s^{\mathcal{I}}_{i},\boldsymbol{s}_{-i})$ using Corollary~\ref{co:util2}. 
Set $r_{i}=p_{i}/q\in [0,1]$ and  
$N_{e}^{Y}(\boldsymbol{s}) =\! \sum_{j\in\mathcal{N}}Y_{j}\mathds{1}_{\{e\in s_{j}\}}$
with 
$Y_{j}\sim \Bernoulli(r_{j})$
independent random variables.
Then, from \eqref{eq:disp12} we get
\begin{subequations}
\label{eq:bound-Ci-sI}
\begin{align}
C_{i}(s^{\mathcal{I}}_{i},\boldsymbol{s}_{-i}) 
&= \Expect\bracks*{\sum_{e\in s^{\mathcal{I}}_{i}} c_{e}^{q} \parens*{N_{e}^{Y}(s^{\mathcal{I}}_{i},\boldsymbol{s}_{-i})}}\\
&= \Expect\bracks*{r_{i}\sum_{e\in s^{\mathcal{I}}_{i}} c_{e}^{q}\parens*{1+N_{e}^{Y}(\boldsymbol{s}_{-i})}}\\
&\leq r_{i} \Expect\bracks*{\sum_{e\in s^{\mathcal{I}}_{i}}c_{e}^{q} \parens*{1+N^{Y}_{e}(\boldsymbol{s})}}.
\end{align}
\end{subequations}
Plugging the bound in \eqref{eq:bound-Ci-sI} into \eqref{eq:initial}, and using  $r_{i}/p_{i}=1/q$, we obtain
\begin{align*}
    R&\leq  \Expect\bracks*{\sum_{i\in \mathcal{N}}\sum_{\mathcal{I}:i\in\mathcal{I}}p(\mathcal{I})\sum_{e\in s^{\mathcal{I}}_{i}}\frac{1}{q} c_{e}^{q}\parens*{1+N^{Y}_{e}(s)}}\\
   &=  \Expect\bracks*{\sum_{\mathcal{I}\subseteq\mathcal{N}}\sum_{e\in\mathcal{E}}p(\mathcal{I})\frac{1}{q}\,n^{\mathcal{I}}_{e}(s^{\mathcal{I}})\, c_{e}^{q}\parens*{1+N^{Y}_{e}(s)}}.
\end{align*}
Now, we invoke \eqref{eq:ccp} for $c_{e}(\argdot)$ with $k=n^{\mathcal{I}}_{e}(s^{\mathcal{I}})$ and $m=N^{Y}_{e}(\boldsymbol{s})$,
to derive the bound
\begin{align*}
R
&\leq \Expect\Bigg[
\sum_{\mathcal{I}\subseteq\mathcal{N}}p(\mathcal{I})\sum_{e\in\mathcal{E}}
\Big(\lambda\,n^{\mathcal{I}}_{e}(s^{\mathcal{I}})\,c_{e}
\parens*{n^{\mathcal{I}}_{e}(s^{\mathcal{I}})}
+ \mu\,N^{Y}_{e}(\boldsymbol{s})\,c^{q}_{e}
\parens*{N^{Y}_{e}(\boldsymbol{s})}\Big)\Bigg].
\end{align*}
From \eqref{eq:PESC} and \eqref{eq:disp22}, the right hand side is precisely $\lambda\,C_{\mathsf{pr}}+\mu\,C(\boldsymbol{s})$, which proves \eqref{eq:estim}.

Let $\boldsymbol{s}\in\mathsf{NE}(\Gamma^{\boldsymbol{p}})$ be a \acl{BNE}
and fix $(\lambda,\mu)$ satisfying \eqref{eq:ccp}. 
For each $s_{i}^{\ast}$ we have 
$C_{i}(\boldsymbol{s}) \leq C_{i}(s_{i}^{\ast},\boldsymbol{s}_{-i})$,
so that $C_{i}(\boldsymbol{s})\leq\Expect_{s_{i}^{\ast}\sim\sigma_{i}^{\ast}}[C_{i}(s_{i}^{\ast},\boldsymbol{s}_{-i})]$,
and then \eqref{eq:estim} implies
\begin{equation*}
C(\boldsymbol{s})
=\sum_{i\in \mathcal{N}} C_{i}(\boldsymbol{s})\leq  \lambda\,C_{\mathsf{pr}}+\mu\,C(\boldsymbol{s}). 
\end{equation*}
Thus, $C(\boldsymbol{s})/C_{\mathsf{pr}}\leq\lambda/(1-\mu)$ and we conclude by taking the infimum over $(\lambda,\mu)$ and maximizing over all $\boldsymbol{s}\in\mathsf{NE}(\Gamma^{\boldsymbol{p}})$.
\Halmos 
\endproof

\proof{Proof of Theorem~\ref{th:lambda-mu-pr}.}
From Proposition~\ref{pr:different-p-i-pro} we have $\PPoA(\mathcal{C},q)\leq\gamma_{\mathsf{pr}}(\mathcal{C},q)$, so we only need to show that this bound is tight. We distinguish two cases.

\begin{case}
$c(k)=0$ for some $c\in\mathcal{C}\setminus\{c_{0}\}$ and $k\geq 1$.

As observed in proof of 
Lemma~\ref{le:boundalt}, this is a degenerate case where no pair $(\lambda,\mu)$ satisfies 
\eqref{eq:ccp}, so that $\gamma_{\mathsf{pr}}(\mathcal{C},q)=\infty$. We will build a \acl{BCG} $\Gamma^{\boldsymbol{p}}$ with $n$  players and homogeneous probabilities $p_{i}\equiv q$, such that $\PPoA(\Gamma^{\boldsymbol{p}})=\infty$.

By increasing $k$ we may assume that  $c(k)=0<c(k+1)$, and therefore 
$c^{q}(k)=0<c^{q}(k+1)$.
Consider a game with $n$ players and resource set composed of two disjoint cycles $\mathcal{E}=\mathcal{E}_{1}\cup\mathcal{E}_{2}$ of $n$ resources each with costs $c(\argdot)$ (see Figure~\ref{fi:cycle_npPoA1}).
Every player $i$ has only two possible strategies:
\begin{align*}
 \text{(blue)\quad}s_{i}&=\{a_{i+1},\ldots,a_{i+k}\}\cup\{b_{i+1},\ldots,b_{i+k+1}\},\\
 \text{(red)\quad}s_{i}'&=\{a_{i+k+1},\ldots,a_{i+2k}\}\cup\{b_{i+k+2},\ldots,b_{i+2k+1}\},
\end{align*}
with the convention $i+j\equiv i+j-n$ when $i+j>n$. 
We take $n \geq 2k+1$ so that the strategies $s_{i}$ and $s_{i}'$ do not overlap.
\begin{figure}[ht]
\centering
\begin{tikzpicture}[thick,scale=0.5, every node/.style={transform shape}]

\def \n {7}
\def \radius {4cm}

\def \margin {6} % margin in angles, depends on the radius
\node[scale=1.5](E1)  at (0,0) {\Large$\mathcal{E}_{1}$};

\foreach \s in {1,...,1}
{
 \node[draw, circle,minimum size = 0.6cm](H\s)  at ({90-360/(\n+1) * (\s-1)}:\radius) {};
  \draw[->, >=latex] ({90-360/(\n+1) * (\s-1)-\margin}:\radius) arc ({90-360/(\n+1) * (\s-1)-\margin}:{90-360/(\n+1)*\s+\margin}:\radius) node[midway,fill=white,scale=1.5] {$a_{\s}$};
}  
\foreach \s in {2,...,3}
{
 \node[draw, circle,minimum size = 0.6cm](H\s)  at ({90-360/(\n+1) * (\s-1)}:\radius) {};
  \draw[->,color=blue, >=latex] ({90-360/(\n+1) * (\s-1)-\margin}:\radius) arc ({90-360/(\n+1) * (\s-1)-\margin}:{90-360/(\n+1) * \s+\margin}:\radius) node[midway,fill=white,scale=1.5] {$a_{\s}$};
}
\foreach \s in {4,...,5}
{
 \node[draw, circle,minimum size = 0.6cm](H\s)  at ({90-360/(\n+1) * (\s-1)}:\radius) {};
  \draw[->,color=red, >=latex] ({90-360/(\n+1) * (\s-1)-\margin}:\radius) arc ({90-360/(\n+1) * (\s-1)-\margin}:{90-360/(\n+1) * \s+\margin}:\radius) node[midway,fill=white,scale=1.5] {$a_{\s}$};
}
\foreach \s in {6,...,8}
{
 \node[draw, circle,minimum size = 0.6cm](H\s)  at ({90-360/(\n+1) * (\s-1)}:\radius) {};
  \draw[->, >=latex] ({90-360/(\n+1) * (\s-1)-\margin}:\radius) arc ({90-360/(\n+1) * (\s-1)-\margin}:{90-360/(\n+1) *\s+\margin}:\radius) node[midway,fill=white,scale=1.5] {$a_{\s}$};
}

\end{tikzpicture}\hspace{1cm}
\begin{tikzpicture}[thick,scale=0.5, every node/.style={transform shape}]

\def \n {7}
\def \radius {4cm}

\def \margin {6} % margin in angles, depends on the radius
\node[scale=1.5](E2) at (0,0) {\Large$\mathcal{E}_{2}$};

\foreach \s in {1,...,1}
{
 \node[draw, circle,minimum size = 0.6cm](H\s)  at ({90-360/(\n+1) * (\s-1)}:\radius) {};
  \draw[->, >=latex] ({90-360/(\n+1) * (\s-1)-\margin}:\radius) arc ({90-360/(\n+1) * (\s-1)-\margin}:{90-360/(\n+1)*\s+\margin}:\radius) node[midway,fill=white,scale=1.5] {$b_{\s}$};
}  
\foreach \s in {2,...,4}
{
 \node[draw, circle,minimum size = 0.6cm](H\s)  at ({90-360/(\n+1) * (\s-1)}:\radius) {};
  \draw[->,color=blue, >=latex] ({90-360/(\n+1) * (\s-1)-\margin}:\radius) arc ({90-360/(\n+1) * (\s-1)-\margin}:{90-360/(\n+1) * \s+\margin}:\radius) node[midway,fill=white,scale=1.5] {$b_{\s}$};
}
\foreach \s in {5,...,7}
{
 \node[draw, circle,minimum size = 0.6cm](H\s)  at ({90-360/(\n+1) * (\s-1)}:\radius) {};
  \draw[->,color=red, >=latex] ({90-360/(\n+1) * (\s-1)-\margin}:\radius) arc ({90-360/(\n+1) * (\s-1)-\margin}:{90-360/(\n+1) * \s+\margin}:\radius) node[midway,fill=white,scale=1.5] {$b_{\s}$};
}
\foreach \s in {7,...,8}
{
 \node[draw, circle,minimum size = 0.6cm](H\s)  at ({90-360/(\n+1) * (\s-1)}:\radius) {};
  \draw[->, >=latex] ({90-360/(\n+1) * (\s-1)-\margin}:\radius) arc ({90-360/(\n+1) * (\s-1)-\margin}:{90-360/(\n+1) *\s+\margin}:\radius) node[midway,fill=white,scale=1.5] {$b_{\s}$};
}  

\end{tikzpicture}
\caption{\label{fi:cycle_npPoA1} The set $\mathcal{E}=\mathcal{E}_{1}\cup\mathcal{E}_{2}$ with two cycles of $n=8$ resources each. 
The strategies $s_{1}$ and $s_{1}'$ for player 1 are shown in blue and red, with $k=2$. For subsequent players these strategies are turned clockwise.
}

\end{figure}

\noindent If all the players choose the blue strategy $s_{i}$, then the $a_{j}$'s have a load $k$ and the $b_{j}$'s a load $k+1$. 
Because $c^{q}(k)=0$, the expected cost for each player $i$ is just $(k+1)\,c^{q}(k+1)$. 
Now, deviating to the red strategy $s_{i}'$ yields a larger cost $k\,c^{q}(k+1)+k\,c^{q}(k+2)$, so that all players choosing $s_{i}$ is an equilibrium with social cost $n\,(k+1)\,c^{q}(k+1)>0$. 
On the other hand, if the prophet assigns $s_{i}'$ to every player, then all the resources have a load  $k$ and the social cost is 0. 
Therefore, this game has $\PPoA(\Gamma^{\boldsymbol{p}})=\infty$ as required.
\end{case}

\begin{case}
$c(k)>0$ for all $c\in\mathcal{C}\setminus\{c_{0}\}$ and $k\geq 1$.

We will use the alternative formula \eqref{eq:bnd-equiv2} for $\gamma_{\mathsf{pr}}(\mathcal{C},q)$ to show that
for each 
$M<\gamma_{\mathsf{pr}}(\mathcal{C},q)$ there exists some game $\Gamma^{\boldsymbol{p}}$ with homogeneous probabilities $p_{i}\equiv q$ such that $\PPoA(\Gamma^{\boldsymbol{p}})>M$. 

We note that for each $c\in\mathcal{C}\setminus\{c_{0}\}$ we have $\Psi(\omega)\geq \beta_{c,1,0}(\omega)=\omega$ for all $\omega\in\mathbb{R}$, whereas for $\omega\in[0,1)$ we have $\Psi(\omega)\geq\lim_{m\to\infty}\beta_{c,1,m}(\omega)=\infty$.
Thus $\Psi(\omega)\geq \gamma_{\mathsf{pr}}(\mathcal{C},q)>M$ for all $\omega\geq 0$. 
It follows that the sets $\{\omega:\beta_{c,k,m}(\omega)>M\}$ with $(c,k,m)\in\mathcal{T}$ are an open cover of the compact interval $[0,M]$. 
Let us extract a finite subcover $\mathcal{F}\subset \mathcal{T}$ and assume,  without loss of generality, that $(c,1,0)\in\mathcal{F}$ for some $c\in\mathcal{C}\setminus\{c_{0}\}$.
Then, the piece-wise affine function
\begin{equation}
\label{eq:Psi-finite}
\Psi_{\mathcal{F}}(\omega)
\eqdef 
\max_{(c,k,m)\in\mathcal{F}}\beta_{c,k,m}(\omega)
\end{equation}
satisfies
$\Psi_{\mathcal{F}}(\omega)>M$ for all $\omega\in[0,M]$, and also
$\Psi_{\mathcal{F}}(\omega)\geq \beta_{c,1,0}(\omega) = \omega$
for all $\omega\geq 0$.
It follows that the minimum 
$\gamma_{\mathsf{pr}}(\mathcal{F})\eqdef\min_{\omega\geq 0}\Psi_{\mathcal{F}}(\omega)$ is strictly larger than $M$ and is attained at some $\bar\omega\geq 0$.
To construct $\Gamma^{\boldsymbol{p}}$ we distinguish two sub-cases.

\medskip

\noindent 
\emph{Sub-case} 2.1: $\bar \omega=0$.

In this case there exists a triple $(c,k,m)\in\mathcal{F}$ such that  
$\beta_{c,k,m}(0)=\gamma_{\mathsf{pr}}(\mathcal{F})$ and $\beta'_{c,k,m}(0)\geq 0$, that is, 
\begin{equation}
\label{eq:double-cond}
\frac{m\,c^{q}(m)}{k\,c(k)}=\gamma_{\mathsf{pr}}(\mathcal{F}) \quad\text{and}\quad
\frac{1}{q}\,k\,c^{q}(1+m)\geq m\,c^{q}(m).    
\end{equation}
 Consider  a rational approximation $\tilde{q}=\zeta/\nu\approx q$ with  $\zeta,\nu\in\mathbb{N}$ and $\tilde{q}\le q$, so that
\begin{equation}\label{eq:case1}
\nu\,k\,c^{q}(1+m)\geq \zeta\,m\,c^{q}(m).
\end{equation}

We build a sequence of \aclp{BCG}  
$\Gamma^{\boldsymbol{p}}_{n}$ with
$n$  players with homogeneous probabilities $p_{i}\equiv q$, such that   $\PPoA(\Gamma^{\boldsymbol{p}}_{n})>M$
for $n$ large.
The resource set is composed of $h=n\, \zeta$ disjoint cycles $\mathcal{E}_{1},\ldots,\mathcal{E}_{h}$ of $n$ resources each with costs $c(\argdot)$ (see Figure~\ref{fi:cycle_PPoA}).

\begin{figure}[ht]
\centering
\begin{tikzpicture}[thick,scale=0.5, every node/.style={transform shape}]

\def \n {7}
\def \radius {4cm}

\def \margin {6} % margin in angles, depends on the radius
\node[scale=1.5](E1)  at (0,0) {\Large$\mathcal{E}_{1}$};

\foreach \s in {1,...,1}
{
 \node[draw, circle,minimum size = 0.6cm](H\s)  at ({90-360/(\n+1) * (\s-1)}:\radius) {};
  \draw[->, >=latex] ({90-360/(\n+1) * (\s-1)-\margin}:\radius) arc ({90-360/(\n+1) * (\s-1)-\margin}:{90-360/(\n+1) * (\s)+\margin}:\radius) node[midway,fill=white,scale=1.5] {$a_{\s}$};
}  
\foreach \s in {2,...,4}
{
 \node[draw, circle,minimum size = 0.6cm](H\s)  at ({90-360/(\n+1) * (\s-1)}:\radius) {};
  \draw[->,color=blue, >=latex] ({90-360/(\n+1) * (\s-1)-\margin}:\radius) arc ({90-360/(\n+1) * (\s-1)-\margin}:{90-360/(\n+1) * (\s)+\margin}:\radius) node[midway,fill=white,scale=1.5] {$a_{\s}$};
}
\foreach \s in {5,...,8}
{
 \node[draw, circle,minimum size = 0.6cm](H\s)  at ({90-360/(\n+1) * (\s-1)}:\radius) {};
  \draw[->,color=black, >=latex] ({90-360/(\n+1) * (\s-1)-\margin}:\radius) arc ({90-360/(\n+1) * (\s-1)-\margin}:{90-360/(\n+1) * (\s)+\margin}:\radius) node[midway,fill=white,scale=1.5] {$a_{\s}$};
}

\end{tikzpicture}\hspace{1cm}
\begin{tikzpicture}[thick,scale=0.5, every node/.style={transform shape}]

\def \n {7}
\def \radius {4cm}

\def \margin {6} % margin in angles, depends on the radius
\node[scale=1.5](E1)  at (0,0) {\Large$\mathcal{E}_{2}$};

\foreach \s in {1,...,1}
{
 \node[draw, circle,minimum size = 0.6cm](H\s)  at ({90-360/(\n+1) * (\s-1)}:\radius) {};
  \draw[->, >=latex] ({90-360/(\n+1) * (\s-1)-\margin}:\radius) arc ({90-360/(\n+1) * (\s-1)-\margin}:{90-360/(\n+1) * (\s)+\margin}:\radius) node[midway,fill=white,scale=1.5] {$a_{\s}$};
}  
\foreach \s in {2,...,4}
{
 \node[draw, circle,minimum size = 0.6cm](H\s)  at ({90-360/(\n+1) * (\s-1)}:\radius) {};
  \draw[->,color=blue, >=latex] ({90-360/(\n+1) * (\s-1)-\margin}:\radius) arc ({90-360/(\n+1) * (\s-1)-\margin}:{90-360/(\n+1) * (\s)+\margin}:\radius) node[midway,fill=white,scale=1.5] {$a_{\s}$};
}
\foreach \s in {5,...,8}
{
 \node[draw, circle,minimum size = 0.6cm](H\s)  at ({90-360/(\n+1) * (\s-1)}:\radius) {};
  \draw[->,color=black, >=latex] ({90-360/(\n+1) * (\s-1)-\margin}:\radius) arc ({90-360/(\n+1) * (\s-1)-\margin}:{90-360/(\n+1) * (\s)+\margin}:\radius) node[midway,fill=white,scale=1.5] {$a_{\s}$};
}

\end{tikzpicture}\hspace{1cm}
\begin{tikzpicture}[thick,scale=0.5, every node/.style={transform shape}]

\def \n {7}
\def \radius {4cm}

\def \margin {6} % margin in angles, depends on the radius
\node[scale=1.5](E1)  at (0,0) {\Large$\mathcal{E}_{3}$};

\foreach \s in {1,...,1}
{
 \node[draw, circle,minimum size = 0.6cm](H\s)  at ({90-360/(\n+1) * (\s-1)}:\radius) {};
  \draw[->, >=latex] ({90-360/(\n+1) * (\s-1)-\margin}:\radius) arc ({90-360/(\n+1) * (\s-1)-\margin}:{90-360/(\n+1) * (\s)+\margin}:\radius) node[midway,fill=white,scale=1.5] {$a_{\s}$};
}  
\foreach \s in {2,...,4}
{
 \node[draw, circle,minimum size = 0.6cm](H\s)  at ({90-360/(\n+1) * (\s-1)}:\radius) {};
  \draw[->,color=blue, >=latex] ({90-360/(\n+1) * (\s-1)-\margin}:\radius) arc ({90-360/(\n+1) * (\s-1)-\margin}:{90-360/(\n+1) * (\s)+\margin}:\radius) node[midway,fill=white,scale=1.5] {$a_{\s}$};
}
\foreach \s in {5,...,8}
{
 \node[draw, circle,minimum size = 0.6cm](H\s)  at ({90-360/(\n+1) * (\s-1)}:\radius) {};
  \draw[->,color=black, >=latex] ({90-360/(\n+1) * (\s-1)-\margin}:\radius) arc ({90-360/(\n+1) * (\s-1)-\margin}:{90-360/(\n+1) * (\s)+\margin}:\radius) node[midway,fill=white,scale=1.5] {$a_{\s}$};
}

\end{tikzpicture}

\caption{\label{fi:cycle_PPoA} The resource set for $h=3$ cycles with $n=8$ resources each. The strategy $s_{1}$ for player 1 with $m=3$ is shown in blue. For subsequent players this strategy is turned clockwise.
}

\end{figure}

Each player $i\in\braces*{1,\ldots,n}$ has one equilibrium strategy and multiple alternative strategies. 
Player $i$'s equilibrium strategy $s_{i}$ picks the resources $a_{i+1},\ldots,a_{i+m}$ 
from each and every cycle (the blue resources in 
Figure~\ref{fi:cycle_PPoA}), with the identification
$i+j\equiv i+j-n$ when $i+j>n$.
The alternative strategies consist of picking an arbitrary set containing $\kappa=n\, \nu\, k$ resources, excluding those in $s_{i}$ (i.e.,  only black resources can be chosen).
If each player plays the strategy $s_{i}$, 
then the load on every resource is $m$ and the expected cost for each player $i$ is 
$h\,m\,c^{q}(m)$,
whereas a unilateral deviation to any of the alternative strategy produces the cost 
$\kappa\,c^{q}(1+m)$. From \eqref{eq:case1}
it follows that the strategy profile in which all players choose $s_{i}$ is an equilibrium,
with expected social cost 
\begin{equation}
\label{eq:SCNE02}
C(\boldsymbol{s})= n\,h\,m\,c^{q}(m).
\end{equation}

Now, the prophet observes the demand $N\sim\Binomial\parens*{n,q}$ and tries to minimize the expected cost by distributing the players as uniformly as possible across the resources using the alternative strategies. 
Recall however that the resources in $s_{i}$ are forbidden in player $i$'s alternative strategies. 
So we consider the following upper bound on the optimal prophet cost. 

Assume that instead of picking $\kappa$ resources, the prophet uses the following greedy procedure to allocate $\tilde\kappa=\kappa+mj$ resources to each player (with $j\in\mathbb{N}$ to be fixed later), including $mj$ redundant  resources that can be dropped later. 
Starting from the first cycle $\mathcal{E}_{1}$ consider sequentially each one of the $N$
players assigning $\tilde\kappa$ contiguous resources, and continuing from there with the next player. Once the resources of a given cycle are exhausted the process jumps to the next cycle, and after reaching the end of the last cycle $\mathcal{E}_{h}$ it jumps back to $\mathcal{E}_{1}$, continuing the process until all players have been assigned $\tilde\kappa$ resources. 

The $\tilde\kappa$ resources allocated to a given player $i$ may include some forbidden resources in $s_{i}$. However, these $\tilde\kappa$ resources span at most $\ceil{\tilde\kappa/n}$ cycles and therefore the number of such forbidden links for $i$ is at most $m\,\ceil{\tilde\kappa/n}$.
If we choose $j$ such that $mj\geq m\,\ceil{\tilde\kappa/n}$, we may then remove $mj$ resources eliminating the forbidden links to obtain feasible strategies for every player $i$. 
This can be accomplished by choosing $j\geq \tilde\kappa/n$,
that is $n\, j\geq \kappa + mj$,
so it suffices to take $j=\ceil{\kappa/(n-m)}$.

The social cost for this feasible 
strategy profile is smaller than the cost of the
greedy allocation including the redundant resources, which then provides an upper bound for the 
optimal cost achievable by the prophet. 
To compute this upper bound we observe that the greedy procedure yields an  average load of $X=N\,\tilde\kappa/n\,h$ on each resource, with some resources having a load $\ceil{X}$ and the others $\floor{X}$. 
More explicitly, because  $N\tilde\kappa= n\,h\,  \floor{X} + v$ with $0\leq v<n\,h$, there will be $v$ resources with a load $\ceil{X}$ and  $(n\,h  -v)$ resources with a load $\floor{X}$. 
Observing that $v=n\,h \, (X-\floor{X})$ and defining 
\begin{align}
\label{eq:Qc}
Q^{c}(x)
&\eqdef (x-\floor{x})\,\ceil{x}\,c(\ceil{x})+(1-x+\floor{x})\,\floor{x} \,c(\floor{x}),     
\end{align}
the corresponding social cost can be expressed as
\begin{equation}
\label{eq:soc-cost-Qc}
v\,\ceil{X}\,c(\ceil{X})+(n\,h-v)\,\floor{X} c(\floor{X})
=n\,h\,Q^{c}(X).    
\end{equation}
Combining this upper bound for the prophet optimal cost with \eqref{eq:SCNE02}, we obtain a lower bound for the \acl{PPoA}, that is,
\begin{equation}
\label{eq:ppoabnd}
\PPoA(\Gamma_{n}^{\boldsymbol{p}})\geq\frac{n\,h\,m\,c^{q}(m)}{n\,h\Expect[Q^{c}(X)]}=\frac{m\,c^{q}(m)}{\Expect[Q^{c}(X)]}.
\end{equation}
Now, for $n\to\infty$ we have that the number of redundant resources $mj=m\ceil{n\,\nu\,k/(n-m)}$ remains bounded, so that $\tilde\kappa/h$ converges to $\nu\,k/\zeta=k/\tilde{q}$,
and therefore
\begin{align*}
X =\frac{N\,\tilde\kappa}{n\, h}\xrightarrow{\text{a.s.}} q\, k/\tilde q.
\end{align*}
Because $X$ is bounded and $Q^{c}(\argdot)$ is continuous, the portmanteau theorem implies that $\Expect[Q^{c}(X)]\to Q^{c}(q\, k / \tilde{q})$ as $n\to\infty$.
By taking $\tilde{q}\to q$, the latter converges to $Q^{c}(k)=k\,c(k)$ and, because 
\begin{equation}
\label{eq:frac>M}
\frac{m\,c^{q}(m)}{kc(k)}=\gamma_{\mathsf{pr}}(\mathcal{F})>M,   
\end{equation}
we may choose $n$ large enough and $\tilde{q}\approx q$ such that the right hand side of \eqref{eq:ppoabnd} is larger than $M$. 
It follows that $\PPoA(\Gamma^{\boldsymbol{p}}_{n})>M$, as was to be proved.

\medskip

\noindent 
\emph{Sub-case} 2.2:
$\bar \omega>0$

From optimality we have $0\in\partial\Psi_{\mathcal{F}}(\bar \omega)$ so we can find two affine functions $c_{1}, c_{2}$ with
$\beta_{c_{1},k_{1},m_{1}}(\bar \omega)=\beta_{c_{2},k_{2},m_{2}}(\bar \omega)=\gamma_{\mathsf{pr}}(\mathcal{F})$ such that $\beta'_{c_{1},k_{1},m_{1}}(\bar \omega)
\leq 0\leq \beta'_{c_{2},k_{2},m_{2}}(\bar \omega)$, that is,
\begin{align}
&\frac{1}{q}\,k_{1}\,c^{q}_{1}(1+m_{1})\bar \omega+m_{1}\, c^{q}_{1}(m_{1})(1-\bar \omega)=\gamma_{\mathsf{pr}}(\mathcal{F})\, k_{1}\, c_{1}(k_{1}),\label{eq:C32}\\
&\frac{1}{q}\,k_{2}\,c^{q}_{2}(1+m_{2})\bar \omega+m_{2}\, c^{q}_{2}(m_{2})(1-\bar \omega)=\gamma_{\mathsf{pr}}(\mathcal{F})\, k_{2}\, c_{2}(k_{2}),\label{eq:C42}\\
&\frac{1}{q}\,k_{1}\,c^{q}_{1}(1+m_{1})-m_{1}\,c^{q}_{1}(m_{1})~~\leq ~~0~~\leq~~ \frac{1}{q}\,k_{2}\,c^{q}_{2}(1+m_{2})-m_{2}\,c^{q}_{2}(m_{2}).\label{eq:C52}
\end{align}
The latter implies that 0 can be expressed as an average of the 
left and right expressions, so there exists $\eta\in[0,1]$ such that
\begin{equation}
\label{eq:peq}
\eta\, m_{1}\,c_{1}^{q}(m_{1})+(1-\eta)\,m_{2}\,c_{2}^{q}(m_{2})=\eta\,\frac{1}{q}\,k_{1}\,c_{1}^{q}(1+m_{1})+ (1-\eta)\,\frac{1}{q}\,k_{2}\,c_{2}^{q}(1+m_{2}).
\end{equation}

We will construct a sequence of \aclp{BCG} with costs $c_{1}(\argdot)$ and $c_{2}(\argdot)$ and $n$ homogeneous players with  probabilities $p_{i}\equiv q$, such that for large $n$ the \acl{PPoA} approaches $\gamma_{\mathsf{pr}}(\mathcal{F})$ and therefore can  be made larger than $M$ as required.

Firstly, we take rational approximations
\begin{equation}
\label{eq:rational-approx}
\tilde\eta=\theta/\tau\approx\eta 
\quad\text{and}\quad \tilde{q}=\zeta/\nu\approx q
\end{equation} 
with $\theta,\tau,\zeta,\nu\in \mathbb{N}$, and $\tilde{q}< q$, in such a way that
in \eqref{eq:peq} we preserve an inequality,
namely 
\begin{equation}\label{eq:tightppoa}
\theta\, \zeta\,m_{1}\,c_{1}^{q}(m_{1}) + (\tau-\theta)\,\zeta\,m_{2}\,c_{2}^{q}(m_{2})\leq \theta\,\nu\,k_{1}\,c_{1}^{q}(1+m_{1}) + (\tau-\theta)\,\nu\,k_{2}\,c_{2}^{q}(1+m_{2}).
\end{equation}

Define 
\begin{equation}
\label{eq:def-approx-bucket}
h_{1}=\theta\, \zeta\,n,
\quad h_{2} = (\tau-\theta)\,\zeta\,n, 
\quad \kappa_{1} = \theta\,\nu\,n\,k_{1},
\quad\text{and}\quad
\kappa_{2}=(\tau-\theta)\,\nu\,n\,k_{2}.
\end{equation}
The resource set is composed of $h=h_{1}+h_{2}$ disjoint cycles $\mathcal{E}_{1},\ldots,\mathcal{E}_{h}$ with $n$ resources $a_{1},\ldots,a_{n}$ each,  as in Figure~\ref{fi:cycle_PPoA}. 
The resources in the cycles $\mathcal{E}_{1},\ldots,\mathcal{E}_{h_{1}}$ all have cost function $c_{1}(\argdot)$, and 
the resources in $\mathcal{E}_{h_{1}+1},\ldots,\mathcal{E}_{h_{1}+h_{2}}$ all have cost function $c_{2}(\argdot)$.

Each player $i\in\braces*{1,\ldots,n}$ has one equilibrium strategy and multiple alternative strategies. 
Player $i$'s equilibrium strategy $s_{i}$ is as follows: from each $\mathcal{E}_{1},\ldots,\mathcal{E}_{h_{1}}$ pick resources $a_{i+1},\ldots,a_{i+m_{1}}$ and from each $\mathcal{E}_{h_{1}+1},\ldots,\mathcal{E}_{h_{1}+h_{2}}$ pick resources $a_{i+1},\ldots,a_{i+m_{2}}$. 
As before, the indices $i+j$ are interpreted modulo $n$, so that 
$i+j\equiv i+j -n$ when $i+j>n$.
Each of the alternative strategies consists of picking an arbitrary set of $\kappa_{1}$ resources from $\mathcal{E}_{1}\cup\cdots\cup\mathcal{E}_{h_{1}}$ and $\kappa_{2}$ resources from $\mathcal{E}_{h_{1}+1}\cup\cdots\cup\mathcal{E}_{h_{1}+h_{2}}$, excluding  the resources that are part of $s_{i}$.

If every player plays the strategy $s_{i}$, the expected cost for each player $i$ is 
$h_{1}\,m_{1}\,c_{1}^{q}(m_{1})+h_{2}\,m_{2}\,c_{2}^{q}(m_{2})$,
whereas a unilateral deviation to any of the alternative strategies produces the cost 
$\kappa_{1}\,c_{1}^{q}(1+m_{1})+\kappa_{2}\,c_{2}^{q}(1+m_{2})$. It follows from \eqref{eq:tightppoa}
that the profile where all players choose $s_{i}$ is an equilibrium with corresponding expected social cost
\begin{equation}
\label{eq:SCNE03}
C(\boldsymbol{s}) = n\,h_{1}\,m_{1}\,c_{1}^{q}(m_{1}) + n\,h_{2}\,m_{2}\,c_{2}^{q}(m_{2}).
\end{equation}

Consider next the cost for the prophet when dealing with $N\sim\Binomial\parens*{n,q}$ players. 
By applying the same greedy procedure as in sub-case 2.1, separately for the first $h_{1}$ cycles and the second $h_{2}$ cycles, we obtain an upper bound on the prophet optimal costs. 
That is, in the first $h_{1}$ cycles the greedy procedure assigns $\tilde\kappa_{1}=\kappa_{1}+m_{1}\,\lceil\kappa_{1}/(n-m_{1})\rceil$ resources per player, and in the second $h_{2}$ cycle it assigns $\tilde\kappa_{2}=\kappa_{2}+m_{2}\,\lceil\kappa_{2}/(n-m_{2})\rceil$ resources per player. 
Repeating the
analysis of sub-case 2.1 we obtain an upper bound for the optimal prophet cost of
\begin{equation}
\label{eq:bound-opt-proph-cost}  
n\,h_{1}\, Q^{c_{1}}(X_{1})+n\,h_{2}\, Q^{c_{2}}(X_{2}),
\end{equation}
where $X_{1}=N\tilde\kappa_{1}/n h_{1}$
and $X_{2}=N\tilde\kappa_{2}/n h_{2}$.
Hence 
\begin{equation}
\label{eq:fin}
\PPoA(\Gamma^{\boldsymbol{p}}_{n})
\geq\frac{n\,h_{1}\,m_{1}\,c_{1}^{q}(m_{1})+n\,h_{2}\,m_{2}\,c_{2}^{q}(m_{2})}{n\,h_{1}\, \Expect[Q^{c_{1}}(X_{1})]+n\,h_{2}\, \Expect[Q^{c_{2}}(X_{2})]}.
\end{equation}
When $n\to\infty$ we have 
$X_{1}\xrightarrow{\text{a.s.}} q\, k_{1}/\tilde{q}$ and $X_{2}\xrightarrow{\text{a.s.}} q\, k_{2}/\tilde q$,
so that, using the portmanteau theorem, 
the right hand side in \eqref{eq:fin} 
converges to   
\begin{equation}
\label{eq:rhs-converges}
\frac{\theta\,m_{1}\,c_{1}^{q}(m_{1})+(\tau-\theta)\,m_{2}\,c_{2}^{q}(m_{2})}{\theta\, Q^{c_{1}}(q\, k_{1}/\tilde{q}) + (\tau-\theta)\, Q^{c_{2}}(q\, k_{2}/\tilde{q})}
 = \frac{\tilde\eta\,m_{1}\,c_{1}^{q}(m_{1}) + (1-\tilde\eta)\,m_{2}\,c_{2}^{q}(m_{2})}{\tilde\eta\, Q^{c_{1}}(q\, k_{1}/\tilde{q}) + (1-\tilde\eta)\, Q^{c_{2}}(q\, k_{2}/\tilde{q})}.    
\end{equation}
Moreover, letting $\tilde{q} \to q$ and 
$\tilde\eta\to\eta$, the latter quotient converges
to 
\begin{equation}
\label{eq:quotient-converges}
\frac{\eta\,m_{1}\,c_{1}^{q}(m_{1})+(1-\eta)\,m_{2}\,c_{2}^{q}(m_{2})}{\eta\, Q^{c_{1}}(k_{1})+(1-\eta)\, Q^{c_{2}}(k_{2})}=\frac{\eta\,m_{1}\,c_{1}^{q}(m_{1})+(1-\eta)\,m_{2}\,c_{2}^{q}(m_{2})}{\eta\, k_{1}\,c_{1}(k_{1})+(1-\eta)\, k_{2}\,c_{2}(k_{2})}=\gamma_{\mathsf{pr}}(\mathcal{F}),   
\end{equation}
where the last equality follows by averaging \eqref{eq:C32} and \eqref{eq:C42} with weights $\eta$ and $(1-\eta)$ respectively, and using \eqref{eq:peq}.
Finally, because $\gamma_{\mathsf{pr}}(\mathcal{F})>M$, we may choose $\tilde{q}\approx q$, 
$\tilde\eta\approx\eta$, and $n$ large, so that the right hand side in \eqref{eq:fin} is strictly larger than $M$. 
This completes the proof.
\Halmos
\end{case}
\endproof

\proof{Proof of Proposition~\ref{pr:growth-cond}.}
By arguing as in the proof of Proposition~\ref{pr:different-p-i-pro}, it suffices to show that the $(\lambda,\mu,q)$-smoothness condition \eqref{eq:ccp} holds with  
$\mu=0$ and some $\lambda=\lambda(q)$ such that
$\lambda(q)\to 1 $ when $q\to 0$, namely 
\begin{equation}
\label{eq:auxxx}
\frac{1}{q} k \,c^{q}(1+m)  \leq  \lambda\, k\, c(k)\quad \forall k\geq 1, m\geq 0,\;\forall c\in\mathcal{C}.
\end{equation}
Clearly, this is equivalent to $c^{q}(1+m)  \leq  \lambda\,q\, c(k)$, and it is most restrictive when $k=1$. 
On the other hand, because the number of players $n$ is fixed, in any strategy profile the load of any given resource is at most $n$. 
Thus, one can always modify the costs $c(k)$ to be constant for $k\geq n$, without affecting the equilibria nor the social costs. 
Therefore it suffices to have the previous inequality for $k=1$ and $m =n$.
Now, taking $X\sim\Binomial(n,q)$ we have
\begin{align*}
c^{q}(1+n)
= q\Expect[c(X)]
&= q \Prob(X=0)\,c(1) +q\Prob(X>0)\Expect[c(X) \mid X>0]\\
&\le q \Prob(X=0)\,c(1) +q\Prob(X>0) H c(1)\\
&= q \,c(1)\left[\Prob(X=0) + (1-\Prob(X=0)) H \right],
\end{align*}
and because $\Prob(X=0)=(1-q)^{n}$, \eqref{eq:auxxx} holds with  $\lambda(q)=(1-q)^{n}+(1-(1-q)^{n})H\to 1$.
\Halmos
\endproof

\subsection*{Proofs of Section~\ref{se:poa_and_pos}}
\label{Ap:poa_and_pos}
Theorem~\ref{th:upper} is a direct consequence 
of Theorem~\ref{th:different-p-i} and the following Proposition~\ref{pr:lambda-mu}, which determines the optimal parameters $(\lambda,\mu)$ for which the class of games  $\mathcal{G}(\mathcal{C}_{\mathsf{aff}}^{q})$ is $(\lambda,\mu)$-smooth.

\begin{proposition}
\label{pr:lambda-mu} 
Let $\bar{q}_{0}=1/4$ and let $\bar{q}_{1}\sim 0.3774$ be the real root of $8 q^{3}+4q^{2}=1$.
The optimal parameters $(\lambda,\mu)$
that attain the tight bounds $\gamma(\mathcal{C}_{\mathsf{aff}}^{q})$ are given by
\begin{equation}\label{eq:para}
(\lambda,\mu)=
\begin{cases}
  \left(1,\dfrac{1}{4}\right)&\text{ if }\ 0< q\leq \bar{q}_{0},\\
  \left(\dfrac{1+q+\sqrt{q(2+q)}}{2},\dfrac{1+q-\sqrt{q(2+q)}}{2}\right)&\text{ if }\ \bar{q}_{0}\leq q\leq \bar{q}_{1},\\
\left(\dfrac{1+2q+2q^{2}}{1+2q},\dfrac{q}{1+2q}\right)&\text{ if }\ \bar{q}_{1}\leq q\leq 1.
\end{cases}
\end{equation}
\end{proposition}

We split the proof of Proposition~\ref{pr:lambda-mu} into three technical lemmas and three propositions, each one dealing with one of the three subintervals of $[0,1]$ determined by $\bar{q}_{0}$ and $\bar{q}_{1}$. 

A sketch of the proof goes as follows. For each fixed $q$ we proceed to minimize $\lambda/(1-\mu)$ over all
 $(\lambda,\mu)$-smoothness parameters satisfying \eqref{eq:r1}. The latter is simplified as in condition \eqref{eq:tightPoAaff} below, which still
provides the optimal parameters for the tight bounds $\gamma(\mathcal{C}_{\mathsf{aff}}^{q})$. 
We next reduce the optimization over $(\lambda,\mu)$ to the minimization of a one dimensional convex function $\psi_{q}(y)$ over the region $y\geq0$.
This auxiliary function $\psi_{q}(\argdot)$ is an upper envelope of a countable family of affine functions, and for each $q$ it has a minimizer $y_{q}$, which takes different values, depending on where $q$ is located with respect to $\bar{q}_{0}$ and $\bar{q}_{1}$. 
This optimal solution yields the three alternative expressions for $\gamma(\mathcal{C}_{\mathsf{aff}}^{q})$, with the corresponding optimal smoothness parameters $(\lambda,\mu)$.

Our starting point is the following simple observation.

\begin{lemma}
\label{le:lambda-mu}
A pair $(\lambda,\mu)$ with $\lambda\ge 0$ and $\mu\in[0,1)$ satisfies \eqref{eq:r1} 
for the class $\mathcal{C}_{\mathsf{aff}}$ iff
\begin{equation}\label{eq:lambda-mu}
k(1+m q)\leq \lambda\, k(1-q+q k)+\mu\, m(1-q+q m)\quad\forall k,m\in\mathbb{N}.
\end{equation}
\end{lemma}

\proof{Proof.}
For an affine function $c(x)=a\, x+b$, we have $c^{q}(x)= q[a\,(1-q+q\, x)+b]$.
It follows that \eqref{eq:lambda-mu} is just the special case of \eqref{eq:r1} with $a=1$ and $b=0$. Conversely, 
starting from \eqref{eq:lambda-mu} and taking $k=1$ and $m=0$ we get $\lambda\geq 1$, so that multiplying by $a\geq 0$ and adding $kb\geq 0$ on both sides we readily get \eqref{eq:r1} for $c^{q}(\argdot)$.
\Halmos
\endproof

From Lemma~\ref{le:lambda-mu} it follows that
\begin{equation}
\label{eq:tightPoAaff}
\gamma(\mathcal{C}_{\mathsf{aff}}^{q})=\inf\{\lambda/(1-\mu): (\lambda,\mu) \text{ satisfies }\eqref{eq:lambda-mu}\}
\end{equation}
which can be reduced to a one-dimensional problem. Indeed, condition \eqref{eq:lambda-mu} is trivially satisfied for $k=0$
so we may restrict to the set $\mathcal{P}$ of all pairs 
$(k,m)\in\mathbb{N}^{2}$ with $k\geq 1$.   
Then, for any given  $\mu\in [0,1)$, the smallest possible value of $\lambda$ compatible with \eqref{eq:lambda-mu} is 
\begin{align*}
\lambda
&=\sup_{(k,m)\in\mathcal{P}}\frac{k(1+q m)- \mu m(1-q+q m)}{k(1-q+q k)}\\
&=\sup_{(k,m)\in\mathcal{P}}\mu\bracks*{\frac{k(1+q m)- m(1-q+q m)}{k(1-q+q k)}} +(1-\mu)\frac{1+q m}{1-q+q k}\,,
\end{align*}
from which it follows that
\begin{equation*}
\gamma(\mathcal{C}_{\mathsf{aff}}^{q})=\inf_{\lambda,\mu}\frac{\lambda}{1-\mu}=\inf_{\mu\in [0,1)}\sup_{(k,m)\in\mathcal{P}}\frac{\mu}{1-\mu}\left[\frac{k(1+q m)-m(1-q+q m)}{k(1-q+q k)}\right]+\frac{1+q m}{1-q+q k}\,.
\end{equation*}
Defining 
\begin{equation}\label{eq:y}
y\eqdef\frac{\mu}{1-\mu}\in [0,\infty)
\end{equation}
and introducing the functions 
\begin{align}
\label{eq:vrphikm}
\psi^{k,m}_{q}(y)&=y\left[\frac{k(1+q m)-m(1-q+q m)}{k(1-q+q k)}\right]+\frac{1+q m}{1-q+q k},\\
\label{eq:vrphi}
\psi_{q}(y)&=\sup_{(k,m)\in\mathcal{P}} \psi^{k,m}_{q}(y),
\end{align}
we obtain the following equivalent expression for the optimal bound in \eqref{eq:tightPoAaff}:
\begin{equation}\label{eq:B(p)}
\gamma(\mathcal{C}_{\mathsf{aff}}^{q})=\inf_{y\geq 0} \psi_{q}(y).
\end{equation}
If this infimum is attained at a certain $y_{q}$, then we get $\gamma(\mathcal{C}_{\mathsf{aff}}^{q})=\psi_{q}(y_{q})$ together with the corresponding optimal parameters 
\[
\mu=\frac{y_{q}}{1+y_{q}}\quad \text{and} \quad \lambda=(1-\mu)\,\gamma(\mathcal{C}_{\mathsf{aff}}^{q})=\frac{\psi_{q}(y_{q})}{1+y_{q}}.
\]

To proceed, we need the auxiliary function $\psi_{\infty}(y)$ defined in the next lemma.
\begin{lemma}
\label{le:L12} 
For all $q>0$ and $y>0$ the following limit is well defined and does not depend on $q$
\begin{equation}\label{eq:varphi-infty}
\psi_{\infty}(y)=\lim_{k\to\infty}\sup_{m\geq 0} \psi^{k,m}_{q}(y)=\frac{(y+1)^{2}}{4y}.
\end{equation}
This function is strictly decreasing  for $y\in (0,1)$ and strictly increasing for $y\in (1,\infty)$.
\end{lemma}

\proof{Proof.}
Fix $y>0$ and $q\in (0,1]$. The maximum of $\psi^{k,m}_{q}(y)$ for $m\in\mathbb{N}$ is attained at the integer $\overline{m}$
that is closest to the unconstrained (real) maximizer (because $\psi^{k,m}_{q}(y)$ is quadratic in $m$)
\begin{equation}\label{eq:mhat}
\widehat{m} = \frac{(y+1)q k-y(1-q)}{2yq}\,.
\end{equation}
For a  large $k$, we have $\widehat{m}\geq 0$ and we may find $f\in (-\frac{1}{2},\frac{1}{2}]$ such that $\widehat{m} =\overline{m}+f$.
Then,
\begin{align*}
\sup_{m\in\mathbb{N}}\; ((y+1)q k-y(1-q))m -yq m^{2} &= yq(2\widehat{m} - \overline{m})\overline{m}\\
&=yq(\widehat{m}+f)(\widehat{m}-f)\\
&=yq(\widehat{m}^{2}-f^{2})\\
&=\frac{((y+1)q k-y(1-q))^{2}}{4yq}-yq f^{2},
\end{align*}
from which it follows that 
\[
\psi_\infty(y)=\lim_{k\to\infty}\dfrac{(y+1)k+\dfrac{1}{4yq}((y+1)q k-y(1-q))^{2}-yq f^{2}}{k(1-q+q k)}=\frac{(y+1)^{2}}{4y}\,.
\]
The monotonicity claims follow at once by computing the derivative $\psi'_\infty(y)=(y^{2}-1)/(4y^{2})$.
\Halmos
\endproof

The following lemma gathers some basic facts about the function $\psi_{q}:[0,\infty)\to\mathbb{R}$ and
shows in particular that its infimum is attained.

\begin{lemma}
\label{le:phip}
For each $q>0$ the function $\psi_{q}(\argdot)$ is convex and finite over $(0,\infty)$,
with $\psi_{q}(y)\to\infty$ both when $y\to 0$ 
and $y\to\infty$. In particular, the minimum of $\psi_{q}(\argdot)$ is attained 
at a point $y_{q}>0$.
\end{lemma}

\proof{Proof.}
Convexity is obvious as $\psi_{q}(\argdot)$ is a supremum of affine functions.
The infinite limits at 0 and $\infty$ follow by noting that $\psi_{q}(y)\geq\psi_{\infty}(y)$ for $y>0$,
together with the fact that $\psi_{q}(0)=\infty$ which results from letting $m\to\infty$ in
the inequality  $\psi_{q}(0)\geq \psi^{1,m}_{q}(0)=1+q m\to \infty$.

To show that $\psi_{q}(y)<\infty$ for   $y\in(0,\infty)$, we rewrite the expression of $\psi_{q}(y)$ as
\begin{equation}\label{eq:alt}
\psi_{q}(y)=\sup_{k\geq 1} \frac{1}{k(1-q+q k)}\left[(y+1)k+\sup_{m\geq 0}\left[((y+1)q k-y (1-q))m -y q m^{2}\right]\right].
\end{equation}
Relaxing the inner supremum and considering the maximum with $m\in\mathbb{R}$ we get
\[
\psi_{q}(y)\leq \sup_{k\geq 1} \frac{1}{k(1-q+q k)}\left[(y+1)k+\frac{((y+1)q k-y (1-q))^{2}}{4y q}\right].
\]
The latter is a quotient of two quadratics in $k$ so it remains bounded and the supremum is finite.

Because $\psi_{q}(\argdot)$ is convex and finite on $(0,\infty)$, it is continuous. Moreover, because it goes to $\infty$
at 0 and $\infty$, it is inf-compact  and therefore its minimum is attained.
\Halmos
\endproof

Our next step is to find the exact expression for the optimal solution $y_{q}$ for all $q\in[0,1]$.
We will show that, for $q$ large, the minimum of $\psi_{q}(\argdot)$ is attained at a point $y_{q}$ for which the supremum in \eqref{eq:vrphi} is reached with $k=1$ and simultaneously for $m=1$ and $m=2$, that is, 
\begin{equation*}
\psi_{q}(y_{q})=\psi_{q}^{1,1}(y_{q})=\psi_{q}^{1,2}(y_{q}).
\end{equation*} 
For smaller values of $q$ the supremum is still reached at $k=1$
with either $m=1$ or $m=0$, but also for $k$ and $m$ tending to $\infty$. 
This suggests to consider the solutions of the equations
\begin{align} 
\label{eq:y-0-p}
\psi_\infty(y)=\psi_{q}^{1,0}(y)&\quad\iff\quad y=y_{0,q}\eqdef 1/3\\
\label{eq:y-1-p}
\psi_\infty(y)=\psi_{q}^{1,1}(y)&\quad\iff\quad y=y_{1,q}\eqdef\frac{1}{1+2q+ 2\sqrt{q(2+q)}}\\
\label{eq:y-2-p}
\psi_{q}^{1,1}(y)=\psi_{q}^{1,2}(y)&\quad\iff\quad y=y_{2,q}\eqdef\frac{q}{1+q}.
\end{align}
Note that these three solutions belong to $(0,1)$. Let also $\bar{q}_{0}=1/4$ be the point at which $y_{0,q}=y_{1,q}$, and $\bar{q}_{1}\sim 0.3774$ the point where $y_{1,q}=y_{2,q}$ which is the unique real root of $8 q^{3}+4q^{2}=1$.

\begin{proposition}
\label{pr:low}
The minimum of $\psi_{q}(\argdot)$ is attained at $y_{0,q}$ if and only if $q\in (0,\bar q_{0}]$.
\end{proposition}

\proof{Proof.}
We will prove that
\begin{equation}\label{eq:igualdad0}
\psi_{q}(y_{0,q})=\psi_{\infty}(y_{0,q})=\psi_{q}^{1,0}(y_{0,q})\quad\text{iff}\quad q\leq\bar q_{0}.
\end{equation}
Assuming this, because both $\psi_{q}^{1,0}(\argdot)$ and $\psi_{\infty}(\argdot)$ are minorants of $\psi_{q}(\argdot)$, their slopes $(\psi_{q}^{1,0})'(y_{0,q})=1$ and $\psi_{\infty}'(y_{0,q})=-2$ 
are subgradients of $\psi_{q}(\argdot)$ at $y_{0,q}$. Hence 
$0\in[-2,1]\subseteq \partial\psi_{q}(y_{0,q})$ and $y_{0,q}$ is indeed a minimizer, as claimed.

To prove \eqref{eq:igualdad0}, we observe that the second part of this equality stems from the definition of $y_{0,q}$ in \eqref{eq:y-0-p}.
To establish the first equality, we note that $\psi_{\infty}(\argdot)\leq \psi_{q}(\argdot)$, so it suffices to show that 
$\psi_{q}(y_{0,q})\leq\psi_{\infty}(y_{0,q})$, which is equivalent to
\begin{equation}\label{eq:D0}
y_{0,q}\left[\frac{k(1+q m)-m(1-q+q m)}{k(1-q+q k)}\right]+\frac{1+q m}{1-q+q k}\leq\frac{(1+y_{0,q})^{2}}{4y_{0,q}}\quad\forall (k,m)\in\mathcal{P}\quad\text{iff }q\leq \bar q_{0}.
\end{equation}

Substituting $y_{0,q}=1/3$, the left inequality can be written equivalently as 
\[
0\leq q(2k-m-1)^{2}+m(1-3q)-q\quad\forall (k,m)\in\mathcal{P}.
\]
This holds trivially for $m=0$ so we just consider $m\geq 1$. 
Now, for $k=m=1$ this requires $q\leq 1/4$. Conversely, if $q\leq 1/4$ we have
$1-3q>0$ and therefore $m(1-3q)$ increases with $m$ so that 
\begin{equation*}
q(2k-m-1)^{2}+m(1-3q)-q\geq m(1-3q)-q\geq 1-4q\geq 0. 
\Halmos
\end{equation*}
\endproof

\begin{proposition} 
\label{pr:mid}
The minimum of $\psi_{q}(\argdot)$ is attained at $y_{1,q}$ if and only if $q\in[ \bar{q}_{0},\bar{q}_{1}]$.
\end{proposition}

\proof{Proof.}
We will prove that  
\begin{equation}\label{eq:igualdad}
\psi_{q}(y_{1,q})=\psi_{\infty}(y_{1,q})=\psi_{q}^{1,1}(y_{1,q})\quad\text{iff}\quad q\in [\bar q_{0},\bar q_{1}]. 
\end{equation}
Assuming this, it follows that 
\begin{equation}\label{eq:subgrad}
(\psi_{q}^{1,1})'(y_{1,q})= q \quad\text{and}\quad \alpha\eqdef\psi_{\infty}'(y_{1,q})
\end{equation} 
are subgradients of $\psi_{q}(\argdot)$ at $y_{1,q}$. Now, because $y_{1,q}<1$, by Lemma~\ref{le:L12}, we have 
$\alpha<0$ so that
$0\in[\alpha,q]\subseteq \partial\psi_{q}(y_{1,q})$ and therefore $y_{1,q}$ is a minimizer.

To prove \eqref{eq:igualdad}, we observe that the second equality stems from the definition of $y_{1,q}$ in \eqref{eq:y-1-p}. 
To establish the first equality, we note that $\psi_{\infty}(\argdot)\leq \psi_{q}(\argdot)$, so it suffices to show that 
$\psi_{q}(y_{1,q})\leq\psi_{\infty}(y_{1,q})$, which is equivalent to
\begin{equation}\label{eq:D1}
y_{1,q}\left[\frac{k(1+q m)-m(1-q+q m)}{k(1-q+q k)}\right]+\frac{1+q m}{1-q+q k}\leq\frac{(1+y_{1,q})^{2}}{4y_{1,q}}\quad\forall (k,m)\in\mathcal{P},\quad\text{iff }q\in [\bar q_{0},\bar q_{1}].
\end{equation}
Dividing by $y_{1,q}$ and letting 
\begin{equation}\label{eq:z}
z=\frac{1+y_{1,q}}{2y_{1,q}}=1+q+\sqrt{q(2+q)},
\end{equation}
the left inequality becomes 
\[
\left[\frac{k(1+q m)-m(1-q+q m)}{k(1-q+q k)}\right]+\frac{1+q m}{1-q+q k}(2z-1)\leq z^{2}.
\]
Multiplying by $k(1-q+q k)$ and factorizing, this can be rewritten as
\begin{equation}\label{eq:dd0}
Q_{q}(k,m)\eqdef
q\left(z k-m+\frac{(1-q) z-2}{2q}\right)^{2}+((1-q) z-1-q)m-\frac{((1-q) z-2)^{2}}{4q} 
\ge 0,
\end{equation}
so that, the left inequality of \eqref{eq:D1} is equivalent to
$Q_{q}(k,m)\geq 0$ for all $(k,m)\in\mathcal{P}$.
We observe that
\begin{alignat}{5}
\label{eq:nnn}
Q_{q}(1,0)\geq 0 & \iff & z\geq 2 & \iff & q\geq \bar q_{0}\\
\label{eq:nnnn}
Q_{q}(1,2)\geq 0 & \iff & 8q^{3}+4q^{2}-1\leq 0 & \iff & q\leq \bar q_{1}
\end{alignat}
so that $q\in[\bar q_{0},\bar q_{1}]$ is a necessary condition for \eqref{eq:D1}.
We now show that it is also sufficient.

\begin{case}
$m=0$: The inequality $Q_{q}(k,0)\geq 0$ is equivalent to $z(1-q+q k)\geq 2$ so that the most stringent
condition is for $k=1$, which holds for all $q\geq \bar q_{0}$, as already noted in \eqref{eq:nnn}.

\end{case}

\begin{case}
$m=1$:
From the very definition of $y_{1,q}$ we have that \eqref{eq:D1} holds with equality for $(k,m)=(1,1)$, 
so that $Q_{q}(1,1)=0$. 
Because $Q_{q}(k,1)$ is quadratic in $k$, in order to have $Q_{q}(k,1)\geq 0$ 
for all $k\geq 1$, it suffices to check that $Q_{q}(2,1)\geq 0$. The latter can be factorized as
\[
Q_{q}(2,1)=2(1+q)z(z-2)+1,
\]
so that, substituting $z$ and simplifying, the resulting inequality becomes
\[
4q(1+q)\sqrt{q(2+q)}+4q^{3}+8q^{2}+2q-1\geq 0.
\]
The conclusion follows because this expression increases with $q$ and the inequality holds for $q=1/4$.
\end{case}
 
\begin{case}
$m=2$: As noted in \eqref{eq:nnnn} we have $Q_{q}(1,2)\geq 0$ for all $q\leq\bar q_{1}$. On the other hand, because $z>1$
we have that $Q_{q}(k,2)$ increases for $k\geq 2$, so that it suffices to show that $Q_{q}(2,2)\geq 0$. Now, $Q_{q}(2,2)$ can be factorized as
\[
Q_{q}(2,2)=2(1+q)(z-1)^{2}-4q z
\]
and substituting $z$ we get
\[
Q_{q}(2,2)=4q^{2}(1+q+\sqrt{q(2+q)})\geq 0.
\] 
\end{case}

\begin{case}
$m\geq 3$:
Let $\alpha=(1-q)z-1-q$ be the slope of the linear term in $Q_{q}(k,m)$.
Neglecting the quadratic part we have
\begin{equation}\label{eq:ooo1}
Q_{q}(k,m)\geq \alpha\, m -\frac{((1-q)z-2)^{2}}{4q}
\end{equation}
and therefore it suffices to show that the latter linear expression is nonnegative.
We claim that for all $q\leq\bar q_{1}$ we have $\alpha\geq 0$. 
Indeed, substituting $z$ we get
\[
\alpha=(1-q)\sqrt{q(2+q)}-q(1+q),
\]
so that $\alpha\geq 0$ if and only if $(1-q)^{2}(2+q)\geq q(1+q)^{2}$ which simplifies as
$q^{2}+2q\leq 1$ and holds for $q\leq\sqrt{2}-1$, and in particular for $q\leq \bar q_{1}$.
Thus, the right hand side in \eqref{eq:ooo1} increases with $m$, so what remains to be shown is that it is nonnegative for $m=3$.
The latter amounts to 
\[
3\,\alpha\geq \frac{((1-q)z-2)^{2}}{4q}\,,
\] 
which is equivalent to
\[
2(6q+1+q^{2})(1-q)\sqrt{q(2+q)}\geq 1+2q+11q^{2}+12q^{3}+2q^4
\]
and can be seen to hold for all $q\in [\bar q_{0},\bar q_{1}]$. 
\Halmos
\end{case}
\endproof

\begin{proposition}
\label{pr:hig}
The minimum of $\psi_{q}(\argdot)$ is attained at $y_{2,q}$ if and only if $q\in[ \bar{q}_{1},1]$.
\end{proposition}

\proof{Proof.}
For $y=y_{2,q}$ and $k=1$ the unconstrained maximizer in \eqref{eq:mhat} is $\widehat{m} =3/2$ so that the supremum $\sup_{m\geq 0}\psi^{1,m}_{q}(y_{2,q})$ is attained at $m=1$ and $m=2$.
The slopes of the corresponding terms are
\[
(\psi^{1,m}_{q})'(y)=
\begin{dcases}
q &\text{if } m=1,\\ 
-1&\text{if } m=2.
\end{dcases}
\]
If the outer supremum $\sup_{k\geq 1}$ in \eqref{eq:alt} is attained for $k=1$ it follows that $0\in [-1,q]\subseteq \partial\psi_{q}(y_{2,q})$ and, as a consequence, $y_{2,q}$ is a minimizer.

Considering the expression in \eqref{eq:alt}, and substituting the value of $y_{2,q}$ and using the fact that for $k=1$ the $\sup_{m\geq 0}\psi^{1,m}_{q}(y_{2,q})$ is attained at $m=1$, it follows that $\sup_{k\geq 1}$ is attained at $k=1$ if and only if
\begin{equation}\label{eq:cond}
(1+2q)k+\sup_{m\in\mathbb{N}}\left[((1+2q)k-(1-q))q m -q^{2}m^{2}\right]\leq [(1+q)^{2}+q^{2}]k(1-q+q k)\quad\forall\;k\geq 2.
\end{equation}
We claim that this holds if and only if $q\in[ \bar{q}_{1},1]$.
To this end, we note that for all $k\geq 1$ the unconstrained maximum of the quadratic  $((1+2q)k-(1-q))q m -q^{2}m^{2}$ is attained at
\[
\widehat{m} = \frac{(1+2q)k-(1-q)}{2q}>1.
\]
Proceeding as in the proof of Lemma~\ref{le:L12} we may find an integer $\overline{m}\geq 1$ and $f\in (-\frac{1}{2},\frac{1}{2}]$ such that $\widehat{m}=\overline{m}+f$. Hence, the supremum for $m\in\mathbb{N}$ is attained at $\overline{m}$ and
\begin{equation}\label{eq:solfrac}
\sup_{m\in\mathbb{N}}[((1+2q)k-(1-q))q m -q^{2}m^{2}]= q^{2}(\widehat{m}^{2}-f^{2})
=\frac{1}{4}((1+2q)k-(1-q))^{2}-q^{2}f^{2}.
\end{equation}
Replacing this expression into \eqref{eq:cond} and after simplification, the condition becomes
\begin{equation}\label{eq:cond1}
0 \leq [8q^{3}+4q^{2}-1]k^{2}+[2-2q-4q^{2}-8q^{3}]k+4q^{2}f^{2}-(1-q)^{2}\quad\forall\;k\geq 2.
\end{equation}
It follows that a necessary condition is $8q^{3}+4q^{2}-1\geq 0$ which amounts to $q\geq\bar{q}_{1}$. 
It remains to be shown that, once $q\geq\bar{q}_{1}$, the inequality \eqref{eq:cond1} holds automatically.
Consider first the case $k\geq 3$. 
Ignore the nonnegative term $4q^{2}f^{2}$ and define 
\begin{equation}\label{eq:Q}
Q(x)=[8q^{3}+4q^{2}-1]x^{2}+[2-2q-4q^{2}-8q^{3}]x-(1-q)^{2}.
\end{equation}
For $q\geq\bar{q}_{1}$ this is quadratic and convex in $x$ and we have 
\[
Q'(3)=40q^{3}+20q^{2}-2q-4\geq 0\quad \forall\,q\geq\bar{q}_{1}.
\]
Hence $Q(x)$ is increasing for $x\in [3,\infty)$ and then \eqref{eq:cond1} holds for all $k\geq 3$
because
\[
Q(x)\geq Q(3)=48q^{3}+23q^{2}-4q-4\geq 0\quad \forall\,q\geq\bar{q}_{1}.
\]
For $k=2$ it is not always the case that $Q(2)\geq 0$ so we must consider also the role of the
fractional residual $4q^{2}y^{2}$. The inequality to be proved is
\[
2(1+2q)+\sup_{m\in\mathbb{N}}\; (1+5q)q m -q^{2}m^{2}\leq [(1+q)^{2}+q^{2}]2(1+q).
\]
The supremum for $m\in\mathbb{N}$ is attained at the integer closest to
\[
\widehat{m} = \frac{1+5q}{2q}=2+\frac{1}{2}+\frac{1}{2q}\,,
\]
which can be either $m=3$ or $m=4$ depending on whether $q$ is larger or smaller than $1/2$.
Now, for these values of $m$, the inequalities to be checked are
\begin{align*}
2(1+2q)+ (1+5q)3q-9q^{2}&\leq [(1+q)^{2}+q^{2}]2(1+q),\\
2(1+2q)+ (1+5q)4q-16q^{2}&\leq [(1+q)^{2}+q^{2}]2(1+q),
\end{align*}
which reduce, respectively, to 
\begin{align*}
0&\leq 4q^{3}+2q^{2}-q,\\
0&\leq 4q^{3}+4q^{2}-2q,
\end{align*}
and are easily seen to hold for all $q\in [\bar{q}_{1},1]$.
\Halmos
\endproof

With all the previous ingredients the proof of our main result is straightforward.

\proof{Proof of Proposition~\ref{pr:lambda-mu}.}
Substituting the expressions for the optimal solution $y_{q}$ derived in Propositions~\ref{pr:low},  \ref{pr:mid}, and \ref{pr:hig} we 
get the optimal bound
$\gamma(\mathcal{C}_{\mathsf{aff}}^{q})=\psi_{q}(y_{q})$ which gives \eqref{eq:PoAbound}, as well as the optimal parameters
\[
\mu=\frac{y_{q}}{1+y_{q}}\quad\text{and}\quad \lambda=\frac{\psi_{q}(y_{q})}{1+y_{q}},
\]
which are shown in \eqref{eq:para}.
\Halmos
\endproof

\proof{Proof of Proposition~\ref{pr:lm-linear}.}
We observe that for the purely linear cost $c_{e}(x)=x$, condition \eqref{eq:ccp} reduces to
\begin{equation}\label{eq:smooth}
k\,(1+q\,m)\leq \lambda\,k^2+\mu\,m\,q\,(1+q(m-1))\quad(\forall k,m\in\mathbb{N}).
\end{equation}
If we find a value for $\lambda$ and $\mu$ with $\lambda\geq 1$, then this guarantees that \eqref{eq:ccp} holds automatically for all affine costs $c_{e}(x)=a_{e}\,x+b_{e}$ with 
$a_{e}\geq 0, b_{e}\geq 0$. 
Indeed, we have $c_{e}^{q}(x)= q\parens*{a_{e}\,(1+q(x-1))+ b_{e}}$
and because $\lambda\geq 1$ we get
\begin{align*}
\frac{1}{q}\,k\, c_{e}^{q}(1+m)
&=a_{e}\, k(1+q\,m) + b_{e} k\\
&\leq\lambda\, a_{e}\, k^2 + b_{e} k+\mu\, m\,q\,a_{e}(1+q(m-1)) \\
&\leq \lambda\,k\, c_{e}(k)+\mu\,m\, c_{e}^{q}(m).
\end{align*}
Let $\ell\in\mathbb{N}\setminus\{0\}$. 
We now show that, with 
\begin{equation}
\label{eq:lambda-mu-ell}   
(\lambda,\mu)=\parens*{\frac{\ell(\ell+1)q^2+2\ell q+1}{2\ell q+1},\frac{1}{2\ell q+1}}
\end{equation}
\eqref{eq:smooth} holds for all $k,m\in\mathbb{N}$.
For $k=0$ this holds because $\mu>0$ and $q\leq 1$. 
If $k=1$, then \eqref{eq:smooth} can be reduced to $(\ell-m)(\ell+1-m)q^2\geq 0$, which clearly holds for all $m\in\mathbb{N}$. 
Finally, for $k\geq 2$ we note that the expression $\lambda\cdot  k^2+\mu\cdot q\cdot m(1+q\cdot (m-1))-k(1+m\cdot q)$ is quadratic in $m$. Its miminizer for $m\in\mathbb{R}$ is 
\begin{equation*}
m^*=\frac{k+q +2\ell kq -1}{2q}
\end{equation*} 
and its minimum value is
\begin{equation*}
\frac{(k+1)(k-1)+k(k-2) + 2q (1 - q) +4 \ell q(1+q)k (k - 1)+ (k-q)^2}{4(2\ell q +1)}\geq 0.
\end{equation*}
This implies  \eqref{eq:smooth} for all $k\geq 2$ and $m\in\mathbb{N}$, from which it follows that $\gamma_{\mathsf{pr}}(\mathcal{C}_{\mathsf{aff}},q)\leq\xi_{\ell}(q)$. The proof is then completed by taking the infimum over $\ell\in\mathbb{N}\setminus\{0\}$.
\Halmos
\endproof

\section{Routing Games with Linear Costs are Tight for $\OPoA(\mathcal{C}_{\mathsf{aff}},q)$}
\label{se:routing_games_are_tight}

The following examples show that the upper bounds for the \ac{OPoA} in Theorem~\ref{th:upper} are tight
and are in fact attained (at least asymptotically) by network routing games with purely linear costs and homogeneous players. 
We proceed in order with three examples that address the three regimes $(0,\bar{q}_{0}],  [\bar{q}_{0},\bar{q}_{1}], [\bar{q}_{1},1]$, for $\bar{q}_{0}=1/4$ and $\bar{q}_{1}\sim 0.3774$. These examples are inspired by the 
minimization problems that define $\gamma(\mathcal{C}_{\mathsf{aff}}^{q})$
subject to the constraints  \eqref{eq:r1}.

\begin{example} 
\label{ex:low}
Let $k\in\mathbb{N}$ and consider a routing game with $n=2k$ players on the bypass network $\mathcal{B}_{k}$ shown in Figure~\ref{fi:bypass}. 
Assume that $p_{i}= q>0$ for all $i\in\mathcal{N}$.
Players $i=1,\ldots,k$ have two strategies, $s_{i}$ and $\bar{s}_{i}$, to travel from origin $\mathsf{O}_{i}$ to destination $\mathsf{D}_{i}$. 
Strategy $s_{i}$ consists of an exclusive direct link $e_{i}$ with cost $c_{i}(x)=x$, whereas the bypass strategy $\bar{s}_{i}$ uses a faster shared link $\bar{e}$ with cost 
\[
\bar{c}(x)=\frac{1}{1+2kq}\cdot x
\] 
connected to $\mathsf{O}_{i}$ and $\mathsf{D}_{i}$ by zero cost links (dashed). 
The remaining players $i=k+1,\ldots, 2k$ have a common origin $\bar{\mathsf{O}}$ and destination $\bar{\mathsf{D}}$ with a unique strategy $\bar{s}_{i}$ using the shared link $\bar{e}$.

\begin{tikzpicture}[->,>=stealth,shorten >=1pt,auto,node distance=5cm,
  thick,main node/.style={circle,fill=blue!20,draw,minimum size=25pt,font=\sffamily\Large\bfseries},source node/.style={circle,fill=green!20,draw,minimum size=25pt,font=\sffamily\Large\bfseries},dest node/.style={circle,fill=red!20,draw,minimum size=25pt,font=\sffamily\Large\bfseries}]

\end{tikzpicture}

\begin{figure}[ht]
\centering
\begin{tikzpicture}[thick,scale=0.6, every node/.style={transform shape}]
\def \n {5}
\def \radius {4cm}
\def \radiuss {6.5cm}
\def \margin {6} % margin in angles, depends on the radius

  \node[draw, circle,minimum size = 0.6cm](Ob)  at (-1,-3.6) {\small $\bar{\mathsf{O}}$};
  \node[draw, circle,minimum size = 0.6cm](Db)  at (10,-3.6) {\small $\bar{\mathsf{D}}$};
  \draw[->, >=latex,out = -90, in = -90] (Ob) to (Db) ;
  \node[fill=white]  at (4.5,-7.2)  {$\bar{e}$};
     
\foreach \s in {1,...,\n}
{
  \node[draw, circle,minimum size = 0.6cm](O\s)  at (3,-1.2*\s) {\small $\mathsf{O}_\s$};
  \node[draw, circle,minimum size = 0.6cm](D\s)  at (6,-1.2*\s) {\small $\mathsf{D}_\s$};
  \draw[->, >=latex] (O\s) -- (D\s) node[midway,fill=white,scale=0.9] {$e_\s$};
  \draw[->, dashed,>=latex] (O\s) -- (Ob) ;
  \draw[->, dashed,>=latex] (Db) -- (D\s) ;
}
\end{tikzpicture}

\caption{\label{fi:bypass} The bypass network $\mathcal{B}_{5}$. Dashed links have zero cost.}
\end{figure}

We claim that for each player $i=1,\ldots,k$ the bypass $\bar{s}_{i}$ is a strictly dominant strategy. 
Indeed, in every strategy profile there are at most $2k$ players on $\bar{e}$, 
and thus, for all $\boldsymbol{s}'_{-i}\in \mathcal{S}_{-i}$,
\begin{subequations}
\label{eq:cost-bypass}
\begin{align}
C_{i}(\bar{s}_{i},\boldsymbol{s}'_{-i})
&\leq q\cdot\frac{1}{1+2kq}\cdot(1+(2k-1)q)\\
&< q=C_{i}(s_{i},\boldsymbol{s}'_{-i}).
\end{align}
\end{subequations}
Hence, in the unique \ac{BNE} all players use $\bar{s}_{i}$, whereas in the optimal profile $\boldsymbol{s}^{\ast}$ players $i=1,\ldots,k$ use their exclusive route $s_{i}^{\ast}=s_{i}$ and players
$i=k+1,\ldots,2k$ use their only available strategy $s_{i}^{\ast}=\bar{s}_{i}$. This yields the lower bound
\[
\OPoA(\mathcal{C}_{\mathsf{aff}},q)\geq\frac{2kq\cdot\dfrac{1}{1+2kq}\cdot(1+(2k-1)q)}{kq\cdot\dfrac{1}{1+2kq}\cdot(1+(k-1)q)+kq}=\frac{4kq+2-2q}{3kq+2-q}.
\]
This quantity increases towards $4/3$ as $k$ grows to $\infty$. In particular, it follows that for $q\in (0,1/4]$, the bound of Theorem~\ref{th:upper} is tight.
\end{example}

\begin{example}
\label{ex:mid}
Consider a pair of integers $m>k\geq 1$ and set $n=m+k$. We build a graph $G_{k,m}$ consisting of a roundabout with $n$ edges of the form $(A_{i},B_{i})$, with linear costs $h_{i}(x)=\gamma x$, where 
\begin{equation}
\label{eq:alpha}
\gamma=\frac{q}{m(1-q+q m)- k(1+q m)},
\end{equation}
connected by zero-cost links $(B_{i},A_{i+1})$ (modulo $n$). Notice that $\gamma>0$.
Additionally there are $n$ exit edges $(B_{i},F_{i})$ with costs $g_{i}(x)=x$. 
Figure~\ref{fi:Roundabout2} illustrates the roundabout network $G_{2,4}$.

\begin{figure}[ht]
\centering
\begin{tikzpicture}[thick,scale=0.6, every node/.style={transform shape}]

\def \n {6}
\def \radius {4cm}
\def \radiuss {6.5cm}

\def \margin {6} % margin in angles, depends on the radius
\foreach \s in {1,...,\n}
{
  \node[draw, circle,minimum size = 0.6cm](H\s)  at ({90-360/\n * (\s - 1)}:\radius) {\small $A_\s$};
  \node[draw, circle](G\s)  at ({90-360/\n * (\s - 1/3)}:\radius) {\small $B_\s$};
  \node[draw, circle](E\s)  at ({90-360/\n * (\s - 1/3)}:\radiuss) {\small $F_\s$};
  \draw[->, >=latex] ({90-360/\n * (\s - 1)-\margin}:\radius) arc ({90-360/\n * (\s - 1)-\margin}:{90-360/\n * (\s-1/3)+\margin}:\radius) node[midway,fill=white,scale=0.9] {$h_\s$};
  \draw[->, >=latex] (G\s) -- (E\s) node[midway,fill=white,scale=0.9] {$g_\s$};
  \draw[-, dashed,>=latex] ({90-360/\n * (\s - 1/3)-\margin}:\radius) arc ({90-360/\n * (\s -1/3)-\margin}:{90-360/\n * (\s)+\margin}:\radius);
}

  \node[draw, circle](O1)  at (0,0) {\small $\mathsf{O}_{1}$};
     \draw[->, dashed,>=latex] (O1) -- (H1);
     \draw[->, dashed,>=latex] (O1) -- (H3);
  \node[draw, circle](D1)  at (5.5,6.5) {\small $\mathsf{D}_{1}$};
     \draw[->, dashed,>=latex,out = 22, in = 170] (E6) to (D1);
     \draw[->, dashed,>=latex,out = 80, in=293] (E2) to (D1);

\end{tikzpicture}

\caption{\label{fi:Roundabout2}The roundabout network $G_{2,4}$. 
For clarity only the origin and destination for player $i=1$ are shown. 
The corresponding strategies are
$s_{1}^{\ast}=\{h_{1},h_{2},g_{2}\}$ and $s_{1}=\{h_{3},h_{4},h_{5},h_{6},g_{6}\}$.
Dashed links have zero cost. 
}

\end{figure}

Consider players $i=1,\ldots,n$ with $p_{i}= q>0$.
Players have origin nodes $\mathsf{O}_{i}$, each of which has two outgoing links connecting to the roundabout at the nodes $A_{i}$ and $A_{i+k}$ (modulo $n$).
Similarly, players have destination nodes $\mathsf{D}_{i}$, each of which can be reached from the exit nodes $F_{i+k-1}$ and $F_{i+k+m-1}$ (modulo $n$). 
Each player $i$ has two undominated strategies that consist of entering the roundabout through one of the two available entrances and proceeding clockwise to the closest exit leading to $\mathsf{D}_{i}$: (1)
the \emph{short route} $s_{i}^{\ast}=\braces{h_{i},\ldots,h_{i+k-1},g_{i+k-1}}$, which uses $k$ resources of type $h_{j}$ and only one $g_{j}$, and (2) the \emph{long route} $s_{i}=\braces{h_{i+k},\ldots,h_{i+k+m-1},g_{i+k+m-1}}$, which uses $m$ resources of type $h_{j}$ and only one $g_{j}$. 

If all players choose the long route $s_{i}$, then each $h_{j}$ has a load of $m$ players and each $g_{j}$ a load of $1$,
so that every player experiences the same cost
\[
q[m\gamma (1-q+q m)+1].
\]
Shifting individually to the short route $s_{i}^{\ast}$ implies the cost
\[
q[k\gamma(1+q m)+1+q],
\]
so that, by the choice of $\gamma$, all players using $s_{i}$ constitutes an equilibrium.
The social cost of this equilibrium is
\begin{equation*}
C(\boldsymbol{s})=n q[m\gamma (1-q+q m)+1].
\end{equation*}
Now, the feasible routing where all players use the short route $s_{i}^{\ast}$ gives an upper 
bound for the optimal social cost. 
In this case the loads are $k$ on each $h_{j}$ and again $1$ on each $g_{j}$, so that
\[
C(\boldsymbol{s}^{\ast})\leq n q[k\gamma (1-q+q k)+1],
\]
which yields the following lower bound for the \ac{PoA}
\begin{equation}\label{eq:Apoab2}
\OPoA(\mathcal{C}_{\mathsf{aff}},q)\geq \frac{(1+q) m (1-q+q m)- k(1+q m)}{q k(1-q+q k)+m(1-q+q m)- k(1+q m)}\,.
\end{equation}

Take $z=1+q+\sqrt{q(2+q)}$ and $m=\floor{z k}$. 
Then $m>k$ for $k$ large enough. In fact,
\[
\frac{m}{k}=\frac{\floor{z k}}{k}\to z,\quad\text{as } k\to\infty.
\]
With this choice of $m$ both the numerator and denominator in \eqref{eq:Apoab2} grow quadratically 
with $k$, so that dividing by $k^{2}$ and letting $k\to\infty$ we get the asymptotic lower bound
\begin{equation}\label{eq:Apoab4}
\OPoA(\mathcal{C}_{\mathsf{aff}},q)\geq \frac{(1+q)z^{2}-z}{q+z^{2}-z}=\frac{1+q+\sqrt{q(2+q)}}{1-q+\sqrt{q(2+q)}}\,.
\end{equation}
In particular, it follows that for $q\in [\bar{q}_{0},\bar{q}_{1}]$, the bound of Theorem~\ref{th:upper} is tight.
\end{example}

\begin{example}
\label{ex:high}
Consider the network congestion game of Figure~\ref{fi:triangulo2}. The game contains $3$ players, 6 costly resources $\{h_{1},h_{2},h_{3},g_{1},g_{2},g_{3}\}$,
and 15 connecting links (the dashed links).
Assume that $p_{i}= q>0$ for all $i\in\mathcal{N}$.
The cost functions are $c_{e}(x)= q \cdot x$ for $e\in \{h_{1},h_{2},h_{3}\}$ and $c_{e}(x)=x$ for $e\in \{g_{1},g_{2},g_{3}\}$,
whereas the dashed links have zero cost. 
Ignoring the dashed links, each player $i$ has two pure strategies 
$\{h_{i},g_{i}\}$ and $\{h_{i-1},h_{i+1},g_{i+1}\}$ (all indices are modulo 3).

\begin{figure}[ht]
\centering
\begin{tikzpicture}
\def \nn {3}
\def \radiusg {2.5cm}
\def \radius {5cm}
\def \radiusn {3.3cm}
\def \radiusnn {1.6cm}
\def \radiusnnn {0.8cm}
\def \margin {4} % margin in angles, depends on the radius
\def \marginn {5} 
\def \marginnn {8} 
\foreach \s in {1,...,\nn}
{
\node[draw, circle,thick,scale=0.7](D\s) at ({30-360/\nn * \s}:\radiusg) {$\mathsf{D}_\s$};
\node[draw, circle,thick,scale=0.8](I\s) at ({90-360/\nn *\s}:\radius) {};
}

\draw[->,dashed,thick,red] (I1) to (D1);
\draw[->,dashed,thick] (I1) to (D3);
\draw[->,dashed,thick] (I2) to (D2);
\draw[->,dashed,thick,blue] (I2) to (D1);
\draw[->,dashed,thick] (I3) to (D3);
\draw[->,dashed,thick] (I3) to (D2);

\foreach \s in {1,...,\nn}  \node[draw, circle,thick,scale=0.8](G\s) at ({90-360/\nn*\s}:\radiusn) {};
\draw[->,thick,red] (G1) -- (I1) node[midway,fill=white,scale=0.65] {$g_{1}$};
\draw[->,thick,blue] (G2) -- (I2) node[midway,fill=white,scale=0.65] {$g_{2}$};
\draw[->,thick] (G3) -- (I3) node[midway,fill=white,scale=0.65] {$g_{3}$};

\foreach \s in {1,...,\nn}  \node[draw, circle,thick,scale=0.8](H\s) at ({90-360/\nn*\s}:\radiusnn) {};
\draw[->,thick,red] (H1) -- (G1) node[midway,fill=white,scale=0.65] {$h_{1}$};
\draw[->,thick,blue] (H2) -- (G2) node[midway,fill=white,scale=0.65] {$h_{2}$};
\draw[->,thick,blue] (H3) -- (G3) node[midway,fill=white,scale=0.65] {$h_{3}$};

\foreach \s in {1,...,\nn}  \node[draw, circle,thick,scale=0.7](O\s) at ({150-360/\nn*\s}:\radiusnnn) {$\mathsf{O}_\s$};
\draw[->,dashed,thick,red] (O1) to (H1);
\draw[->,dashed,thick,blue] (O1) to (H3);
\draw[->,dashed,thick] (O2) to (H2);
\draw[->,dashed,thick] (O2) to (H1);
\draw[->,dashed,thick] (O3) to (H3);
\draw[->,dashed,thick] (O3) to (H2);

\draw[->,dashed,bend right,thick] (G1) to (H3);
\draw[->,dashed,bend right,thick] (G2) to (H1);
\draw[->,dashed,bend right,thick,blue] (G3) to (H2);
\end{tikzpicture}
\caption{The triangle network. The pure strategies $\{h_{i},g_{i}\}$ and $\{h_{i-1},h_{i+1},g_{i+1}\}$ for player $i=1$ are highlighted in red and blue respectively. Dashed links have zero cost.}
\label{fi:triangulo2} 
\end{figure}
A strategy profile $\boldsymbol{s}$ is a \ac{BNE} if $s_{i}=\{h_{i-1},h_{i+1},g_{i+1}\}$ for all $i\in \mathcal{N}$, because
\[
2q(q+1)+1\leq q(2q+1)+(q+1).
\]
The corresponding expected total costs are
\[
C(\boldsymbol{s})=3(q^{2}\cdot 4q+ 2q (1 - q)\cdot q)+3q.
\]
Second, the strategy profile $\boldsymbol{s}^{\ast}$ in which $s_{i}=\{h_{i},g_{i}\}$ yields an expected total cost of
\[
C(\boldsymbol{s}^{\ast})=3(q^{2}+q).
\]
Therefore,
\[
\OPoA(\mathcal{C}_{\mathsf{aff}},q)\geq\frac{3q(1+2q+2q^{2})}{3q(1+q)}=1+q+\frac{q^{2}}{1+q}.
\]
In particular, it follows that for $q\in [\bar{q}_{1},1]$, the bound of Theorem~\ref{th:upper} is tight.
\end{example}

\section{Congestion Games with Linear Costs are Tight for $\PPoA(\mathcal{C}_{\mathsf{aff}},q)$}
\label{se:proaff}

The following example, which is a variation of the congestion game in \cite{ChrKou:STOC2005}, will show that the bound of Theorem~\ref{th:Tight-PPoA-Affine} is tight. Let $q\in(0,1]$ and select $\ell$ such that \eqref{eq:range} holds. 
We will construct a sequence of \aclp{BCG} with purely linear costs and $n$ homogeneous players with  probabilities $p_{i}\equiv q$, such that, as $n\to\infty$, the \acl{PPoA} approaches $\xi_{\ell}(q)=\Xi(q)$.

The resource set is composed of $m$ disjoint buckets $\mathcal{E}_{1},\ldots,\mathcal{E}_m$
where each $\mathcal{E}_k$ contains $\binom{n}{\ell}+\binom{n}{\ell+1}$ resources (see Figure~\ref{fi:illus} below). 
Specifically, for every subset $I\subseteq\{1,\ldots n\}$ of cardinality $\abs*{I}=\ell$ we include in $\mathcal{E}_k$ a resource $a_{k}^{I}$ with linear cost $\alpha_{1} x$, and  for each $J\subseteq\braces*{1,\ldots n}$ with $\abs*{J}=\ell+1$
a resource $b_{k}^{J}$ with cost $\alpha_{2} x$. 
Each player $i\in\braces*{1,\ldots,n}$ has $m+1$  strategies: either choose a single bucket $\mathcal{E}_k$ with all the resources in it, or select a player-specific strategy $s_{i}$ that contains all the resources $a_{k}^{I}$ and $b_{k}^{J}$ (across all buckets $\mathcal{E}_k$) whose label sets $I$ and $J$ include player $i$.

\begin{figure}[ht]
\centering
\begin{tikzpicture}
\node at (-1.5,0) {$\alpha_{1} x$};
\node at (-1.5,-1) {$\alpha_{2} x$};
\node[circle,draw=black,fill=blue!20, minimum size=5pt]  at (0,0) {$1$};
\node[circle,draw=black, minimum size=5pt]  at (1,0) {$2$};
\node[circle,draw=black, minimum size=5pt]  at (2,0) {$3$};
\node[circle,draw=black,fill=blue!20, minimum size=4pt]  at (0,-1) {$1,2$};
\node[circle,draw=black,fill=blue!20, minimum size=4pt]  at (1,-1) {$1,3$};
\node[circle,draw=black, minimum size=4pt]  at (2,-1) {$2,3$};
\node[circle,draw=black, minimum size=100pt]  at (1,-0.5) {};
\node at (1,1.5) {$\mathcal{E}_{1}$};
\node[circle,draw=black, fill=blue!20, minimum size=5pt]  at (4,0) {$1$};
\node[circle,draw=black, minimum size=5pt]  at (5,0) {$2$};
\node[circle,draw=black, minimum size=5pt]  at (6,0) {$3$};
\node[circle,draw=black, fill=blue!20, minimum size=4pt]  at (4,-1) {$1,2$};
\node[circle,draw=black, fill=blue!20, minimum size=4pt]  at (5,-1) {$1,3$};
\node[circle,draw=black, minimum size=4pt]  at (6,-1) {$2,3$};
\node[circle,draw=black, minimum size=100pt]  at (5,-0.5) {};
\node at (5,1.5) {$\mathcal{E}_{2}$};
\end{tikzpicture}
\caption{Example with $n=3$ players, $m=2$ buckets, and $\ell=1$. The $m+1$ strategies of player $i=1$ are either $\mathcal{E}_{1}$ or $\mathcal{E}_{2}$, or the player specific strategy $s_{1}$ shown in blue.}
\label{fi:illus}
\end{figure}

We will fix $\alpha_{1},\alpha_{2}$ so that the profile where each player $i$ selects her player-specific strategy $s_{i}$
turns out to be a (pure) Nash equilibrium. For this profile, the expected cost for every player $i$ is 
\begin{align*}
C_{i}(s_{i},\boldsymbol{s}_{-i})=m\,q\,\parens*{\binom{n-1}{\ell-1}(1+(\ell-1) q)\alpha_{1}+\binom{n-1}{\ell}(1+\ell q)\alpha_{2}},
\end{align*}
whereas a unilateral deviation to any of the alternative strategies $\mathcal{E}_k$ produces the cost 
\begin{align*}
C_{i}(\mathcal{E}_{k},\boldsymbol{s}_{-i})
&=
q\,\left(\binom{n-1}{\ell-1}(1+(\ell-1) q)\right.\\
&\quad\left.+\binom{n-1}{\ell}(1+\ell q)\right)\alpha_{1} \\
&\quad+q\,\left(\binom{n-1}{\ell}(1+\ell q)\right.\\
&\quad\left.+\binom{n-1}{\ell+1}(1+(\ell+1) q)\right)\alpha_{2}.
\end{align*}
The equilibrium conditions impose $C_{i}(s_{i},\boldsymbol{s}_{-i})=C_{i}(\mathcal{E}_{k},\boldsymbol{s}_{-i})$, which simplifies to
\begin{equation*}
\frac{(m(\ell+1)-n)(1+\ell q) + (\ell+1-n)q}{\ell+1}\,\alpha_{2} 
= \frac{(n-m\ell )(1+\ell q)+(m-1)\ell q}{n-\ell}\,\alpha_{1}
\end{equation*}
and can be achieved by setting 
\begin{align}
\label{eq:alpha1}
\alpha_{1}
&= \frac{(m(\ell+1)-n)(1+\ell q)+(\ell+1-n)q}{\ell+1},\\
\label{eq:alpha2}
\alpha_{2}
&= \frac{(n - m\ell)(1+\ell q)+(m-1)\ell q}{n -\ell}.
\end{align}
For the costs to be nondecreasing, the slopes $\alpha_{1},\alpha_{2}$ must be nonnegative, which translates into
\begin{align}
\label{eq:cota1}
\alpha_{1}\geq 0 & \iff \frac{m}{n}\geq \frac{(1+\ell q)+q \parens*{1-\frac{\ell+1}{n}}}{(\ell+1) (1+\ell q)},\\
\label{eq:cota2}
\alpha_{2}\geq 0 & \iff \frac{m}{n}\leq \frac{1+\ell q \parens*{1-\frac{1}{n}}}{\ell (1+\ell q-q)}.
\end{align}
To compute the social cost for this equilibrium we observe that the load of the resources $a_{j}^{I}$ and $b_{j}^{J}$ are distributed as $\Binomial(\ell,q)$ and $\Binomial(\ell+1,q)$, 
respectively. 
Letting  $\rho_{1}=\ell q(1+(\ell-1) q)$ and $\rho_{2}=(\ell+1)q(1+\ell q)$ denote  their corresponding second moments, we get
\begin{equation}
\label{eq:SCNE0}
\max_{\boldsymbol{s}\in\mathsf{NE}(\Gamma^{\boldsymbol{p}})}C(\boldsymbol{s}) \geq m\parens*{\binom{n}{\ell}\alpha_{1}\rho_{1}+\binom{n}{\ell+1}\alpha_{2}\rho_{2}}.
\end{equation}

Now, the prophet observes the demand $N\sim\Binomial\parens*{n,q}$ and, among all the possible rules,  can choose to distribute the players 
as uniformly as possible over the strategies $\mathcal{E}_{j}$: when the number of players present in the game is
$N=m\,\floor*{N/m}+j$ with $0\leq j<m$, 
put a load $\ceil*{N/m}$ on $j$ buckets and a load $\floor*{N/m}$ on the remaining $m-j$ buckets, which entails the social cost 
\begin{equation*}
\parens*{j\ceil*{\frac{N}{m}}^2+(m-j)\,\floor*{\frac{N}{m}}^2}\parens*{\binom{n}{\ell}\alpha_{1}+\binom{n}{\ell+1}\alpha_{2}}.    
\end{equation*}
Introducing the function  $L:\mathbb{R}\to[0,1/4]$ defined as 
\begin{equation}
\label{eq:L} 
L(x)=(x-\floor*{x})(\ceil*{x}-x), 
\end{equation}
and defining $M=N/m$ so that $j=m(M-\floor*{M})$,
we can simplify the expression
\begin{align*}
\parens*{j\ceil*{{N}/{m}}^2+(m-j)\,\floor*{{N}/{m}}^2}
&=\parens*{m\,\floor*{M}^2+j\,\parens*{\ceil*{M}^2-\floor*{M}^2}}\\
&=m\,\parens*{\floor*{M}^2+\parens*{M-\floor*{M}}\,\parens*{\ceil*{M}^2-\floor*{M}^2}}\\
&=m\parens*{M^2+L(M)}
\end{align*}
where this last identity is readily checked by considering the cases
$M=k$ and $k<M<k+1$ with $k\in\mathbb{N}$.
Hence, with $N\sim\Binomial(n,q)$, the expected social cost for the prophet is at most
\begin{align*}
C_{\mathsf{pr}}&\leq m\Expect\bracks*{(N/m)^2
+L(N/m)}\parens*{\binom{n}{\ell}\alpha_{1}+\binom{n}{\ell+1}\alpha_{2}},
\end{align*}
which combined with \eqref{eq:SCNE0} implies
\begin{equation}
\label{eq:PoA0}
\PPoA\geq R_{m,n}\eqdef
\frac{1}{\Expect\bracks*{(N/m)^2+L(N/m)}}\cdot
\frac{\binom{n}{\ell}\alpha_{1}\rho_{1}+\binom{n}{\ell+1}\alpha_{2}\rho_{2}}
{\parens*{\binom{n}{\ell}\alpha_{1}+\binom{n}{\ell+1}\alpha_{2}}}.
\end{equation}

Now, allowing $m$ and $n$ to increase to infinity subject to \eqref{eq:cota1}-\eqref{eq:cota2}, the quotients $m/n$ can approximate any element of the interval
\begin{equation*}
\bracks*{\frac{1+\ell q+q}{(\ell+1) (1+\ell q)}\,,\frac{1+\ell q}{\ell (1+\ell q-q)}}.
\end{equation*}
We note that $q$  lies in this interval if and only if it satisfies \eqref{eq:range}, which holds because of our choice of $\ell$. 
So, let us consider a sequence of instances with $m,n$ tending to $\infty$ with $m/n \to q$. 
Then
\begin{equation*}
\frac{N}{m}=\frac{n}{m}\frac{N}{n}\xrightarrow{\text{a.s.}}\frac{1}{q}\,q=1
\end{equation*}
and, because $L(\argdot)$ is continuous and bounded with $L(1)=0$, the portmanteau theorem implies that the term 
$\Expect\bracks*{(N/m)^2+L(N/m)}$ in the denominator
of $R_{m,n}$ converges to 1.  
The remaining terms in the quotient $R_{m,n}$ can be simplified as
\begin{equation*}
\tilde R_{m,n}\eqdef \frac{\binom{n}{\ell}\alpha_{1}\rho_{1}+\binom{n}{\ell+1}\alpha_{2}\rho_{2}}{\binom{n}{\ell}\alpha_{1}+\binom{n}{\ell+1}\alpha_{2}}=
\frac{(\ell+1)\alpha_{1}\rho_{1}+(n-\ell)\alpha_{2}\rho_{2}}{(\ell+1)\alpha_{1}+(n-\ell)\alpha_{2}}.
\end{equation*}
Dividing  numerator and denominator by $n$, and noting that \eqref{eq:alpha1}--\eqref{eq:alpha2} yield
 \begin{align*}
{(\ell+1)\alpha_{1}}/{n}&\to \ell(\ell+1)q^2-1,\\
{(n-\ell)\alpha_{2}}/{n}&\to 1-\ell(\ell-1)q^2,
\end{align*}
it follows that
\begin{equation*}
\tilde R_{m,n}\to
\frac{(\ell(\ell+1)q^2-1)\rho_{1}+(1-\ell(\ell-1)q^2)\rho_{2}}{2\ell q^2}.
\end{equation*}
Substituting the values of $\rho_{1}$ and $\rho_{2}$, and simplifying the resulting expression, we conclude
\begin{equation*}
\PPoA(\mathcal{C}_{\mathsf{aff}},q)\geq \lim R_{m,n}
= \lim\tilde R_{m,n}
=\frac{\ell(\ell+1)q^2+2\ell q+1}{2\ell q}=\xi_{\ell}(q)=\Xi(q).
\end{equation*}
This, combined with Proposition~\ref{pr:lm-linear}, completes the proof.

\section{Congestion Games with Polynomial Costs}
\label{se:propoly}

In this section we prove Propositions~\ref{pr:proph} and \ref{pr:polycosts}.

\proof{Proof of Proposition~\ref{pr:proph}.}
By considering only the monomial $c_{d}(x)=x^{d}\in\mathcal{P}_{d}$ and letting $\mathcal{P}$
denote the set of all pairs 
$(k,m)\in\mathbb{N}^{2}$ with $k\geq 1$,
we get the following lower bound 
\begin{equation*}
\PPoA(\mathcal{P}_{d},q)\geq \gamma(\{c_{d}\},q)
=\inf_{\mu\in [0,1)}\sup_{(k,m)\in\mathcal{P}}\frac{k\,c_{d}^{q}(m+1)-\mu\, m\, q\, c_{d}^{q}(m)}{k^{d+1}\,(1-\mu)\,q}.
\end{equation*}
Taking $q=1/n$ and choosing $k=1$ and $m=n$, we can further minorize
\begin{equation}
\label{eq:gamma-inf}
\gamma(\{c_{d}\},1/n)
\geq\inf_{\mu\in [0,1)}\frac{c_{d}^{1/n}(n+1)-\mu\,c_{d}^{1/n}(n)}{(1-\mu)/n}.
\end{equation}
Because $c_{d}^{1/n}(n+1)\geq c_{d}^{1/n}(n)$,
the infimum in \eqref{eq:gamma-inf} is attained at $\mu=0$.
Therefore, 
\begin{equation*}
\gamma(\{c_{d}\},1/n)
\geq n\,c_{d}^{1/n}(n+1)=\Expect[(1+Y_{n})^{d}],
\end{equation*}
with $Y_{n}\sim\Binomial(n,1/n)$. 
Letting $n\to\infty$, we have that $Y_{n}$ converges weakly towards a  random variable $Y\sim \Poisson(1)$. 
Therefore,
\begin{equation*}
\PPoA(\mathcal{P}_{d},0^+)\geq
\lim_{n\to\infty}\Expect[(1+Y_{n})^{d}]=\Expect[(1+Y)^{d}].
\end{equation*}
The latter can be computed as
\begin{equation*}
\Expect[(1+Y)^{d}]=\sum_{k=0}^{\infty}(k+1)^{d}\,\frac{\expo^{-1}}{k!}=\sum_{k=0}^\infty(k+1)^{d+1}\,\frac{\expo^{-1}}{(k+1)!}=\Expect[Y^{d+1}],
\end{equation*}
which is known to be  the $(d+1)$-th Bell number $B_{d+1}$ \citep[see][]{Dob:1877,Tou:AM1939}.
\Halmos
\endproof

To prove Proposition~\ref{pr:polycosts}, we will exploit the following property.

\begin{lemma}
\label{le:BinExp}
Let $h:\mathbb{N}^2\to\mathbb{R}$ such that $h(i,j)+h(j,i)=0$ for all $i,j\in\mathbb{N}$, and 
$h(i,j)>0$ if $i>j$. 
Let $Y\sim\Binomial(n,p)$ and $Z\sim\Binomial(m,p)$ be two 
independent Bernoulli variables. 
Then,  $\Expect[h(Y,Z)]=0$ if $n=m$ and 
$\Expect[h(Y,Z)]>0$ if $n>m$. 
\end{lemma}

\proof{Proof.}
The case $n=m$ follows because $\Expect[h(Y,Z)]=\Expect[h(Z,Y)]$ and $\Expect[h(Y,Z)+h(Z,Y)]=0$. 
Suppose next that $n>m$ and let  
$p_{i,n}=\binom{n}{i}p^i(1-p)^{n-i}$ denote the Binomial probabilities. 
Then,
\begin{align*}
\Expect[h(Y,Z)]
&=\sum_{i=0}^n\sum_{j=0}^m h(i,j)\,p_{i,n}\,p_{j,m}\\
&>\sum_{i=0}^m\sum_{j=0}^m h(i,j)\,p_{i,n}\,p_{j,m}\\
&=\sum_{0\leq j<i\leq m}[ h(i,j)\,p_{i,n}\,p_{j,m}+h(j,i)\,p_{j,n}\,p_{i,m}]\\
&=\sum_{0\leq j<i\leq m} h(i,j)\,[p_{i,n}\,p_{j,m}-\,p_{j,n}\,p_{i,m}],
\end{align*}
where the inequality is a consequence of the assumption $h(i,j)>0$ for $i>j$ and the last equality follows from $h(j,i)=-h(i,j)$. 
The conclusion $\Expect[h(Y,Z)]>0$ follows by using again the fact that $h(i,j)>0$ combined with $p_{i,n}\,p_{j,m}>p_{j,n}\,p_{i,m}$, where the latter follows itself from the inequality $\binom{n}{i}\binom{m}{j} > \binom{n}{j}\binom{m}{i}$ for all $i>j$ and $n>m$.    
\Halmos
\endproof

\proof{Proof of Proposition~\ref{pr:polycosts}.}
It suffices to show that any $(\lambda,\mu)$-smoothness parameter which is valid for 
$c_d^{q}(\argdot)$ is also valid for all $c_{j}^{q}(\argdot)$
with $j\leq d$.
Now, for any fixed degree $d$ and probability $q$,  the smallest feasible $\lambda$ 
as a function of $\mu\in [0,1)$ is
\begin{equation*}
\lambda_{d}^{q}(\mu)=\sup_{(k,m)\in\mathcal{S}}\Phi_{k,m}^{q}(d,\mu),\quad\text{where}\quad \Phi_{k,m}^{q}(d,\mu)\eqdef
\frac{ k\,c_{d}^{q}(1+m) -\mu\, m \, c_{d}^{q}(m)}{k\, c_{d}^{q}(k)}
\end{equation*}
so that it suffices to prove that $\lambda_{d}^{q}(\mu)$ increases with $d$. 

For $m=0$ we have $\Phi_{k,0}^{q}(d,\mu)=c_{d}^{q}(1)/c_{d}^{q}(k)$ whose maximum is 1 (attained for $k=1$), whereas for $k>m\geq 1$ we have $\Phi_{k,m}^{q}(d,\mu)\leq c_{d}^{q}(1+m)/c_{d}^{q}(k)\leq 1$. 
Therefore, in the supremum that gives $\lambda_{d}^{q}(\mu)$ it suffices to consider $1\leq k\leq m$ plus the special case $(k,m)=(1,0)$.
Moreover, for $\mu$ fixed the supremum can be further restricted to those pairs $(k,m)$ such that $\Phi_{k,m}^{q}(d,\mu)\geq 0$.
Altogether, to establish the monotonicity of $\lambda_{d}^{q}(\mu)$ with respect to $d$, it suffices to show that for $1\leq k\leq m$ with $\Phi_{k,m}^{q}(d,\mu)\geq 0$
the quotient
$\Phi_{k,m}^{q}(d,\mu)$ increases with $d$. 
To prove this, we will show  that $\Phi_{k,m}^{q}(d,\mu)\geq 0$ implies $\partial_{d} [\Phi_{k,m}^{q}(d,\mu)]\geq 0$.

Let us fix $1\leq k\leq m$ with $\Phi_{k,m}^{q}(d,\mu)\geq 0$. Take independent  variables $Y\sim\Binomial(m,q)$, $Z\sim\Binomial(m-1,q)$, and $X\sim\Binomial(k-1,q)$, so that 
\begin{equation*}
\Phi_{k,m}^{q}(d,\mu)=\frac{ k\,\Expect[(1+Y)^{d}] -\mu\, m \, \Expect[(1+Z)^{d}]}{k\, \Expect[(1+X)^{d}] }.    
\end{equation*}
Differentiating with respect to $d$, and denoting $f(i)=(1+i)^{d}$ and $g(i)=(1+i)^{d}\ln(1+i)$, it follows that $\partial_{d}[ \Phi_{k,m}^{q}(d,\mu)]\geq 0$ is equivalent to
\begin{equation*}
\parens*{k\,\Expect[g(Y)] -\mu\, m \, \Expect[g(Z)]}\Expect[f(X)]
\geq \parens*{k\,\Expect[f(Y)] -\mu\, m \, \Expect[f(Z)]}\Expect[g(X)].
\end{equation*}
Thus, using the independence, we deduce that $\partial_{d}[\Phi_{k,m}^{q}(d,\mu)]\geq 0$ if and only if
\begin{equation}
\label{eq:Binkm} 
k\,\Expect[g(Y)f(X)-f(Y)g(X)]
\geq  \mu\, m \, \Expect[g(Z)f(X)-f(Z)g(X)].
\end{equation}

Now, for $k=m$ this inequality follows directly by applying Lemma~\ref{le:BinExp} with 
\begin{equation*}
h(i,j)=g(i)f(j)-f(i)g(j)=(1+i)^{d}(1+j)^{d}\ln\parens*{\frac{1+i}{1+j}},
\end{equation*}
which shows that the right hand side of  \eqref{eq:Binkm} is $0$ whereas the expression on the left is strictly positive.
Similarly, for $k<m$ the right hand side of  \eqref{eq:Binkm} is strictly positive, and the inequality can be rewritten as
\begin{equation*}
\frac{\mu m}{k}\leq  \frac{\Expect[g(Y)f(X)-f(Y)g(X)]}{\Expect[g(Z)f(X)-f(Z)g(X)]}.
\end{equation*}
Given that the assumption $\Phi_{k,m}^{q}(d,\mu)\geq 0$ translates into 
\begin{equation*}
\frac{\mu m}{k}\leq  \frac{\Expect[f(Y)]}{\Expect[f(Z)]},
\end{equation*}
it remains to show that 
\begin{equation*}
\frac{\Expect[f(Y)]}{\Expect[f(Z)]}\leq \frac{\Expect[g(Y)f(X)-f(Y)g(X)]}{\Expect[g(Z)f(X)-f(Z)g(X)]}.
\end{equation*}
This is equivalent to
\begin{equation*}
\Expect[f(Y)]\Expect[g(Z)f(X)-f(Z)g(X)]\leq \Expect[g(Y)f(X)-f(Y)g(X)]\Expect[f(Z)]
\end{equation*}
which further simplifies to 
\begin{equation*}
\Expect[g(Y)f(Z)-f(Y)g(Z)]\Expect[f(X)]\geq 0.
\end{equation*}
The latter follows again from Lemma~\ref{le:BinExp} which gives $\Expect[h(Y,Z)]>0$.
Thus $\Phi_{k,m}^{q}(d,\mu)\geq 0$ implies $\partial_d [\Phi_{k,m}^{q}(d,\mu)]\geq 0$, as was to be proved.
\Halmos
\endproof

\section{List of Symbols}
\label{se:list-of-symbols}

\begin{longtable}{p{.13\textwidth} p{.82\textwidth}}
$a_{e}$ & slope of the affine cost function, defined in \eqref{eq:affine-costs} \\

$b_{e}$ & constant of the affine cost function, defined in \eqref{eq:affine-costs} \\

$B_{d}$ & $d$-th Bell number, defined in Proposition~\ref{pr:proph}\\

$\mathcal{B}_{k}$ & bypass network, shown in Figure~\ref{fi:bypass} \\

$c_{e}$ & cost of using resource $e$, introduced in \eqref{eq:Ci} \\

$C$ & expected social cost, defined in \eqref{eq:SC}, \eqref{eq:c-soc-stoch}, and \eqref{eq:ESC} \\

$C_{i}$ & cost function of player $i$, defined in \eqref{eq:Ci} and \eqref{eq:c-i-stoch} \\

$C_{\mathsf{ord}}$ & ordinary social optimum expected cost, defined in \eqref{eq:ordinary-social-optimum} \\

$C_{\mathsf{pr}}$ & prophet social optimum expected cost, defined in \eqref{eq:PESC} \\

$\mathcal{C}$ & class of cost functions \\

$\mathcal{C}^{q}$ & class of cost functions derived by the Bernoulli game, defined in Definition~\ref{de:Cp} \\

$\mathcal{C}_{\mathsf{aff}}$ & class of affine cost functions, defined in \eqref{eq:affine-costs} \\

$d$ & degree of  \ac{BPR} functions\\

$\mathsf{D}$ & destination \\

$e$ & resource \\
$\mathcal{E}_{k}$ & resource bucket, defined in the proof of Theorem~\ref{th:Tight-PPoA-Affine} \\
$f$ & difference between maximizer and its closest integer \\

$\mathcal{F}$ & finite subcover of $[0,M]$ \\

$\mathcal{G}$ & class of games \\

$G$ & graph \\

$h_{1}$ & number of buckets of type $1$ \\

$h_{2}$ & number of buckets of type $2$ \\

$i$ & player \\

$\mathcal{I}$ & subset of players \\

$L(x)$ & $(x-\floor*{x})(\ceil{x}-x)$, defined in \eqref{eq:L} \\

$\widehat{m}$ & unconstrained maximizer, defined in \eqref{eq:mhat} \\

$\overline{m}$ & closest integer \\
$M$ & constant strictly smaller than $\gamma_{\mathsf{pr}}(\mathcal{C},q)$ \\
$n$ & number of players \\

$n_{e}$ & number of players who use resource $e$, defined in \eqref{eq:N-e} \\

$\mathcal{N}$ & set of players \\

$N_{e}$ & random number of players who use resource $e$, defined in \eqref{eq:N-e-stoch} \\

$N^{-i}_{e}$ & random number of players different from $i$ who use resource $e$, defined in \eqref{eq:N-e-stoch} \\

$N_{e}^{Y}(\boldsymbol{s})$ & $\sum_{j\in\mathcal{N}}Y_{j}\mathds{1}_{\{e\in s_{j}\}}$, defined in Corollary~\ref{co:util2} \\

$\mathsf{NE}^{\mathsf{coa}}$ & set of coarse correlated equilibria \\

$\mathsf{NE}^{\mathsf{cor}}$ & set of correlated equilibria \\
 
$\mathsf{NE}$ & set of Bayes-Nash pure equilibria \\

$\mathsf{NE}^{\mathsf{mix}}$ & set of  Bayes-Nash mixed equilibria \\

$\OPoA$ & ordinary \acl{PoA}, defined in \eqref{eq:PoAnp} \\

$\mathsf{O}$ & origin \\

$p_{i}$ & probability that player $i$ is active \\

$\boldsymbol{p}$ & vector of probabilities of being active \\

$\mathcal{P}_{d}$ & class of polynomials of degree $d$ with nonnegative coefficients\\

$\PoA(\Gamma)$ & \acl{PoA} of game $\Gamma$, defined in \eqref{eq:PoA-game} \\

$\PoA(\mathcal{C})$ & \acl{PoA} of class $\mathcal{C}$, defined in \eqref{eq:PoA-class} \\

$\PPoA$ & prophet \acl{PoA}, defined in \eqref{eq:PoAp} \\

$\CoaPPoA$ & prophet \acl{PoA} for coarse correlated equilibria, defined in \eqref{eq:PoAp-coarse}\\

$\CorPPoA$ & prophet \acl{PoA} for correlated equilibria, defined in \eqref{eq:PoAp-corr}\\

$\MPPoA$ & prophet \acl{PoA} for mixed equilibria, defined in \eqref{eq:PoAp-mixed} \\

$q$ & upper bound for $p_{i}$, $i\in\mathcal{N}$ \\

$\bar{q}_{0}$ & $=1/4$, defined in Theorem~\ref{th:upper} \\

$\bar{q}_{1}$ & $\approx 0.3774$,  unique real root of $8 q^{3}+4q^{2}=1$, defined in Theorem~\ref{th:upper} \\

$\tilde{q}$ & $\zeta/\nu\approx q$ \\

$Q$ & defined in \eqref{eq:Q} \\
$Q^{c}(x)$ & $(x-\floor{x})\,\ceil{x}\,c(\ceil{x})+(1-x+\floor{x})\,\floor{x} \,c(\floor{x}) $, defined in \eqref{eq:Qc} \\
$Q_{q}(k,m)$ & defined in \eqref{eq:dd0} \\

$r_{i}$ & $p_{i}/q$, defined in Corollary~\ref{co:util2} \\

$R$ & $\sum_{i\in\mathcal{N}}\Expect_{s_{i}^{\ast}\sim\sigma_{i}^{\ast}}
\bracks*{C_{i}(s_{i}^{\ast},\boldsymbol{s}_{-i})}$, 
defined in \eqref{eq:initial} \\

$\boldsymbol{s}$ & strategy profile \\

$\hat{\boldsymbol{s}}$ & equilibrium strategy profile, defined in \eqref{eq:Nash} \\

$\boldsymbol{s}^{\ast}$ & ordinary optimum strategy profile \\

$\mathcal{S}$ & set of strategy profiles \\

$s_{i}$ & strategy of player $i$ \\

$\mathcal{S}_{i}$ & strategy set of player $i$ \\

$\boldsymbol{s}^{\mathcal{I}}$ & optimal strategy profile when the realized player set is $\mathcal{I}$, defined in \eqref{eq:optimal-prophet-action} \\ 

$\mathcal{T}$ & $\braces{(c,k,m) \colon c\in\mathcal{C}\setminus\{c_{0}\}, k \in \mathbb{N}_{+},m\in\mathbb{N} }$, defined in \eqref{eq:triples} \\

$W_{i}$ & indicator of player $i$ being active \\

$x_{e}$ & number of players who choose resource $e$ \\

$X$ & random variable $\sim\Binomial(k-1,q)$, defined in \eqref{eq:cq} \\

$X_{e}$ & random load on resource $e$ \\

$y$ & $\mu/(1-\mu)$, defined in \eqref{eq:y} \\

$Y$ & $\sum_{i=1}^mY_{i}$, defined in Lemma~\ref{le:CBH} \\ 

$Y_{i}$ & random variable $\sim \Bernoulli(r_{i})$, defined in Lemma~\ref{le:CBH} \\

$z$ & $(1+y_{1,q})/2y_{1,q}$, defined in \eqref{eq:z} \\

$Z_{i}$ & random variable $\sim \Bernoulli(q)$, defined in Lemma~\ref{le:CBH} \\

$\alpha$ & subgradient of $\psi_{q}$, defined in \eqref{eq:subgrad} \\

$\beta_{c,k,m}(\omega)$ & $
\frac{\frac{1}{q}\,c^{q}(1\!+\!m)}{c(k)}\omega+\frac{m \, c^{q}(m)}{k\, c(k)}(1-\omega)$, defined in \eqref{eq:affine-functions}\\

$\gamma(\mathcal{C})$ & bound for the \acl{PoA} for the class $\mathcal{C}$, defined in \eqref{eq:PoA_Bnd} \\

$\gamma_{\mathsf{pr}}(\mathcal{C},q)$ & bound for the \acl{PPoA}, defined in \eqref{eq:BndPr} \\

$\Gamma$ & game \\

$\Gamma^{\boldsymbol{p}}$ & \acl{BCG} \\

$\zeta$ & integer such that $\tilde{q}=\zeta/\nu\approx q$, defined in \eqref{eq:rational-approx}\\

$\eta$ & constant in $[0,1]$ defined in \eqref{eq:peq} \\
$\theta$ & integer such that $\tilde\eta=\theta/\tau\approx\eta $, defined in \eqref{eq:rational-approx} \\

$\kappa$ & $n\, \nu\, k$ \\

$\tilde\kappa$ & $\kappa+mj$ \\
$\lambda$ & parameter of $(\lambda,\mu)$-smoothness, defined in \eqref{eq:def-lambda-mu} \\
$\mu$ & parameter of $(\lambda,\mu)$-smoothness, defined in \eqref{eq:def-lambda-mu} \\

$\nu$ & integer such that $\tilde{q}=\zeta/\nu\approx q$, defined in \eqref{eq:rational-approx}\\

$\xi_{\ell}(q)$ & $\frac{\ell(\ell+1)q^2+2\ell q+1}{2\ell q}$, defined in \eqref{eq:Xi-xi} \\
$\Xi(q)$ & $\inf_{\ell\geq 1}\xi_{\ell}(q)$, defined in \eqref{eq:Xi-xi} \\

$\tau$ & integer such that $\tilde\eta=\theta/\tau\approx\eta $, defined in \eqref{eq:rational-approx} \\

$\Phi$ & potential, defined in \eqref{eq:potential} \\

$\psi_{q}$ & defined in \eqref{eq:vrphi} \\

$\psi^{k,m}_{q}$ & defined in \eqref{eq:vrphikm} \\

$\psi_{\infty}$ & defined in \eqref{eq:varphi-infty} \\

$\Psi(\omega)$ & defined in Lemma~\ref{le:boundalt} \\

$\Psi_{\mathcal{C}',n}(\argdot)$ & defined in \eqref{eq:Psi-finite} \\

$\omega$ & $1/(1-\mu)$ \\

$\stirling{d}{j}$ & Stirling number of the second kind\\
\end{longtable}

\end{APPENDICES}

\end{document}